\documentclass[12pt]{article}

\usepackage[a4paper, total={6.5in, 9in}]{geometry}
\usepackage{graphicx} 
\usepackage{braket}
\usepackage{amsmath,amssymb,amstext}
\usepackage{amsthm}
\usepackage{dsfont}

\newtheorem{lem}{Lemma}
\newtheorem{thm}{Theorem}
\newtheorem{cor}{Corollary}
\newtheorem{rem}{Remark}

\title{Second order energy expansion of  Bose gases with three-body interactions}
\author{Morris Brooks}

\date{July 2024}

\begin{document}
\maketitle

\begin{abstract}
We provide a second order energy expansion for a gas of $N$ bosonic particles with three-body interactions in the Gross-Pitaevskii regime. We especially confirm a conjecture by Nam, Ricaud and Triay in \cite{NRT3}, where they predict the subleading term in the asymptotic expansion of the ground state energy to be of the order $\sqrt{N}$. In addition, we show that the ground state satisfies Bose-Einstein condensation with a rate of the order $\frac{1}{\sqrt{N}}$. 
\end{abstract}

\section{Introduction}
In this manuscript we study a dilute Bose gas consisting of $N$ quantum particles subject to Bose-Einstein statistics, in which the individual particles interact with each other via a three-body potential 
\begin{align}
\label{Eq:Potential_Rescaled}
    V_N(x,y,z):=N V\!\left(\sqrt{N}(x-y),\sqrt{N}(x-z)\right),
\end{align}
defined in terms of a given bounded and non-negative function $V:\mathbb{R}^{3}\times \mathbb{R}^3\longrightarrow \mathbb{R}$ with compact support. The quantum gas is then described by the self-adjoint operator
\begin{align}
\label{Eq:Hamilton_Operator}
    H_N:=-\sum_{1\leq k\leq N} \Delta_{x_k}+\sum_{1\leq i<j<k\leq N}V_N(x_i,x_j,x_k),
\end{align}
acting on the space of permutation symmetric functions $L^2_{\mathrm{sym}}\! \left(\Lambda^N\right)$, where $\Lambda:=\left[-\frac{1}{2},\frac{1}{2}\right]^3$ is the three dimensional torus and we assume that $V_N$ defined in Eq.~(\ref{Eq:Potential_Rescaled}) is permutation symmetric in order to assure that $H_N$ preserves permutation symmetry. The particular scaling in Eq.~(\ref{Eq:Potential_Rescaled}) with the number of particles $N$ is referred to as the Gross-Pitaevskii regime and yields a short range, but strong, interaction on the scale $r=\frac{1}{\sqrt{N}}$. This especially means that we are dealing with a dilute gas taking up a volume of the order $Nr^3=\frac{1}{\sqrt{N}}$. Dilute Bose gases with three-particle interactions have previously been studied in \cite{NRT1,NRT2,NRT3,V,JV}, where the leading order asymptotics of the ground state energy in the limit $N\rightarrow \infty$ has been established as well as Bose-Einstein condensation (BEC) in the Gross-Pitaevskii regime. Here (BEC) refers to the observation that most of the particles occupy the state with zero momentum. In the Gross-Pitaevskii regime specifically, the leading order term in the asymptotics of the ground state energy
\begin{align}
\label{Eq:First_Order_Energy_Asymptotics}
    E_N:=\inf \sigma(H_N)=\frac{1}{6}b_{\mathcal{M}}(V)N+o_{N\rightarrow \infty}\! \left(N\right)
\end{align}
is proportional to the number of particles $N$ with a rather explicit constant $b_{\mathcal{M}}(V)$, see \cite{NRT1}. While naive first order perturbation theory would suggest that the constant $b_{\mathcal{M}}(V)$ should be given by $\widehat{V}(0)$, it is due to the singular nature of the scaling in Eq.~(\ref{Eq:Potential_Rescaled}) that we cannot ignore the presence of three particle correlations leading to a renormalized constant $b_{\mathcal{M}}(V)<\widehat{V}(0)$. In the following we will address a conjecture in \cite{NRT3}, which claims that the subleading term in the asymptotic expansion of $E_N$ is proportional to $\sqrt{N}$, see our main Theorem \ref{Eq:Main_Theorem_Introduction}. The contributions to the ground state energy $E_N$ of the order $\sqrt{N}$ arise based on two-particle, three-particle and four-particle correlations in the ground state. As a byproduct from the proof of Theorem \ref{Eq:Main_Theorem_Introduction}, we obtain in addition that the ground state $\Psi^{\mathrm{GS}}_N$ of the operator $H_N$ satisfies (BEC) with a rate $\frac{1}{\sqrt{N}}$, i.e. we show that the ratio of particles outside the state with zero momentum compared to the total number of particles $N$ is of the order $O_{N\rightarrow \infty}\! \left(\frac{1}{\sqrt{N}}\right)$. This is an improvement of the (BEC) result in \cite{NRT1}, where the authors showed that the ratio is of the order $o_{N\rightarrow \infty}\! \left(1\right)$.

It is worth pointing out that much more is known for Bose gases with two-particle interactions, where the expansion of the ground state energy to second order is well known in the Gross-Pitaevskii regime, the thermodynamic limit and interpolating regimes, see e.g. \cite{BBCS,BCS,B,FS1,FS2,HST,NNRT}. Furthermore, (BEC) is well known for the Gross-Pitaevskii regime and regular enough interpolating regimes, even with an (optimal) rate, see e.g. \cite{ABS,BBCS0,BBCS1,BS,BBCO,F,HHNST,LSe}, and the subleading term in the expansion of the ground state energy is known to be of the order $O_{N\rightarrow \infty}(1)$. This resolution of the energy is sharp enough to see the spectral gap, which is of the order $O_{N\rightarrow \infty}(1)$ as well. For a Bose gas with three-particle interaction in the Gross-Pitaevskii regime we expect the spectral gap to be of the magnitude $O_{N\rightarrow \infty}(1)$, see the conjecture in \cite{NRT3}, however the second order expansion of the energy only allows for a resolution of the order $O_{N\rightarrow \infty}\! \left(\sqrt{N}\right)$, which is not sharp enough to see the spectral gap.

As it is not the goal of this manuscript to optimize in the regularity of $V$, we will assume $V\in C^\infty(\mathbb{R}^6)$ for the sake of convenience (although assuming e.g. $V\in H^9(\mathbb{R}^6)$ would certainly be sufficient).\\

The correct constant $b_\mathcal{M}(V)$ in the energy asymptotics Eq.~(\ref{Eq:First_Order_Energy_Asymptotics}) can be derived formally by making a translation-invariant ansatz for the correlation structure $\varphi(x-u,y-u)$ between three particles at position $x$, $y$ and $u$, where $\varphi:\mathbb{R}^6\longrightarrow \mathbb{R}$. Utilizing the matrix 
\begin{align*}
   \mathcal{M}:=\sqrt{\frac{1}{2}\begin{bmatrix}
2 & 1 \\
1 & 2 
\end{bmatrix}}  
\end{align*}
and the modified Laplace operator $\Delta_{\mathcal{M}}:=\left(\mathcal{M}\nabla_{\mathbb{R}^{3}\times \mathbb{R}^{3}}\right)^2$, let us first express the action of the Laplace operator in relative coordinates as
\begin{align*}
    \left(\Delta_x+\Delta_y+\Delta_u\right)\varphi(x-u,y-u)=(2\Delta_{\mathcal{M}}\varphi)(x-u,y-u).
\end{align*}
The energy correction associated to the correlation structure $\varphi$ is then given by
\begin{align*}
    \braket{\varphi,(-2\Delta_{\mathcal{M}}+V)\varphi}-2\braket{\varphi,V}=\int_{\mathbb{R}^6}\left\{2\big|\mathcal{M}\nabla \varphi(x)\big|^2+V(x)\varphi(x)^2-2V(x)\varphi(x) \! \right\}\mathrm{d}x,
\end{align*}
 where $\braket{\varphi,V}$ describes the interaction between the condensate wavefunction $\Psi\equiv 1$ and $\varphi$. Including the potential energy $\braket{\Psi,V\Psi}=\int_{\mathbb{R}^6}V(x)\mathrm{d}x$ and optimizing in $\varphi$ leads to
    \begin{align}
\label{Eq:Definition_of_b}
    b_{\mathcal{M}}(V):=\inf_{\varphi\in \dot{H}^1(\mathbb{R}^6)} \int_{\mathbb{R}^6}\left\{2\big|\mathcal{M}\nabla \varphi(x)\big|^2+V(x)\! \left|1-\varphi(x)\right|^2 \! \right\}\mathrm{d}x,
\end{align}
where $\dot{H}^1(\mathbb{R}^d)$ refers to the space of functions $g:\mathbb{R}^d\longrightarrow \mathbb{C}$ vanishing at infinity with $|\nabla g|\in L^2\! \left(\mathbb{R}^d\right)$. It has been verified in \cite{NRT1} that minimizer $\omega$ to the variational problem in Eq.~(\ref{Eq:Definition_of_b}) exist and the (modified) scattering length $b_{\mathcal{M}}(V)$ describes the leading order asymptotics of the ground state energy correctly, see Eq.~(\ref{Eq:First_Order_Energy_Asymptotics}). Our main Theorem \ref{Eq:Main_Theorem_Introduction} confirms that the next term in the energy asymptotics in Eq.~(\ref{Eq:First_Order_Energy_Asymptotics}) is of the order $O_{N\rightarrow \infty}\! \left(\sqrt{N}\right)$ due to contributions from the three-particle correlation $\omega$, as well as from two-particle and four-particle correlations. Regarding the impact of the correlation structure $\omega$ on the $\sqrt{N}$ order, we observe in the presence of an additional particle at position $z$ further interaction terms between $\omega$ and itself of the form
 \begin{align*}
     \braket{\omega(y,z),V(x,y)\omega(y,z)} & =\int_{\mathbb{R}^9}V(x,y)\omega(y,z)^2\mathrm{d}x\mathrm{d}y\mathrm{d}z,\\
        \braket{\omega(x,z),V(x,y)\omega(y,z)} & =\int_{\mathbb{R}^9}V(x,y)\omega(x,z)\omega(y,z)\mathrm{d}x\mathrm{d}y\mathrm{d}z.
 \end{align*}
The corresponding energy correction is then given by $\gamma(V)\sqrt{N}$ with
 \begin{align}
    \label{Eq:Definition_of_gamma}
     \gamma(V): & = \int_{\mathbb{R}^9}V(x,y)\omega(x,z)\omega(y,z)  \mathrm{d}x \mathrm{d}y \mathrm{d}z+\frac{1}{2}\int_{\mathbb{R}^9}V(x,y)  \omega(y,z)^2  \mathrm{d}x \mathrm{d}y \mathrm{d}z.
 \end{align}
 In order to quantify the impact of two-particle correlations, respectively four-particle correlations, we make a translation-invariant ansatz $\xi(x \! - \! u)$ with $\xi:\mathbb{R}^3\longrightarrow \mathbb{R}$, respectively $\eta(x \! - \! u,y \! - \! u,z \! - \! u)$ with $\eta:\mathbb{R}^9\longrightarrow \mathbb{R}$. We observe that the interaction energy of $\xi$ and $\eta$ with the condensate and the three-particle correlation structure $\omega$ is given by
 \begin{align}
 \label{Eq:Correlation_Four_particle_with_omega_II}
  \int_{\mathbb{R}^3} \Big\langle \xi(x) ,V(x ,y)(1-\omega(x ,y))\Big\rangle \mathrm{d}y & =\int_{\mathbb{R}^6}V(x,y)(1-\omega(x,y))\xi(x)\mathrm{d}y\mathrm{d}x,\\  
  \label{Eq:Correlation_Four_particle_with_omega}
  \Big\langle \eta(x,y,z), V(x,y)\omega(y,z)\Big\rangle & =\int_{\mathbb{R}^9}V(x,y)\omega(y,z)\eta(x,y,z)\mathrm{d}x\mathrm{d}y\mathrm{d}z.
 \end{align}
 Including the kinetic energy of $\xi$ and optimizing in $\xi$, Eq.~(\ref{Eq:Correlation_Four_particle_with_omega_II}) immediately gives rise to the energy correction $-\mu(V)\sqrt{N}$ with the proportionality constant $\mu(V)$ defined as
 \begin{align}
 \nonumber
     \mu(V): & =\inf_{\xi\in \dot{H}^1(\mathbb{R}^3)}\left\{\int_{\mathbb{R}^3}|\nabla \xi(x)|^2\mathrm{d}x-\int_{\mathbb{R}^6}V(x,y)(1-\omega(x,y))\xi(x)\mathrm{d}y\mathrm{d}x\right\}\\
           \label{Eq:Definition_of_mu}
      & =\int_{\mathbb{R}^6}\frac{V_\mathrm{eff}(x)V_\mathrm{eff}(y)}{8\pi |x-y|}\mathrm{d}x\mathrm{d}y,
 \end{align}
 where we have introduced the effective two-particle interaction 
 \begin{align*}
    V_\mathrm{eff}:
\begin{cases}
    & \mathbb{R}^3\longrightarrow \mathbb{R},\\
     & x\mapsto \int_{\mathbb{R}^3}V(x,y)(1-\omega(x,y))\, \mathrm{d}y.
\end{cases}
\end{align*}
Finally, in order to identify the energy contribution due to the presence of four-particle correlations $\eta$, we have to express the kinetic and potential energy of $\eta(x-u,y-u,z-u)$ in terms of relative coordinates. For this purpose let us define $\mathbb{V}:\mathbb{R}^9\longrightarrow \mathbb{R}$ and $\mathcal{M}_*$ as
\begin{align*}
  \mathbb{V}(x_1,x_2,x_3): & =V(x_1,x_2)+V(x_1,x_3)+V(x_2,x_3)+V(x_2-x_1,x_3-x_1),\\
       \mathcal{M}_*: & =\sqrt{\frac{1}{2}\begin{bmatrix}
2 & 1 & 1\\
1 & 2 & 1\\
1 & 1 & 2
\end{bmatrix}},
\end{align*}
and identify the action of the Laplace operator and the potential on $\eta$ as
\begin{align}
\label{Eq:Correlation_Four_particle_with_Delta}
 \left(\Delta_{x_1} \!  \! + \! \Delta_{x_2} \!  \! + \! \Delta_{x_3} \!  \! + \! \Delta_{x_4}\right)\eta(x_1 \! - \! x_4,x_2 \! - \! x_4,x_3 \! - \! x_4) &   \! = \! (2\Delta_{\mathcal{M}_*}\eta)(x_1 \! - \! x_4,x_2 \! - \! x_4,x_3 \! - \! x_4),\\
 \label{Eq:Correlation_Four_particle_with_V}
  \sum_{1\leq i<j<k\leq 4} \! \! \! V(x_i \! - \! x_k,x_j \! - \! x_k)\eta(x_1 \! - \! x_4,x_2 \! - \! x_4,x_3 \! - \! x_4) & =(\mathbb{V}\eta)(x_1 \! - \! x_4,x_2 \! - \! x_4,x_3 \! - \! x_4).  
\end{align}
Combining kinetic, potential and interaction energy, see  Eq.~(\ref{Eq:Correlation_Four_particle_with_omega}), Eq.~(\ref{Eq:Correlation_Four_particle_with_Delta}) and Eq.~(\ref{Eq:Correlation_Four_particle_with_V}), yields
\begin{align*}
    \int_{\mathbb{R}^9}\left\{2\big|\mathcal{M}_*\nabla \eta(x)\big|^2+\mathbb{V}(x)\eta(x)^2-2f(x)\eta(x) \! \right\}\mathrm{d}x=\mathcal{Q}(\eta)-\mathcal{Q}(0),
\end{align*}
where we have introduced the function $f(x,y,z):=V(x,y)\omega(y,z)$ and the functional
\begin{align}
\label{Eq:Q-Functional}
\mathcal{Q}(\eta):=  \int_{\mathbb{R}^9}\left\{2\big|\mathcal{M}_*\nabla \eta(x)\big|^2+\mathbb{V}(x)\! \left|\frac{f(x)}{\mathbb{V}(x)}-\eta(x)\right|^2 \! \right\}\mathrm{d}x.
\end{align}
Note that $\frac{f}{\mathbb{V}}$ is well defined and bounded on the support of $\mathbb{V}$, due to the sign of $V\geq 0$. Consequently the corresponding energy correction is given by $-\sigma(V)\sqrt{N}$ with
\begin{align}
      \label{Eq:Definition_of_sigma}
    \sigma(V) \! := \mathcal{Q}(0)-  \inf_{\eta\in \dot{H}^1(\mathbb{R}^9)} \!  \mathcal{Q}(\eta).
\end{align}
It is content of our main Theorem \ref{Eq:Main_Theorem_Introduction}, that $\gamma(V),\mu(V)$ and $\sigma(V)$ describe the second order correction to the leading order asymptotics of the ground state energy $E_N$ in Eq.~(\ref{Eq:First_Order_Energy_Asymptotics}), which is of the order $O_{N\rightarrow \infty}\! \left(\sqrt{N}\right)$. Furthermore, we show that the order $O_{N\rightarrow \infty}\! \left(\sqrt{N}\right)$ term comes with a non-zero pre-factor for a large class of potentials $V$.
\begin{thm}
\label{Eq:Main_Theorem_Introduction}
   Let $V\in C^\infty \! \left(\mathbb{R}^6\right)$ be a bounded and non-negative function with compact support, such that the function $V_N$ defined in Eq.~(\ref{Eq:Potential_Rescaled}) is permutation symmetric. Furthermore, let $\gamma(V),\mu(V)$ and $\sigma(V)\in \mathbb{R}$ be as in Eq.~(\ref{Eq:Definition_of_gamma}), Eq.~(\ref{Eq:Definition_of_mu}) and Eq.~(\ref{Eq:Definition_of_sigma}) respectively, and let $b_{\mathcal{M}}(V)$ be as in Eq.~(\ref{Eq:Definition_of_b}). Then the ground state energy $E_N:=\inf\sigma(H_N)$ satisfies 
    \begin{align}
    \label{Eq:Main_Theorem_Introduction_Equation}
        E_N=\frac{1}{6}b_{\mathcal{M}}(V)N+\Big( \! \gamma(V) - \mu(V) -  \sigma(V)\! \Big) \sqrt{N} +O_{N\rightarrow \infty}\! \left(N^{\frac{1}{4}}\right).
    \end{align}
    Furthermore, there exists a $\lambda(V)>0$, such that for all $0<\lambda\leq \lambda(V)$
    \begin{align*}
        \gamma(\lambda V) - \mu(\lambda V)-  \sigma(\lambda V)< 0.
    \end{align*}
\end{thm}

\begin{rem}
    While Theorem \ref{Eq:Main_Theorem_Introduction} concerns Bose gases in the ultra-dilute Gross-Pitaevskii regime occupying a volume of the order $\frac{1}{\sqrt{N}}$, the leading order behaviour of the ground state energy per unit volume $e(\rho)$ is known in the thermodynamic regime as well as a function of the density $\rho$, see \cite{NRT2}, and given in analogy to the leading order asymptotics in Eq.~(\ref{Eq:First_Order_Energy_Asymptotics}) by
    \begin{align*}
        e(\rho)=\frac{1}{6}b_{\mathcal{M}}(V)\rho^3+o_{\rho\rightarrow 0}\! \left(\rho^3\right).
    \end{align*}
  It is remarkable that the coefficients $ \gamma(V)$, $\mu(V) $ and $\sigma(V)$ from Theorem \ref{Eq:Main_Theorem_Introduction} are defined in terms of variational problems on the unconfined space $\mathbb{R}^{3d}$ and do not depend on the boundary conditions of the box $\Lambda^d$. Substituting $\rho$ with $\frac{1}{\sqrt{N}}$ in Theorem \ref{Eq:Main_Theorem_Introduction} we therefore expect the second order expansion of $e(\rho)$, as $\rho\rightarrow 0$, to be given by
    \begin{align*}
        e(\rho)=\rho^3\left(\frac{1}{6}b_{\mathcal{M}}(V)+\Big( \! \gamma(V) - \mu(V) -  \sigma(V)\! \Big)\rho \right)+o_{\rho\rightarrow 0}\! \left(\rho^4\right).
    \end{align*}
    This would be in contrast with the second order expansion of a Bose gas with two-body interactions, where in the celebrated Lee-Huang-Yang formula the dual group $\left(2\pi \mathbb{Z}\right)^3$ of the box $\Lambda$ appears. It is however expected that there is a corresponding Lee-Huang-Yang term for gases with three-body interactions, which should appear in a third order expansion of the energy as a term of the order $O_{N\rightarrow \infty}(1)$. 
\end{rem}
\textbf{Proof strategy of Theorem \ref{Eq:Main_Theorem_Introduction}.} Following the ideas in \cite{NRT1}, respectively \cite{BBCS0,BBCS,FS1,FS2}, which have been developed in the context of Bose gases with two-body interactions, we are going to unveil the correlation structure of the ground state with the help of a suitable coordinate transformation. Based on the strategy presented in \cite{B}, our initial coordinate transformation will be of algebraic nature, i.e. we introduce a new set of operators and observe that the many-body operator $H_N$ is almost diagonal in these new variables. The algebraic approach immediately allows us to find satisfactory lower bounds on the ground state energy $E_N$. Furthermore, we show that this coordinate transformation can be realized in terms of a unitary map, at least in an approximate sense, which yields the corresponding upper bound on $E_N$.

In order to find a suitable transformation bringing $H_N$ in a diagonal form, we observe that collisions between at most three particles will occur much more frequent compared to collisions between four or more particles, as we are in the dilute regime where the gas occupies only a volume of the magnitude $\frac{1}{\sqrt{N}}$. Consequently, we first look for a diagonalization of a gas with only three particles $N=3$, which will involve the three-particle correlation structure $\omega$, and subsequently lift it to a diagonalization of the full many-body problem. As it turns out, including the three-particle correlation structure is enough to identify the leading order behaviour of the ground state energy. To be more precise, we are able to show at this point 
\begin{align}
\label{Eq:Leading_Order_With_Optimal_Rate}
   E_N=\frac{1}{6}b_{\mathcal{M}}(V)N+O_{N\rightarrow \infty}\! \left(\sqrt{N}\right).
\end{align}
We want to emphasize that the proof of Eq.~(\ref{Eq:Leading_Order_With_Optimal_Rate}) depends on our ability to neglect collisions between four or more particles, and we note that the correlation structure involves mostly particles outside the state with zero-momentum. It is therefore crucial to have strong a priori information regarding the number of particles outside the state with zero momentum, which we will refer to as excited particles. In the language of second quantization, the number of excited particles can naturally be expressed as
\begin{align}
\label{Eq:Def_intro_Particle_Number}
    \mathcal{N}:=N-\mathcal{N}_0:=N-a_0^\dagger a_0,
\end{align}
where $N$ is the total number of particles, $\mathcal{N}_0$ counts the number of particles with zero momentum and $a_0$ is the annihilation operator corresponding to the zero momentum state, see also Section \ref{Sec:Condensation with a Rate} for a more comprehensive introduction. The following result, which has been verified in \cite{NRT1}, tells us that the number of excited particles is indeed small compared to the total number of particles $N$, i.e. 
\begin{align}
\label{Eq:Condensation_No_Rate}
    \frac{1}{N}\braket{\Psi_N^{\mathrm{GS}},\mathcal{N} \Psi_N^{\mathrm{GS}}}=o_{N\rightarrow \infty}(1),
\end{align}
where $\Psi_N^{\mathrm{GS}}$ is the ground state of the operator $H_N$. Using the a priori information in Eq.~(\ref{Eq:Condensation_No_Rate}) then allows us to identify the leading order asymptotics of the ground state energy $E_N$ in Eq.~(\ref{Eq:Leading_Order_With_Optimal_Rate}). In addition, we obtain at this point an improved (BEC) result 
\begin{align}
\label{Eq:Condensation_Yes_Rate}
    \frac{1}{N}\braket{\Psi_N^{\mathrm{GS}},\mathcal{N} \Psi_N^{\mathrm{GS}}}=O_{N\rightarrow \infty} \! \left(\frac{1}{\sqrt{N}}\right)
\end{align}
with a rate of the order $O_{N\rightarrow \infty} \! \left(\frac{1}{\sqrt{N}}\right)$, see also the subsequent Theorem \ref{Eq:Main_Theorem_Introduction_Condensation}, which we believe to be of independent interest. 

Finally, we use an additional coordinate transformation, which implements the two-particle and four-particle correlation structure $\xi$ and $\eta$, together with the improved control on the number of excited particles in Eq.~(\ref{Eq:Condensation_Yes_Rate}), in order to identify the coefficient $C$ in front of the $\sqrt{N}$ term in the energy asymptotics
\begin{align*}
    C=\gamma(V) - \mu(V) -  \sigma(V).
\end{align*}
Notably, collisions between four particles do contribute to the subleading term in the energy expansion in Eq.~(\ref{Eq:Main_Theorem_Introduction_Equation}), however in analogy to Eq.~(\ref{Eq:Leading_Order_With_Optimal_Rate}) we can dismiss collisions between five or more particles.

\begin{thm}
\label{Eq:Main_Theorem_Introduction_Condensation}
   Let $V$ satisfy the assumptions of Theorem \ref{Eq:Main_Theorem_Introduction} and let $\Psi_N^{\mathrm{GS}}$ denote the ground state of the operator $H_N$. Furthermore let $\mathcal{N}$ be the operator counting the number of excitations introduced in Eq.~(\ref{Eq:Def_intro_Particle_Number}). Then there exists a constant $C>0$, such that
   \begin{align*}
       \frac{1}{N}\braket{\Psi_N^{\mathrm{GS}},\mathcal{N} \Psi_N^{\mathrm{GS}}}\leq \frac{C}{\sqrt{N}}.
   \end{align*}
\end{thm}

\textbf{Outline.} In Subsection \ref{Subsec:The three-body Problem} we are first deriving the three-particle correlation structure for a model where the total number of particles is $N=3$. Following the strategy proposed in \cite{B}, we are implementing in a systematic way the correlation structures from Subsection \ref{Subsec:The three-body Problem} for gases with many particles $N\gg 1$ in Section \ref{Sec:Condensation with a Rate}. Using Bose-Einstein condensation of the ground state as an input, this allows us to immediately recover the leading order behaviour of $E_N$ as a lower bound and, in the subsequent Section \ref{Sec:First_Order_Upper_Bound}, also as an upper bound. Furthermore, we obtain at this point an improved version of (BEC) with a rate, which especially concludes the proof of Theorem \ref{Eq:Main_Theorem_Introduction_Condensation}. In Section \ref{Section:Second_Order_Lower_Bound}, we are going to describe the two-particle and four-particle correlation structure, which gives rise to the correction $\mu(V)$ and the correction $\sigma(V)$ defined in Eq.~(\ref{Eq:Definition_of_mu}) and Eq.~(\ref{Eq:Definition_of_sigma}) respectively. It is then the purpose of Subsection \ref{Subsection:Proof_of_Theorem_Main_in_Text} to verify the lower bound in our main Theorem \ref{Eq:Main_Theorem_Introduction}, wherein we use the improved (BEC) result, and the purpose of Section \ref{Sec:Second_Order_Upper_Bound} to verify the corresponding upper bound. The sign of $\gamma(\lambda V) - \mu(\lambda V)-  \sigma(\lambda V)$ is established in Section \ref{Sec:Analysis of the Scattering Coefficients} for small $\lambda>0$, along side other useful properties of the scattering solutions that describe the correlation structure. Finally Appendix \ref{Sec:Auxiliary Results} contains a collection of operator inequalities, and Appendix \ref{Sec:IMS Localisation} introduces a localization formula that allows us to improve (BEC) results, at the cost of modifying the state under consideration.

\subsection{The three-body Problem}
\label{Subsec:The three-body Problem}
While one would naively expect that the leading order behaviour of the ground state energy $E_N$ is dictated by the energy of the uncorrelated wave function $\Gamma_0(x_1,\dots ,x_N):=1$
\begin{align*}
    \braket{\Gamma_0,H_N\Gamma_0} \! = \! \frac{N(N \! - \! 1)(N \! - \! 2)}{6N^2}\int_{\mathbb{R}^3}\int_{\mathbb{R}^3}V(x,y)\, \mathrm{d}x\mathrm{d}y \! = \! \frac{N}{6}\int_{\mathbb{R}^3}\int_{\mathbb{R}^3}V(x,y)\, \mathrm{d}x\mathrm{d}y \! + \! o_{N\rightarrow \infty}(N),
\end{align*}
it is due to the presence of correlations in the ground state of the operator $H_N$, that the leading order coefficient $b_{\mathcal{M}}(V)$ in the energy asymptotics of $E_N$ in Eq.~(\ref{Eq:First_Order_Energy_Asymptotics}) satisfies 
\begin{align*}
    b_{\mathcal{M}}(V)<\int_{\mathbb{R}^3}\int_{\mathbb{R}^3}V(x,y)\, \mathrm{d}x\mathrm{d}y.
\end{align*}
In order to quantify the correlation energy $\int_{\mathbb{R}^3}\int_{\mathbb{R}^3}V(x,y)\, \mathrm{d}x\mathrm{d}y-b_{\mathcal{M}}(V)$, we are going to follow the frame work developed in \cite{B}, and investigate first the corresponding three-particle operator $H_{(3)}:=-\Delta_{3}+V_N$ acting on $L^2(\Lambda^3)$, where $\Delta_3:=\Delta_{x_1}+\Delta_{x_2}+\Delta_{x_3}$, before we study the many particle operator $H_N$ defined in Eq.~(\ref{Eq:Hamilton_Operator}). It will be our goal to find a transformation
\begin{align*}
T:L^2(\Lambda^3)\longrightarrow L^2(\Lambda^3)   
\end{align*}
that removes correlations between states with low momenta and states with high momenta, i.e. we want to bring $H_{(3)}$ in a block-diagonal form, which allows us to extract the correlation energy. In is content of Section \ref{Sec:Condensation with a Rate} to lift the block-diagonalization from the three-particle problem, described by the transformation $T$, to a block-diagonalization of the many-particle operator $H_N$, which will allow us to identify the correlation energy for the many-particle problem. \\

Let us first specify the set of low momenta as either the set where all three particles occupy the zero momentum state
\begin{align}
\label{Eq:Def_Low_Momentum_Set_0}
   \mathcal{L}_0:=\{(0,0,0)\}\subseteq (2\pi \mathbb{Z})^9 
\end{align}
or the set where at most one of the three particles is allowed to have non-zero momentum
\begin{align*}
 \mathcal{L}_K:=\bigcup_{|k|\leq K}\{(k,0,0),(0,k,0),(0,0,k)\}\subseteq (2\pi \mathbb{Z})^9 ,   
\end{align*}
where $0\leq K<\infty$ is a parameter that we are going to specify later. For the purpose of extracting the correlation energy, it is enough to consider $K:=0$, however for technical reasons it is going to be convenient later to consider positive values $K>0$ as well. Having the set $\mathcal{L}_K$ at hand, we can define the projection $\pi_K$ onto states with low momenta as
\begin{align}
\label{Eq:Definition_pi_K}
    \pi_K(\Psi):=\sum_{(k_1 k_2 k_3)\in \mathcal{L}_K}\braket{u_{k_1}u_{k_2}u_{k_3},\Psi}u_{k_1}u_{k_2}u_{k_3},
\end{align}
where $u_k(x):=e^{ikx}$ for $k\in (2\pi \mathbb{Z})^3$ and $u_{k_1}u_{k_2}u_{k_3}$ has to be understood as the function $u_{k_1}(x_1)u_{k_2}(x_2)u_{k_3}(x_3)$. Let us furthermore introduce the projection $Q$ acting on $L^2(\Lambda)$ as
\begin{align*}
    Q(\phi):=\sum_{k\neq 0}\braket{u_k,\phi}u_k.
\end{align*}
Following the strategy in \cite{B}, let $R$ be the pseudo-inverse of the operator $Q^{\otimes 3}(-\Delta_{3}+V_N)Q^{\otimes 3}$ and let us define the Feshbach-Schur like transformation
\begin{align}
\label{Eq:Definition_Feshbach-Schur}
    T:=1+RV_N \pi_K.
\end{align}
Note that $T$ would be a proper Feshbach-Schur map, in case we would exchange $Q^{\otimes 3}$ with the projection $1-\pi_K$, however we prefer to work with $Q^{\otimes 3}$ for technical reasons. Conjugating $H_{(3)}$ with $T^{-1}=1-RV_N \pi_K$ then yields the approximate block-diagonalization
\begin{align*}
  &  (T^{-1})^\dagger H_{(3)}T^{-1}=-\Delta_3+ \pi_K \left(V_N -V_N R V_N\right)\pi_K + \!  (1-\pi_K) V_N  (1-\pi_K) \\
    & \ \ \ \  \ \ \ \ \ \  \ \ \ \ \ \  \ \ \ \  \  \ \ \ \ +  \left\{\pi_K \left(V_N -V_N R V_N \right) \left(1-\pi_K- Q^{\otimes 3} \right)+\mathrm{H.c.}\right\},
\end{align*}
where we introduced the notation $\{A+\mathrm{H.c.}\}:=A+A^*$ and made use of the identity
\begin{align*}
    H_{(3)}T^{-1} & =H_{(3)}-H_{(3)}R V_N \pi_K=H_{(3)}-Q^{\otimes 3}H_{(3)}RV_N \pi_K-(1-Q^{\otimes 3})H_{(3)}R V_N \pi_K\\
    & =H_{(3)}-Q^{\otimes 3}V_N \pi_K-(1-Q^{\otimes 3})V_NR V_N \pi_K.
\end{align*}
Defining the almost block-diagonal renormalized potential $\widetilde{V}_N$ as
\begin{align}
\nonumber
   &   \widetilde{V}_N \! := \pi_K \left(V_N -V_N R V_N\right)\pi_K + \!  (1-\pi_K) V_N  (1-\pi_K)  \\
    \label{Eq:Definition_tilde_V}
    & \ \ \ \ +  \left\{\left(1-\pi_K- Q^{\otimes 3} \right)\left(V_N -V_N R V_N \right) \pi_K +\mathrm{H.c.}\right\},
\end{align}
we therefore obtain the algebraic identity
\begin{align}
\label{Eq:Block_diagonal_3_particle}
    H_{(3)}=T^\dagger \left(-\Delta_3+\widetilde{V}_N\right)T .
\end{align}
The presence of $ \left\{\left(1-\pi_K- Q^{\otimes 3} \right)\left(V_N -V_N R V_N \right) \pi_K +\mathrm{H.c.}\right\}$ in $\widetilde{V}_N$, which are the terms that violate the block-diagonal structure, is due to the usage of $Q^{\otimes 3}$ instead of $1-\pi_K$, however it turns out that these terms do not contribute to the correlation energy to leading order. One therefore expects to read of the leading order coefficient $b_{\mathcal{M}}(V)$ in the asymptotic expansion of the ground state energy $E_N$ in Eq.~(\ref{Eq:First_Order_Energy_Asymptotics}) from the matrix entries of the renormalized potential
\begin{align*}
    \left(\widetilde{V}_N\right)_{000,000}=\braket{u_0u_0u_0,\widetilde{V}_N u_0u_0u_0}=\braket{u_0u_0u_0,\left(V_N -V_N R V_N \right) u_0u_0u_0}.
\end{align*}
As we are going to verify in Lemma \ref{Lem:Coefficient_Control_II}, we have indeed the asymptotic result
\begin{align*}
   b_{\mathcal{M}}(V)=N^2 \braket{u_0u_0u_0,\left(V_N -V_N R V_N \right) u_0u_0u_0}+ O_{N\rightarrow \infty}\! \left(\frac{1}{N}\right).
\end{align*}

\section{First Order Lower Bound}
\label{Sec:Condensation with a Rate}
It is the goal of this Section to bring, in analogy to Eq.~(\ref{Eq:Block_diagonal_3_particle}), the many-particle operator $H_N$ in an approximate block-diagonal form, which allows us to obtain an asymptotically correct lower bound on the ground state energy $E_N$ in Corollary \ref{Cor:Condensation}. For this purpose let us first rewrite the operator $H_N$ defined in Eq.~(\ref{Eq:Hamilton_Operator}) in the language of second quantization as 
\begin{align}
\label{Eq:Hamilton_Operator_2nd_quantized}
    H_N=\sum_{k\in (2\pi \mathbb{Z})^3}|k|^2 a_k^\dagger a_k+\frac{1}{6}\sum_{i j k,\ell m n\in (2\pi \mathbb{Z})^3}\left(V_N\right)_{ijk,\ell m n}a_i^\dagger a_j^\dagger a_k^\dagger a_{\ell}a_{m}a_{n},
\end{align}
where $a_k:=a(u_k)$ annihilates a particle in the mode $u_k(x):=e^{ik\cdot x}$ and its adjoint $a_k^\dagger$ creates a particle in the mode $u_k$, see e.g. \cite{BL} for an introduction to the standard creation and annihilation operators, and $\left(V_N\right)_{ijk,\ell m n}$ are the matrix elements of $V_N$ w.r.t. to the basis $u_i u_j u_k$ defined below Eq.~(\ref{Eq:Definition_pi_K}). If not indicated otherwise, we will always assume that indices run in the set $(2\pi \mathbb{Z})^3$, which we will usually neglect in our notation, and we write $k\neq 0$ in case the index runs in the set $(2\pi \mathbb{Z})^3\setminus \{0\}$. Note that the operator on the right hand side of Eq.~(\ref{Eq:Hamilton_Operator_2nd_quantized}) is defined on the full Fock space 
\begin{align*}
 \mathcal{F}\! \left(L^2 \! \left(\Lambda\right)\right):=\bigoplus_{n=0}^\infty L^2_{\mathrm{sym}} \! \left(\Lambda^n\right)   ,
\end{align*}
while the left hand side is only defined on $L^2_{\mathrm{sym}} \! \left(\Lambda^N\right)\subseteq \mathcal{F}\! \left(L^2 \! \left(\Lambda\right)\right)$, and therefore Eq.~(\ref{Eq:Hamilton_Operator_2nd_quantized}) has to be understood as being restricted to the subspace $L^2_{\mathrm{sym}} \! \left(\Lambda^N\right)$. Furthermore, we observe that $V_N$ is a translation-invariant multiplication operator, and therefore the matrix elements of $V_N$ satisfy $\left(V_N\right)_{ijk,\ell m n}=0$ in case $i+j+k\neq \ell+m+n$ and otherwise
\begin{align}
\label{Eq:Coefficients_in_Fourier}
    \left(V_N\right)_{ijk,\ell m n}=\left(V_N\right)_{(i-\ell)(j-m)(k-n),0 0 0}=N^{-2}\widehat{V}\! \left(\frac{j-m}{\sqrt{N}},\frac{k-n}{\sqrt{N}}\right).
\end{align}
Following the strategy proposed in \cite{B}, we are going to introduce a many-particle counterpart to the three particle map $T$ defined in Eq.~(\ref{Eq:Definition_Feshbach-Schur}), which is realized by the set of operators
\begin{align}
\label{Eq:Definition_c_variable}
c_k:&=a_k+\frac{1}{2}\sum_{i j,\ell m n}(T-1)_{ijk,\ell m n} \, a_i^\dagger a_j^\dagger a_{\ell}a_{m}a_{n},\\
\label{Eq:Definition_psi}
\psi_{i j k}:&=\sum_{\ell m n}T_{ijk,\ell m n} \,  a_{\ell}a_{m}a_{n}.
\end{align}
Here $T_{ijk,\ell m n}:=\braket{u_i u_j u_k, T u_{\ell} u_{m} u_{n}}$ denotes the matrix elements of $T$. The following Lemma \ref{Lem:Many_Body_Block_Diagonal} is the many-particle counterpart to Eq.~(\ref{Eq:Block_diagonal_3_particle}), in the sense that it provides an (approximate) block-diagonal representation of the operator $H_N$ in terms of the new variables $c_k$ and $\psi_{ijk}$.
\begin{lem}
\label{Lem:Many_Body_Block_Diagonal}
    Let $\widetilde{V}_N$ be the operator defined in Eq.~(\ref{Eq:Definition_tilde_V}). Then we have 
    \begin{align}
       \label{Eq:First_Algebraic_Representation}
    H_N=\sum_{k}|k|^2 & c_k^\dagger c_k \! + \! \frac{1}{6}\sum_{i j k,\ell m n}\left(\widetilde{V}_N\right)_{i j k,\ell m n}\psi_{i j k}^\dagger \psi_{\ell m n}-\mathcal{E},
    \end{align}
    where the residual term $\mathcal{E}$ is defined as
    \begin{align*}
    \mathcal{E}:=\frac{1}{4}&\sum_{i j k, \ell m n; i' j', \ell' m' m'}|k|^2 \overline{(T-1)_{i'j'k,\ell' m' n'}}(T-1)_{ijk,\ell m n}\, a_{\ell'}^\dagger a_{m'}^\dagger a_{n'}^\dagger \\ & \times \left(a_{i'}a_{j'}a_i^\dagger a_j^\dagger-\delta_{ii'}\delta_{jj'}-\delta_{ij'}\delta_{ji'}\right)  a_{\ell}a_{m}a_{n}.
\end{align*}
\end{lem}
\begin{proof}
   Using the permutation symmetry of $T$, we first identify $\sum_{k}|k|^2 (c_k-a_k)^\dagger (c_k-a_k)$ as
   \begin{align*}
    \frac{1}{4}& \sum_{i j k, \ell m n; i' j', \ell' m' m'}|k|^2 \overline{(T-1)_{i'j'k,\ell' m' n'}}(T-1)_{ijk,\ell m n}\, a_{\ell'}^\dagger a_{m'}^\dagger a_{n'}^\dagger a_{i'}a_{j'}a_i^\dagger a_j^\dagger  a_{\ell}a_{m}a_{n}\\
    & =\frac{1}{2}  \!  \! \sum_{i j k,\ell m n}\left\{(T-1)^\dagger (-\Delta_{x_3})(T-1)\right\}_{ijk,\ell m n} a_i^\dagger a_j^\dagger a_k^\dagger a_{\ell}a_{m}a_{n}+\mathcal{E}\\
    &   =\frac{1}{6}  \!  \! \sum_{i j k,\ell m n}\left\{(T-1)^\dagger (-\Delta_3)(T-1)\right\}_{ijk,\ell m n} a_i^\dagger a_j^\dagger a_k^\dagger a_{\ell}a_{m}a_{n}+\mathcal{E},
   \end{align*}
where $\Delta_3$ is the Laplace operator on $L^2(\Lambda)^{\otimes 3}$. Similarly
\begin{align*}
   & \sum_{k}|k|^2 a_k^\dagger (c_k-a_k)+\mathrm{H.c.}=\frac{1}{6}  \!  \! \sum_{i j k,\ell m n}\Big\{(-\Delta_3)(T-1)+\mathrm{H.c.}\Big\}_{ijk,\ell m n} a_i^\dagger a_j^\dagger a_k^\dagger a_{\ell}a_{m}a_{n},\\
   & \ \ \ \  \sum_{i j k,\ell m n}\left(\widetilde{V}_N\right)_{i j k,\ell m n}\psi_{i j k}^\dagger \psi_{\ell m n}=\sum_{i j k,\ell m n}\left(T^\dagger \widetilde{V}_N T\right)_{i j k,\ell m n}a^\dagger_i a^\dagger_j a^\dagger_k a_\ell a_m a_n.
\end{align*}
Since
\begin{align*}
    (T \! - \! 1)^\dagger \Delta_3(T \! - \! 1) \! + \! \left\{\Delta_3(T \! - \! 1) \! + \! \mathrm{H.c.}\right\} \! = \! T^\dagger \Delta_3 T \! - \! \Delta_3
\end{align*}
we obtain
\begin{align*}
  &   \ \ \ \ \ \sum_{k}|k|^2 c_k^\dagger c_k \! + \! \frac{1}{6}\sum_{i j k,\ell m n}\left(\widetilde{V}_N\right)_{i j k,\ell m n}\psi_{i j k}^\dagger \psi_{\ell m n}=\\
    & \sum_{k}|k|^2 a_k^\dagger a_k \! +  \!    \frac{1}{6}\! \! \sum_{i j k,\ell m n}\left\{T^\dagger \left(-\Delta_3+\widetilde{V}_N\right)T   \! +  \!  \Delta_3\right\}_{ijk,\ell m n}\! \!  a_i^\dagger a_j^\dagger a_k^\dagger a_{\ell}a_{m}a_{n}  \!  +  \!  \mathcal{E}.
\end{align*}
We observe that $T^\dagger \left(-\Delta_3+\widetilde{V}_N\right)T   \! +  \!  \Delta_3=V_N$ by Eq.~(\ref{Eq:Block_diagonal_3_particle}), which concludes the proof by the representation of $H_N$ in second quantization, see Eq.~(\ref{Eq:Hamilton_Operator_2nd_quantized}).
\end{proof}

Making use of the sign $ (1-\pi_K) V_N (1-\pi_K) \geq 0$, we immediately obtain that
\begin{align*}
    \sum_{i j k,\ell m n}\left((1-\pi_K) V_N(1-\pi_K)\right)_{i j k,\ell m n}\psi_{ijk}^\dagger \psi_{\ell m n}\geq 0.
\end{align*}
Therefore Lemma \ref{Lem:Many_Body_Block_Diagonal} allows us to bound $H_N$ from below by
\begin{align}
\nonumber
     H_N & \geq \sum_{k}|k|^2  c_k^\dagger c_k+\frac{1}{6}\sum_{i j k,\ell m n}\left(\widetilde{V}_N-(1-\pi_K) V_N(1-\pi_K)\right)_{i j k,\ell m n}\psi_{ijk}^\dagger \psi_{\ell m n}-\mathcal{E}\\
     \label{Eq:Estimate_Hamiltonian_a}
     &= \sum_{k}|k|^2  c_k^\dagger c_k+\frac{1}{6}\sum_{i j k,\ell m n}\left(\widetilde{V}_N-(1-\pi_K) V_N(1-\pi_K)\right)_{i j k,\ell m n}a_i^\dagger  a_j^\dagger a_k^\dagger a_\ell a_m a_n-\mathcal{E},
\end{align}
where we have used the fact that $\psi_{ijk}=a_i a_j a_k$ in case one of the indices is zero, which is a direct consequence of the observation that $(T-1)_{ijk,\ell mn}=0$ in case one of the indices in $\{i,j,k\}$ is zero. Note that we can write
\begin{align*}
    \widetilde{V}_N-(1-\pi_K) V_N(1-\pi_K)=A +  B+B^*
\end{align*}
with $A$ and $B$ defined as 
\begin{align*}
    A: & =\pi_K \left(V_N -V_N R V_N\right)\pi_K,\\
    B: & = \left(1-\pi_K- Q^{\otimes 3} \right)\left(V_N -V_N R V_N \right)\pi_K .
\end{align*}
 Let us first analyse the term involving $A$
\begin{align}
    \label{Eq:Identity_Hamiltonian_A}
   \frac{1}{6} \sum_{i j k,\ell m n}\left(A\right)_{i j k,\ell m n}a_i^\dagger  a_j^\dagger a_k^\dagger a_\ell a_m a_n=\lambda_{0,0}\, (a_0^\dagger)^{3}a_0^3+9 a_0^{2\dagger}a_0^2\sum_{0<|k|\leq K}  \lambda_{k,0} a_k^\dagger a_k,
\end{align}
where we define the coefficients 
\begin{align*}
    \lambda_{k,\ell}:=\frac{1}{18}\braket{u_0 u_\ell u_{k-\ell},(V_N -V_N R V_N)(u_0 u_0 u_k+u_0 u_k u_0+u_k u_0 u_0)}.
\end{align*}
 To keep the notation light, we do not explicitly indicate the $N$ dependence of $\lambda_{k,\ell}$. Similarly
\begin{align}
    \label{Eq:Identity_Hamiltonian_B}
   \frac{1}{6} \sum_{i j k,\ell m n}\left(B\right)_{i j k,\ell m n}a_i^\dagger  a_j^\dagger a_k^\dagger a_\ell a_m a_n=3 a_0^\dagger a_0^3\sum_{\ell\neq 0}\lambda_{0,\ell}\,  a_{\ell}^\dagger a_{-\ell}^\dagger+9 a_0^\dagger a_0^2\sum_{\ell,0<|k|\leq K}\lambda_{k,\ell}  a_\ell^\dagger a_{k-\ell}^\dagger a_k.
\end{align}
Putting together Eq.~(\ref{Eq:Estimate_Hamiltonian_a}), Eq.~(\ref{Eq:Identity_Hamiltonian_A}) and Eq.~(\ref{Eq:Identity_Hamiltonian_B}) yields
\begin{align}
\label{Eq:First_Lower_Bound}
   H_N  \geq   \lambda_{0,0}(a_0^\dagger )^3 a_0^3  + \sum_{k}  |k|^2   c_k^\dagger c_k   +  \mathbb{Q}_K  + \mathcal{E}' -  \mathcal{E},
\end{align}
where we define the operator $\mathbb{Q}_K$ and the error term $\mathcal{E}'$ as
\begin{align}
\label{Eq:Definition_V_Quadratic}
    \mathbb{Q}_K: & =  9 a_0^{2\dagger}a_0^2\sum_{0<|k|\leq K}  \lambda_{k,0} a_k^\dagger a_k +3\left(a_0^\dagger a_0^3\sum_{0<|\ell|\leq K}\lambda_{0,\ell}\,  a_{\ell}^\dagger a_{-\ell}^\dagger+\mathrm{H.c.}\right),\\
\label{Eq:Error_Term_With_Primie}
   \mathcal{E}' : & = \bigg(3  \! \sum_{|\ell|>K} \!  \! \lambda_{0,\ell}\,  a_{\ell}^\dagger a_{-\ell}^\dagger\, a_0^\dagger a_0^3  \! + \! 9    \!  \!  \!  \!  \!  \! \sum_{\ell,0<|k|\leq K} \!  \! \lambda_{k,\ell}  a_\ell^\dagger a_{k-\ell}^\dagger a_k\, a_0^\dagger a_0^2\! + \! \mathrm{H.c.}\bigg).
\end{align}

The following Lemma \ref{Lem:Quadratic_Potential_Estimate}, Lemma \ref{Lem:First_Error_Term_Estimate} and Lemma \ref{Lem:Rest_Estimate} will give us sufficient bounds on the various terms appearing in Eq.~(\ref{Eq:First_Lower_Bound}), in order to establish that the ground state energy $E_N$ of $H_N$ is, to leading order, bounded from below by $\frac{1}{6}b_{\mathcal{M}}(V) N$, see Corollary \ref{Cor:Condensation}. In our first Lemma \ref{Lem:Quadratic_Potential_Estimate} we provide a lower bound on 
\begin{align*}
   \lambda_{0,0}(a_0^\dagger)^3 a_0^3+\mathbb{Q}_K 
\end{align*}
for $K$ large enough, which is an operator that is at most quadratic in the variables $a_k$ and $a_k^\dagger$ for $k\neq 0$. In the following let us denote with
\begin{align*}
    \mathcal{N}:=\sum_{k\neq 0}a_k^\dagger a_k
\end{align*}
the operator that counts the number of excited particles, i.e. the number of particles with momentum $k\neq 0$. Since we have the operator identity $\sum_{k}a_k^\dagger a_k=N$ on the Hilbert space $L^2_{\mathrm{sym}}(\Lambda^N)\subseteq \mathcal{F}\! \left(L^2 \! \left(\Lambda\right)\right)$, we observe that $a_0^\dagger a_0=N-\mathcal{N}$, see also Eq.~(\ref{Eq:Def_intro_Particle_Number}), i.e. the number of particles with momentum $k=0$ is given by the difference between the total number of particles $N$ and the number of excited particles.
\begin{lem}
\label{Lem:Quadratic_Potential_Estimate}
    Let $b_{\mathcal{M}}(V)$ be the modified scattering length defined in Eq.~(\ref{Eq:Definition_of_b}). Then there exists for all $\tau,\alpha>0$ constants $C,K>0$ such that
    \begin{align*}
        \lambda_{0,0}(a_0^\dagger)^3 a_0^3+\mathbb{Q}_K\geq \frac{1}{6}b_{\mathcal{M}}(V) N - \alpha \sum_{k}|k|^{2\tau}a_k^\dagger a_k-C\left(1+\frac{\mathcal{N}^2}{N}+\frac{\mathcal{N}}{\sqrt{N}}\right).
    \end{align*}
\end{lem}
\begin{proof}
    First of all we observe that we can write
    \begin{align*}
        (a_0^\dagger)^3 a_0^3 & =N^3-3N^2(\mathcal{N}+3)+N(3\mathcal{N}^2+6\mathcal{N}+2)-\mathcal{N}^3-3\mathcal{N}^2 -2\mathcal{N}\\
        &\geq N^3-3N^2 \mathcal{N}-N^2 C,
    \end{align*}
    for a suitable $C>0$. Defining
    \begin{align*}
        \mathcal{N}_{>}:=\mathcal{N}-\sum_{0<|k|\leq K}a_k^\dagger a_k,
    \end{align*}
we therefore we obtain in combination with Lemma \ref{Lem:Coefficient_Control_II}
\begin{align*}
   & \ \ \ \ \ \ \ \ \ \ \lambda_{0,0}(a_0^\dagger)^3 a_0^3+\mathbb{Q}_K-\frac{1}{6}b_{\mathcal{M}}(V) N\\
   &\geq \frac{1}{6}b_{\mathcal{M}}(V)\left\{9\frac{(a_0^\dagger)^2 a_0^2}{N^2}\sum_{0<|k|\leq K}a_k^\dagger a_k+3\left(\frac{a_0^\dagger a_0^3}{N^2}\sum_{0<|k|\leq K}a_k^\dagger a_{-k}^\dagger +\mathrm{H.c.}\right)-3\mathcal{N}-CN^{-\frac{1}{2}}\mathcal{N}\right\}-C\\
   &=\frac{1}{6}b_{\mathcal{M}}(V)\bigg\{\left(9\frac{(a_0^\dagger)^2 a_0^2}{N^2}-3-CN^{-\frac{1}{2}}\right)\sum_{0<|k|\leq K}a_k^\dagger a_k\\
   &  \ \ \  +3\left(\frac{a_0^\dagger a_0^3}{N^2}\sum_{0<|k|\leq K}a_k^\dagger a_{-k}^\dagger +\mathrm{H.c.}\right)-3\mathcal{N}_{>}\bigg\}-C,
\end{align*}
for a suitable constant $C>0$. Since
\begin{align*}
  3\frac{a_0^\dagger a_0^3}{N^2}\sum_{0<|k|\leq K}a_k^\dagger a_{-k}^\dagger +\mathrm{H.c.}\leq 6\frac{(a_0^\dagger)^2 a_0^2}{N^2}\sum_{0<|k|\leq K}a_k^\dagger a_k,  
\end{align*}
and since we have
\begin{align*}
\left(3\frac{(a_0^\dagger)^2 a_0^2}{N^2}-3-CN^{-\frac{1}{2}}\right)\sum_{0<|k|\leq K}a_k^\dagger a_k\geq -3\frac{\mathcal{N}^2}{N}-C\frac{\mathcal{N}}{\sqrt{N}},
\end{align*}
we obtain 
\begin{align*}
    \lambda_{0,0}(a_0^\dagger)^3 a_0^3+\mathbb{Q}_K\geq \frac{1}{6}b_{\mathcal{M}}(V) N-\frac{1}{2}b_{\mathcal{M}}(V)\mathcal{N}_>-C\left(1+\frac{\mathcal{N}^2}{N}+\frac{\mathcal{N}}{\sqrt{N}}\right)
\end{align*}
for a suitable constant $C$. Using $\mathcal{N}_>\leq K^{-2\tau}\sum_{k}|k|^{2\tau}a_k^\dagger a_k$, this concludes the proof for $K$ large enough.
\end{proof}

In the subsequent Lemma \ref{Lem:First_Error_Term_Estimate}, we provide estimates on the residual term $\mathcal{E}$ defined in Lemma \ref{Lem:Many_Body_Block_Diagonal}, which will allow us to compare the size of $\mathcal{E}$ with the kinetic energy $ \sum_{k}|k|^2 c_k^\dagger c_k$ in the variables $c_k$.

\begin{lem}
\label{Lem:First_Error_Term_Estimate}
For $K\geq 0$, there exists a constant $C_K>0$, such that
\begin{align*}
\pm \mathcal{E}\leq C_K \sum_{k}|k|^2 c_k^\dagger \left(\frac{\mathcal{N}}{N}+N^{-\frac{1}{2}}\right) c_k+\left(\frac{\mathcal{N}}{N}+N^{-\frac{1}{3}}\right)\! (\mathcal{N}+1).
\end{align*}
\end{lem}
\begin{proof}
Let us denote with $\mathcal{I}\subseteq (2\pi \mathbb{Z})^{3\times 3}$ the index set  
\begin{align*}
 \mathcal{I}:=\{(0,0,0)\}\cup \bigcup_{0<|\ell|\leq K}\{(\ell,0,0),(0,\ell,0),(0,0,\ell)\}   
\end{align*}
Then we define for $I=(I_1,I_2,I_3),I'=(I'_1,I'_2,I'_3)\in \mathcal{I}$ the operator $K^{(I,I')}$ acting on $L^2(\Lambda)$ and the operator $G^{(I,I')}$ acting on $L^2(\Lambda)^{\otimes 2}$ as
\begin{align*}
K^{(I,I')}_{i,i'}:&=\frac{1}{2}\sum_{j k }|k|^2\overline{(T-1)_{i'jk,I' }}\left((T-1)_{ijk,I}+(T-1)_{jik,I}\right),\\
  G^{(I,I')}_{ij,i' j'}:&=\frac{1}{4}\sum_{k}|k|^2 \overline{(T-1)_{i'j'k,I' }}(T-1)_{ijk,I},
\end{align*}
as well as $\mathcal{K}_{\tau,2}:=(-\Delta_{x_1})^{\tau}+(-\Delta_{x_2})^{\tau}$ acting on $L^2(\Lambda)^{\otimes 2}$. Then we can write $\mathcal{E}$ as  
\begin{align}
\label{Eq:First_Error_Term}
    \mathcal{E}=\sum_{I,I'\in \mathcal{I}}\left(\sum_{i,i'}K^{(I,I')}_{i,i'} a^\dagger_i \left(a_{I_1}^\dagger a_{I_2}^\dagger a_{I_3}^\dagger a_{I'_1} a_{I'_2} a_{I'_3}\right)a_{i'}+\sum_{ij,i'j'}G^{(I,I')}_{ij,i'j'}a^\dagger_{i}a_j^\dagger \left(a_{I_1}^\dagger a_{I_2}^\dagger a_{I_3}^\dagger a_{I'_1} a_{I'_2} a_{I'_3}\right) a_{i'}a_{j'}\right).
\end{align}
By the weighted Schur test, the operator norm of $\mathcal{K}_{\tau,2}^{-\frac{1}{2}} G^{(I,I')}\mathcal{K}_{\tau,2}^{-\frac{1}{2}}$ is bounded by 
\begin{align*}
  \|\mathcal{K}_{\tau,2}^{-\frac{1}{2}} G^{(I,I')}\mathcal{K}_{\tau,2}^{-\frac{1}{2}}\|\leq \sqrt{\alpha^{(I,I')}\alpha^{(I',I)}},  
\end{align*}
 where we define $\alpha^{(I,I')}:=\sup_{i'j'}\sum_{ij}\frac{|G^{(I,I')}_{ij,i' j'}|}{|i|^{2\tau}+|j|^{2\tau}}$. Let us furthermore introduce $s:=I_1+I_2+I_3$ and $s':=I'_1+I'_2+I'_3$. Making use of Lemma \ref{Lem:Coefficient_Control}, we obtain for the concrete choice $\tau:=\frac{2}{3}$
\begin{align*}
    \alpha^{(I,I')} \! & \leq \!  \sup_{i'j'} \! \sum_{ijk} \! \frac{|k|^2 |(T \! - \! 1)_{i'j'k,I' }|\, |(T \! - \! 1)_{ijk,I}|}{|i|^{2\tau} \! + \! |j|^{2\tau}} \! \lesssim  \!  N^{-4} \sup_{i'j'} \!\sum_{ijk\neq 0}\frac{  \delta_{i'+j'+k=s'} \delta_{i+j+k=s}}{(|i|^{2\tau} \! + \! |j|^{2\tau})(|i|^2 \! + \! |j|^2 \! + \! |k|^2)} \\
 &   \leq \!  N^{-4}  \! \sum_{i\neq 0}\frac{1}{|i|^{2+2\tau}}\lesssim N^{-4}.
\end{align*}
Consequently $\|\mathcal{K}_{\tau,2}^{-\frac{1}{2}} G^{(I,I')}\mathcal{K}_{\tau,2}^{-\frac{1}{2}}\|\lesssim N^{-4}$. Furthermore, the operator 
\begin{align*}
  X^{(I,I')}:=a_{I_1}^\dagger a_{I_2}^\dagger a_{I_3}^\dagger a_{I'_1} a_{I'_2} a_{I'_3}  
\end{align*}
satisfies $\|X^{(I,I')}\|\leq N^3$. Therefore we obtain by Corollary \ref{Cor:Operator_Estimates}
\begin{align*}
\sum_{ij,i'j'}G^{(I,I')}_{ij,i'j'}a^\dagger_{i}a_j^\dagger \left(a_{I_1}^\dagger a_{I_2}^\dagger a_{I_3}^\dagger a_{I'_1} a_{I'_2} a_{I'_3}\right) a_{i'}a_{j'}\lesssim \sum_{k}|k|^2 c_k^\dagger \frac{\mathcal{N}}{N}c_k + (\mathcal{N}+N^\tau)\frac{\mathcal{N}}{N}.
\end{align*}
Again by Lemma \ref{Lem:Coefficient_Control} we have $\|K^{(I,I')}\|\lesssim N^{-\frac{7}{2}}$, which concludes the proof by Corollary \ref{Cor:Operator_Estimates}, together with the observation that the set $\mathcal{I}$ in the definition of $\mathcal{E}$ in Eq.~(\ref{Eq:First_Error_Term}) is finite.
\end{proof}

The next Lemma \ref{Lem:Rest_Estimate} will give us sufficient bounds on the error term $\mathcal{E}'$ defined in Eq.~(\ref{Eq:Error_Term_With_Primie}), which will be responsible for the appearance of an order $O_{N\rightarrow \infty}\! \left(\sqrt{N}\right)$ error in the main results of this Section Theorem \ref{Th:(1)} and Corollary \ref{Cor:Condensation}.

\begin{lem}
\label{Lem:Rest_Estimate}
There exists a constant $C>0$ such that for $K\leq \sqrt{N}$, where $K$ is as in the definition of $\pi_K$ below Eq.~(\ref{Eq:Definition_pi_K}), and $\epsilon>0$
    \begin{align}
    \label{Eq:Rest_Estimate_I}
       \pm \left(\sum_{|\ell|>K}\lambda_{0,\ell}\,  a_{\ell}^\dagger a_{-\ell}^\dagger\, a_0^\dagger a_0^3 \! + \! \mathrm{H.c.}\right)&\leq \epsilon \sum_\ell |\ell|^2 c_\ell^\dagger c_\ell \! + \! \epsilon\,  \mathcal{N} \! + \! C\frac{\mathcal{N}}{\sqrt{N}} \! + \! \frac{C}{\epsilon} \! \left(\sqrt{N} \! + \! \frac{\mathcal{N}}{\sqrt{K \! + \! 1}}\right),\\ 
           \label{Eq:Rest_Estimate_II}
      \pm \left(\sum_{\ell}\lambda_{k,\ell}  a_\ell^\dagger a_{k-\ell}^\dagger a_k\,  a_0^\dagger a_0^2+\mathrm{H.c.}\right)&\lesssim  \epsilon \sum_\ell |\ell|^2 c_\ell^\dagger  c_\ell+\epsilon \frac{\mathcal{N}^2}{N}+C\frac{\mathcal{N}}{N}+\frac{C}{\epsilon}\left(\frac{\mathcal{N}}{\sqrt{N}}+\frac{\mathcal{N}^2}{N}\right).
 \end{align}
 Furthermore, we have
\begin{align}
\label{Eq:Rest_Estimate_I_with_m}
     \pm \left(\sum_{|\ell|>K}\lambda_{0,\ell}\,  a_{\ell}^\dagger a_{-\ell}^\dagger\, a_0^\dagger a_0^3\frac{\mathcal{N}}{N} \! + \! \mathrm{H.c.}\right) \leq N^{-\frac{1}{2}}\left(\sum_\ell |\ell|^2 c_\ell^\dagger c_\ell+\mathcal{N}\right)+N^{-\frac{3}{2}}\mathcal{N}^2\left(\mathcal{N}+\sqrt{N}\right).
\end{align}
\end{lem}
\begin{proof}
Given $m\in \{0,1\}$, let us define the operator $X:=a_0^\dagger a_0^3\frac{\mathcal{N}^m}{N^m}$ and the coefficients
\begin{align*}
    \Lambda^{(n)}_{\ell,k}:=\overline{(T-1)_{(n-k)(k-\ell)\ell,n 0 0}}+\overline{(T-1)_{(n-k)(k-\ell)\ell,0 n 0}}+\overline{(T-1)_{(n-k)(k-\ell)\ell,0 0 n}},
\end{align*}
and observe that by Lemma \ref{Lem:Coefficient_Control} there exists a constant $C>0$ such that
\begin{align}
\label{Eq:Lambda_Estimate_Rest}
    |\Lambda^{(n)}_{\ell,k}|\leq \frac{C}{N^2(|\ell|^2+|k|^2)}\left(1+\frac{|\ell|^2+|k|^2}{N}\right)^{-1},
\end{align}
where we have assumed w.l.o.g. that $|n|\leq \sqrt{N}$, since $\Lambda^{(n)}_{\ell,k}=0$ in case $|n|>K$ and $K\leq \sqrt{N}$. In order to verify Eq.~(\ref{Eq:Rest_Estimate_I}), respectively Eq.~(\ref{Eq:Rest_Estimate_I_with_m}), let us write
\begin{align}
\nonumber
   & \ \ \ \sum_{|\ell|>K}\lambda_{0,\ell}\,  a_{\ell}^\dagger a_{-\ell}^\dagger a_0^\dagger a_0^3\frac{\mathcal{N}^m}{N^m}=\sum_{|\ell|>K}\lambda_{0,\ell}\,  c_{\ell}^\dagger a_{-\ell}^\dagger X-\sum_{|\ell|>K}\lambda_{0,\ell}\,  (c_\ell-a_{\ell})^\dagger a_{-\ell}^\dagger X\\
       \label{Eq:Rest_Estimate_I_rewritten}
   & = \sum_{|\ell|>K}\lambda_{0,\ell}\,  c_{\ell}^\dagger a_{-\ell}^\dagger X-\sum_{|\ell|>K}\lambda_{0,\ell}\,   a_{-\ell}^\dagger (c_\ell-a_{\ell})^\dagger X-\sum_{|\ell|>K}\sum_{n\neq 0} \lambda_{0,\ell}\, \Lambda^{(n)}_{\ell,0} a_0^{2\dagger }a_n^\dagger a_n X.
\end{align}
Regarding the first term in in Eq.~(\ref{Eq:Rest_Estimate_I_rewritten}), note that we have for $\epsilon>0$ the estimate
\begin{align*}
   \pm \sum_{|\ell|>K}\lambda_{0,\ell}\,  c_{\ell}^\dagger a_{-\ell}^\dagger X\pm \mathrm{H.c.}\leq \epsilon\sum_{\ell}|\ell|^2c_\ell^\dagger c_\ell+\frac{1}{\epsilon} X^\dagger \left(\sum_{|\ell|>K}\frac{|\lambda_{0,\ell}|^2}{\ell^2}a_\ell^\dagger a_\ell+\sum_{|\ell|>K}\frac{|\lambda_{0,\ell}|^2}{\ell^2} \right)X.
\end{align*}
Using $|\lambda_{k,\ell}|\lesssim N^{-2}(1+\frac{|\ell|^2}{N})^{-1}$, see Lemma \ref{Lem:Coefficient_Control}, we have $\sum_{|\ell|>K}\frac{|\lambda_{0,\ell}|^2}{\ell^2}\lesssim N^{-\frac{7}{2}}$ and $\frac{|\lambda_{0,\ell}|^2}{\ell^2}\leq \frac{1 }{N^4 (K^2+1)}$ for $|\ell|>K$, and therefore we obtain for such $K$
\begin{align*}
      \pm  \sum_{|\ell|>K}\lambda_{0,\ell}\,  c_{\ell}^\dagger a_{-\ell}^\dagger X & \pm  \mathrm{H.c.}\lesssim \epsilon\sum_{\ell}|\ell|^2c_\ell^\dagger c_\ell \! + \! \frac{1}{\epsilon N^4 }X^\dagger \left(\frac{\mathcal{N} }{ K^2+1} \! + \! \sqrt{N}\right)X\\
      & \lesssim \epsilon\sum_{\ell}|\ell|^2c_\ell^\dagger c_\ell \! + \frac{1}{\epsilon}\frac{\mathcal{N}^{2m}}{N^{2m}}\left(\frac{\mathcal{N}}{K^2+1} \! + \! \frac{\sqrt{N}}{\epsilon}\right),
\end{align*}
where we have used $X^\dagger X\lesssim N^2 \frac{\mathcal{N}^{2m}}{N^{2m}}$ and $[X,\mathcal{N}]=0$. Regarding the second term in Eq.~(\ref{Eq:Rest_Estimate_I_rewritten}), let us use $\sum_{|\ell|>K}|\ell|^{-\frac{1}{2}}|N^4\lambda_{0,\ell}|^2 a_{\ell}^\dagger a_{\ell}\lesssim \frac{\mathcal{N}}{\sqrt{K+1}}$ as well as the fact that $\sum_{\ell}\sqrt{|\ell|}(c_\ell-a_\ell) (c_\ell-a_\ell)^\dagger\lesssim \mathcal{N}$ by Lemma \ref{lem:Comparison}, to estimate for $\kappa>0$
\begin{align*}
   \pm \! \sum_{|\ell|>K}\lambda_{0,\ell}\,   a_{-\ell}^\dagger (c_\ell \! - \! a_{\ell})^\dagger X \pm \mathrm{H.c.}\lesssim \frac{1}{\kappa \sqrt{K+1}}\mathcal{N}+\frac{\kappa }{N^4}X^\dagger \mathcal{N} X\leq \frac{1}{\kappa \sqrt{K+1}}\mathcal{N}+\kappa \frac{\mathcal{N}^{2m}}{N^{2m}}\mathcal{N}.
\end{align*}
Regarding the final term in Eq.~(\ref{Eq:Rest_Estimate_I_rewritten}) we have that $\sum_{\ell}|\Lambda^{(n)}_{\ell,0}|\lesssim N^{-\frac{3}{2}}$ by Eq.~(\ref{Eq:Lambda_Estimate_Rest}), and hence
\begin{align*}
   \pm \sum_{|\ell|>K}\sum_{n\neq 0} \lambda_{0,\ell}\, \Lambda^{(n)}_{\ell,0} a_0^{2\dagger }a_n^\dagger a_n X \pm \mathrm{H.c.}\lesssim N^{-\frac{1}{2}}\mathcal{N}.
\end{align*}
For $m=0$, the choice $\kappa:=\epsilon$ yields Eq.~(\ref{Eq:Rest_Estimate_I}) and for $m=1$ the choice $\kappa:=\sqrt{N}$ and $\epsilon:=\frac{1}{\sqrt{N}}$ yields Eq.~(\ref{Eq:Rest_Estimate_I_with_m}).

Regarding the proof of Eq.~(\ref{Eq:Rest_Estimate_II}), let us define the operators $d_\ell:=\lambda_{k,\ell} N^{\frac{3}{2}} a_{k-\ell}^\dagger a_k$ and write $a^\dagger_\ell d_\ell=c_\ell^\dagger d_\ell+d_\ell (c_\ell-a_\ell)^\dagger+[(c_\ell-a_\ell)^\dagger,d_\ell]$. We compute 
\begin{align*}
   \sum_\ell &[(c_\ell-a_\ell)^\dagger,d_\ell]=\frac{1}{2} N^{\frac{3}{2}}\Big(a^{3\dagger}_{0} \sum_{i j \ell}\frac{1}{3}\Lambda^{(0)}_{\ell,-i}\lambda_{k,\ell}[ a_i a_{-i-\ell} ,a_{k-\ell}^\dagger a_k]\\
    & \ \ \ \ +a^{2\dagger}_{0}\sum_{|n|\leq K}\sum_{i  \ell}\Lambda^{(n)}_{\ell,n-i}\lambda_{k,\ell}[ a^\dagger_{n} a_i a_{n-i-\ell} ,a_{k-\ell}^\dagger a_k]\Big)\frac{a_0^\dagger a_0^2}{N^{\frac{3}{2}}}\\
    &=\frac{a_0^{3\dagger}}{N^{\frac{3}{2}}}\mu^{(1)}_k a_{k}a_{-k}\frac{a_0^\dagger a_0^2}{N^{\frac{3}{2}}}+\frac{a_0^{2\dagger}}{N}\sum_{|n|\leq K}\mu^{(2)}_{k,n} a^{\dagger}_n a_{k}a_{n-k}\frac{a_0^\dagger a_0^2}{N^{\frac{3}{2}}}-\frac{a_0^{2\dagger}}{N}\sum_{i,\ell }\mu^{(3)}_{k, i,\ell}a^\dagger_{k-\ell}a_i a_{k-i-\ell}\frac{a_0^\dagger a_0^2}{N^{\frac{3}{2}}},
\end{align*}
where we define the coefficients
\begin{align*}
    \mu^{(1)}_k: & =N^{3}\sum_{\ell }\frac{1}{3}\Lambda^{(0)}_{\ell, k}\lambda_{k,\ell},\\
    \mu^{(2)}_{k,n}: & =N^{\frac{5}{2}}\sum_{\ell}\Lambda^{(n)}_{\ell, k}\lambda_{k,\ell},\\
    \mu^{(3)}_{k,i,\ell }: & =N^{\frac{5}{2}}\Lambda^{(k)}_{\ell,k-i}\lambda_{k,\ell}.
\end{align*}
Using again Eq.~(\ref{Eq:Lambda_Estimate_Rest}) and $|\lambda_{k,\ell}|\lesssim N^{-2}(1+\frac{|\ell|^2}{N})^{-1}$, we immediately obtain $|\mu^{(1)}_k|\lesssim \frac{1}{\sqrt{N}}$, $|\mu^{(2)}_{k,n}|\lesssim \frac{1}{N}$, $|\mu^{(3)}_{k,i,\ell }|\lesssim N^{-\frac{3}{2}}$ and $\sum_{i}|\mu^{(3)}_{k,i,\ell }|\lesssim \frac{1}{N}$, and therefore by Cauchy-Schwarz
\begin{align*}
    \sum_\ell \left([(c_\ell-a_\ell)^\dagger,d_\ell]\frac{a_0^\dagger a_0^2}{N^{\frac{3}{2}}}+\mathrm{H.c.}\right)\lesssim \frac{\mathcal{N}}{\sqrt{N}}.
\end{align*}
Consequently
   \begin{align*}
     &    \left(\sum_{\ell}a_\ell^\dagger d_\ell \frac{a_0^\dagger a_0^2}{N^\frac{3}{2}}+\mathrm{H.c.}\right)\lesssim \epsilon\sum_{\ell}|\ell|^2c_\ell^\dagger c_\ell+\frac{1}{\epsilon}\sum_{\ell}\frac{1}{|\ell|^2} d_\ell^\dagger d_\ell+\epsilon\sum_{\ell}d_\ell d_\ell^\dagger\\
        & \ \ \ \ \   \ \ \ \ \   \ \ \ \ \  +\frac{1}{\epsilon}\sum_{\ell}\frac{a_0^{3\dagger} a_0}{N^2}(c_\ell-a_\ell) (c_\ell-a_\ell)^\dagger \frac{a_0^\dagger a_0^3}{N^2}+\frac{\mathcal{N}}{\sqrt{N}}.
   \end{align*}
   Similar to the proof of Eq.~(\ref{Eq:Rest_Estimate_I}), we observe that $\sum_{\ell}d_\ell d_\ell^\dagger\lesssim \frac{\mathcal{N}}{N}\mathcal{N}$ and, using Lemma \ref{lem:Comparison},
   \begin{align*}
  & \ \ \ \  \sum_{\ell}\frac{a_0^{3\dagger} a_0}{N^2}(c_\ell-a_\ell) (c_\ell-a_\ell)^\dagger \frac{a_0^\dagger a_0^3}{N^2}\lesssim \frac{a_0^{3\dagger} a_0}{N^2}\frac{\mathcal{N}^2}{N} \frac{a_0^\dagger a_0^3}{N^2}\leq \frac{\mathcal{N}^2}{N},\\
  &  \sum_{\ell}\frac{1}{|\ell|^2} d_\ell^\dagger d_\ell\lesssim \frac{\mathcal{N}}{\sqrt{N}}+N^3\sum_{\ell}\frac{1}{|\ell|^2}|\lambda_{k,\ell}|^2   a_k^\dagger a_{k-\ell}^\dagger a_{k-\ell} a_k\lesssim \frac{\mathcal{N}}{\sqrt{N}}+\frac{\mathcal{N}^2}{N}.
   \end{align*}
\end{proof}

Having Lemma \ref{Lem:Quadratic_Potential_Estimate}, Lemma \ref{Lem:First_Error_Term_Estimate} and Lemma \ref{Lem:Rest_Estimate} at hand, we can use the lower bound in Eq.~(\ref{Eq:First_Lower_Bound}) in order to derive the following Theorem \ref{Th:(1)}, which provides strong lower bounds on the quantity $\braket{\Psi,H_N \Psi}$. Note however that Theorem \ref{Th:(1)} is only applicable for states $\Psi$ satisfying (BEC) in the spectral sense $\mathds{1}\! \left(\mathcal{N}\leq \epsilon N\right)\Psi=\Psi$, where the orthogonal projection $\mathds{1}\! \left(\mathcal{N}\leq \epsilon N\right)$ is defined by the means of functional calculus. Here we refer to a Hilbert space element $\Psi$ as a state, in case $\|\Psi\|=1$.

\begin{thm}
\label{Th:(1)}
    There exist constants $\delta,C>0$ and $\epsilon>0$, such that 
    \begin{align*}
        \braket{\Psi,H_N\Psi}\geq \frac{1}{6} b_{\mathcal{M}}(V)N+\delta \sum_{k}|k|^2 c_k^\dagger  c_k+\delta\braket{\Psi,\mathcal{N} \Psi}-C\sqrt{N}
    \end{align*}
    for any state $\Psi$ satisfying $\mathds{1}\! \left(\mathcal{N}\leq \epsilon N\right)\Psi=\Psi$, where $\mathcal{N}:=\sum_{k\neq 0}a_k^\dagger a_k$.
\end{thm}
\begin{proof}
By Eq.~(\ref{Eq:First_Lower_Bound}) together with the estimates in Lemma \ref{Lem:Quadratic_Potential_Estimate}, Lemma \ref{Lem:First_Error_Term_Estimate} and Lemma \ref{Lem:Rest_Estimate} we can bound for $\alpha,\tau>0$ the operator $H_N-\frac{1}{6} b_{\mathcal{M}}(V)N$ from below by
\begin{align*}
\frac{1}{2}\sum_k |k|^2 c_k^\dagger c_k-C\sum_k |k|^2 c_k^\dagger \left(\frac{\mathcal{N}}{N}+N^{-\frac{1}{2}}\right) c_k - \alpha\sum_k |k|^{2\tau}a_k^\dagger a_k-C\left(\frac{\mathcal{N}}{N}+N^{-\frac{1}{2}}\right)\mathcal{N}-C N^{\frac{1}{2}},
\end{align*}
for a suitable constant $C$. In the following let $\Psi$ be a state satisfying $\mathds{1}_{[0,\epsilon N)}\! \left(\mathcal{N}\right)\Psi=\Psi$ and define $\Psi_k:=c_k \Psi$. By the definition of $c_k$ it is clear that $\mathds{1}_{[0,\epsilon N+2)}\! \left(\mathcal{N}\right)\Psi_k=\Psi_k$, and therefore
\begin{align*}
  & \ \  \left\langle \Psi, \sum_k |k|^2 c_k^\dagger \left(\frac{\mathcal{N}}{N}+N^{-\frac{1}{2}}\right) c_k \Psi\right\rangle=\sum_k |k|^2\left\langle \Psi_k,\left(\frac{\mathcal{N}}{N}+N^{-\frac{1}{2}}\right)\Psi_k\right\rangle\\
  & =\sum_k |k|^2\left\langle \Psi_k,\left(\frac{\mathcal{N}}{N}+N^{-\frac{1}{2}}\right)\mathds{1}_{[0,\epsilon N+2)}\! \left(\mathcal{N}\right)\Psi_k\right\rangle\leq \sum_k |k|^2\left\langle \Psi_k,\left(\frac{\epsilon N +2}{N}+N^{-\frac{1}{2}}\right)\Psi_k\right\rangle\\
  & = \left(\frac{\epsilon N +2}{N}+N^{-\frac{1}{2}}\right)\left\langle \Psi, \sum_k |k|^2 c_k^\dagger  c_k \Psi\right\rangle.
\end{align*}
In a similar fashion we have $\braket{\Psi,\mathcal{N}^2\Psi}\leq \epsilon N \braket{\Psi,\mathcal{N} \Psi}$. Furthermore, note that
\begin{align*}
 \sum_k |k|^{2\tau}a_k^\dagger a_k\lesssim \sum_k |k|^2 c_k^\dagger c_k+\frac{\mathcal{N}^2}{N}+N^{\frac{1}{2}}   
\end{align*}
by Lemma \ref{lem:Comparison} for $\tau<\frac{1}{2}$. Choosing $\alpha$ small enough, we therefore obtain 
\begin{align*}
\braket{\Psi,H_N\Psi}\geq \frac{1}{6} b_{\mathcal{M}}(V)N+\frac{1}{3}\left\langle \Psi, \sum_k |k|^2 c_k^\dagger c_k\Psi\right\rangle - C N^{\frac{1}{2}} - \epsilon C \braket{\Psi, \mathcal{N}\Psi}
\end{align*}
for states $\Psi$ satisfying $\mathds{1}_{[0,\epsilon N)}\! \left(\mathcal{N}\right)\Psi=\Psi$. Again by Lemma \ref{lem:Comparison} we have 
\begin{align*}
  \mathcal{N}\left(1-R\frac{\mathcal{N}}{N}\right)\leq R \sum_k |k|^2 c_k^\dagger c_k+R  
\end{align*}
for a suitable constant $R$. Using $\mathds{1}\! \left(\mathcal{N}\leq \epsilon N\right)\Psi=\Psi$ with $\epsilon$ small enough such that $R\epsilon<1$, we therefore have $\braket{\Psi,\mathcal{N}\Psi}\lesssim \left\langle \Psi, \sum_k |k|^2 c_k^\dagger c_k\Psi\right\rangle+1$. Choosing $\epsilon$ small enough concludes the proof. 
\end{proof}

Before we come to the lower bound on the ground state energy $E_N$ in the main result of this Section Corollary \ref{Cor:Condensation}, let us first state the corresponding upper bound in the subsequent Theorem \ref{Th:First_Order_Upper_Bound}. The proof of Theorem \ref{Th:First_Order_Upper_Bound} is content of the following Section \ref{Sec:First_Order_Upper_Bound}.

\begin{thm}
\label{Th:First_Order_Upper_Bound}
    There exists a constant $C>0$ such that the ground state energy $E_N$ is bounded from above by $E_N\leq \frac{1}{6}b_\mathcal{M}(V)N+C\sqrt{N}$.
\end{thm}

It has been verified in \cite{NRT1}, for the more general setting of particles being confined by an additional external potential, that any approximate ground state $\Psi_N$ of the operator $H_N$ satisfies complete Bose-Einstein condensation $ \braket{\Psi_N,\mathcal{N}\Psi_N}=o_{N\rightarrow \infty}(N)$. Combining this observation with Theorem \ref{Th:(1)}, allows us to derive an asymptotically correct lower bound on the ground state energy in Corollary \ref{Cor:Condensation} with an error of the order $\sqrt{N}$, see Eq.~(\ref{Eq:Condensation_Energy}). In this context we call $\Psi_N$ an approximate ground state, in case $\|\Psi_N\|=1$ and there exists a constant $C>0$ such that
 \begin{align}
 \label{Eq:Def_Approx_GS}
     \braket{\Psi_N,H_N\Psi_N}\leq E_N+C.
 \end{align}
 Note that the assumption in Eq.~(\ref{Eq:Def_Approx_GS}) is more restrictive compared to the one employed in \cite{NRT1}, where the authors call $\Psi_N$ an approximate ground state in case $\|\Psi_N\|=1$ and
 \begin{align}
 \label{Eq:Def_Approx_GS_NRT}
      \lim_{N\rightarrow \infty}\frac{1}{N}\braket{\Psi_N,H_N\Psi_N}= \frac{1}{6}b_\mathcal{M}(V).
 \end{align}
 The fact that Eq.~(\ref{Eq:Def_Approx_GS}) implies Eq.~(\ref{Eq:Def_Approx_GS_NRT}) follows immediately from the leading order asymptotics in Eq.~(\ref{Eq:First_Order_Energy_Asymptotics}), which has been verified in \cite{NRT1}, together with the trivial lower bound $ \braket{\Psi_N,H_N\Psi_N}\geq E_N$. Furthermore, in combination with the upper bound on $E_N$ derived in Theorem \ref{Th:First_Order_Upper_Bound}, we obtain that the ground state $\Psi_N^{\mathrm{GS}}$ of $H_N$ satisfies (BEC) with a rate $\frac{\sqrt{N}}{N}=\frac{1}{\sqrt{N}}$, which concludes the proof of our second main Theorem \ref{Eq:Main_Theorem_Introduction_Condensation}. Finally we can improve this result to (BEC) in the spectral sense 
 \begin{align*}
   \mathds{1}\! \left(\mathcal{N}\leq K\sqrt{N}\right)\Phi_N=\Phi_N,  
 \end{align*}
however we have to consider slightly modified states $\Phi_N$ here. 

\begin{cor}
\label{Cor:Condensation}
The ground state $\Psi_N^{\mathrm{GS}}$ of the operator $H_N$ satisfies for a suitable $C>0$ 
\begin{align*}
    \braket{\Psi_N^{\mathrm{GS}},\mathcal{N}\Psi_N^{\mathrm{GS}}}\leq C\sqrt{N},
\end{align*}
and we have the lower bound
\begin{align}
\label{Eq:Condensation_Energy}
   E_N\geq \frac{1}{6}b_\mathcal{M}(V)N-C\sqrt{N} .
\end{align}
Furthermore there exists a constant $C>0$ and states $\Phi_N$, such that $\Phi_N$ is an approximate ground state of $H_N$ satisfying (BEC) in the spectral sense with rate $\frac{1}{\sqrt{N}}$, i.e.
\begin{align*}
    \braket{\Phi_N, H_N\Phi_N}\leq E_N+C,\\
\mathds{1}\! \left(\mathcal{N}\leq C\sqrt{N}\right)\Phi_N=\Phi_N,
\end{align*}
and we have the estimate on the kinetic energy $\left\langle \Phi_N, \sum_{k}|k|^2 c_k^\dagger  c_k \Phi_N \right\rangle\leq C\sqrt{N}$.
\end{cor}
\begin{proof}
    From the results in \cite{NRT1} we know that the ground state $\Psi^{\mathrm{GS}}_N$ of $H_N$ satisfies 
    \begin{align*}
       \braket{\Psi^{\mathrm{GS}}_N,\mathcal{N}\Psi^{\mathrm{GS}}_N}=o_{N\rightarrow \infty}(N) .
    \end{align*}
 Consequently we know by Lemma \ref{Lem:IMS} that there exist states $\xi_N$ satisfying 
    \begin{align}
    \label{Eq:xi_Energy}
        \braket{\xi_N,H_N \xi_N}\leq E_N+C\left(\frac{1}{\epsilon \sqrt{N}}+\frac{1}{\epsilon^2 N}\right)
    \end{align}
   and $\mathds{1}(\mathcal{N}\leq \epsilon N)\xi_N=\xi_N$, where we choose $\epsilon>0$ as in Theorem \ref{Th:(1)}. By Theorem \ref{Th:(1)} and Theorem \ref{Th:First_Order_Upper_Bound} we therefore obtain for a suitable constant $C>0$
   \begin{align}
   \nonumber
     &  \frac{1}{6} b_{\mathcal{M}}(V)N+\delta \left\langle \xi_N , \sum_{k}|k|^2 c_k^\dagger  c_k \xi_N \right\rangle+\delta\braket{\xi_N,\mathcal{N} \xi_N}-C\sqrt{N} \leq \braket{\xi_N,H_N \xi_N}\\
     \label{Eq:Inequality_Chain}
     & \ \ \ \ \  \ \ \ \ \  \ \ \ \ \  \ \ \ \ \ \leq E_N +C\leq \frac{1}{6} b_{\mathcal{M}}(V)N+C \sqrt{N}.
   \end{align}
   This immediately implies $E_N\geq \frac{1}{6} b_{\mathcal{M}}(V)N-C\sqrt{N}$ for a suitable constant $C>0$ and
   \begin{align}
   \label{Eq:Improved_Rate_BEC}
     \braket{\xi_N,\mathcal{N}\xi_N}=O_{N\rightarrow \infty}\! \left(\sqrt{N}\right) ,
   \end{align}
as well as $\left\langle \xi_N , \sum_{k}|k|^2 c_k^\dagger  c_k \xi_N \right\rangle=O_{N\rightarrow \infty}\! \left(\sqrt{N}\right) $. Furthermore, there exist by Lemma \ref{Lem:IMS} states $\widetilde{\xi}_N$ satisfying $\mathds{1}\! \left(\mathcal{N}> \frac{\epsilon}{2} N\right)\widetilde \xi_N=\widetilde \xi_N$ and
  \begin{align}
  \label{Eq:Condensation_proof_ground_state}
        \braket{\Psi^{\mathrm{GS}}_N, \mathcal{N}\Psi^{\mathrm{GS}}_N}  \lesssim \braket{\xi_N,\mathcal{N} \xi_N }+\frac{ \sqrt{N}}{\braket{\widetilde \xi_N,H_N \widetilde \xi_N}-E_N}\lesssim \sqrt{N}+\frac{ \sqrt{N}}{\braket{\widetilde \xi_N,H_N \widetilde \xi_N}-E_N}.
  \end{align}
In the following we show by a contradiction argument that
\begin{align}
\label{Eq:Condensation_Contradiction_Argument}
    \liminf_N \left\{\braket{\widetilde \xi_N,H_N \widetilde \xi_N}-E_N\right\}=\infty.
\end{align}
For this purpose let us assume that Eq.~(\ref{Eq:Condensation_Contradiction_Argument}) is violated, i.e. we assume that there exists a subsequence $N_j$ and a constant $C>0$ such that $\sup_{j}\braket{\widetilde \xi_{N_j},H_{N_j} \widetilde \xi_{N_j}}-E_{N_j}\leq C$. Let us complete this subsequence to a proper sequence by defining $\xi'_N:=\widetilde \xi_N$ in case $N=N_j$ for some $j$ and $\xi'_N:=\Psi^{\mathrm{GS}}_N$ otherwise. Clearly $\xi'_N$ is a sequence of approximate ground states, see Eq.~(\ref{Eq:Def_Approx_GS}), and as such $\xi'_N$ satisfies complete (BEC) by the results in \cite{NRT1}, i.e. $\braket{\xi'_N,\mathcal{N}\xi'_N}=o_{N\rightarrow \infty}(N)$. This is however a contradiction to
\begin{align*}
    \braket{\xi'_{N_j},\mathcal{N}\xi'_{N_j}}=\braket{\widetilde \xi_{N_j},\mathcal{N}\widetilde \xi_{N_j}}\geq \frac{\epsilon}{2}N,
\end{align*}
which concludes the proof of Eq.~(\ref{Eq:Condensation_Contradiction_Argument}). Combining Eq.~(\ref{Eq:Condensation_proof_ground_state}) and Eq.~(\ref{Eq:Condensation_Contradiction_Argument}) yields
\begin{align*}
 \braket{\Psi^{\mathrm{GS}}_N,\mathcal{N}\Psi^{\mathrm{GS}}_N}\leq C\sqrt{N},
\end{align*}
 for a suitable constant $C>0$. Applying again Lemma \ref{Lem:IMS} for the state $\Psi^{\mathrm{GS}}_N$ and $M:=K \sqrt{N}$, we obtain states $\Phi_N$ satisfying $\mathds{1}(\mathcal{N}\leq K \sqrt{N})\Phi_N=\Phi_N$ and
   \begin{align*}
       \braket{\Phi_N, H_N \Phi_N}\leq E_N+\frac{C}{1-\frac{2}{K\sqrt{N}}\braket{\Psi^{\mathrm{GS}}_N,\mathcal{N}\Psi^{\mathrm{GS}}_N}}\leq E_N+\frac{C}{1-\frac{2C}{K}},
   \end{align*}
   for a large enough $C>0$. Consequently $\Phi_N$ is a sequence of approximate ground states for $K>2C$. Finally we notice that the states $\Phi_N$ satisfy the chain of inequalities in Eq.~(\ref{Eq:Inequality_Chain}) as well, and therefore 
   \begin{align*}
      \left\langle \Phi_N , \sum_{k}|k|^2 c_k^\dagger  c_k \Phi_N \right\rangle=O_{N\rightarrow \infty}\! \left(\sqrt{N}\right) .
   \end{align*}
\end{proof}

\section{First Order Upper Bound}
\label{Sec:First_Order_Upper_Bound}
It is the goal of this section to introduce a trial state $\Gamma$, which simultaneously annihilates the variables $c_k$ for $k\neq 0$ and $\psi_{\ell m n}$ in case $(\ell,m,n)\neq 0$, at least in an approximate sense, which allows us to verify the upper bound on the ground state energy $E_N$ in Theorem \ref{Th:First_Order_Upper_Bound}. For the rest of this Section we specify the parameter $K$ introduced above the definition of $\pi_K$ in Eq.~(\ref{Eq:Definition_pi_K}) as $K:=0$, which especially means that with $\eta_{ijk}:=(T-1)_{i j k, 0 0 0}$ we have
\begin{align}
\label{Eq:c_for_K=0}
c_k &=a_k+\frac{1}{2}\sum_{i j}\eta_{ijk} \, a_i^\dagger a_j^\dagger a_{0}^3,\\
\label{Eq:Psi_for_K=0}
\psi_{i j k} &=a_ia_ja_k+\eta_{ijk} \,  a_0^3.
\end{align}
In order to find a suitable state $\Gamma$, let $\eta_{ijk}:=(T-1)_{i j k, 0 0 0}$ and let us follow the strategy in \cite{NRT1}, respectively in the case of Bose gases with two-particle interactions see e.g. \cite{BBCS,BCS,HST}, by defining the generator
\begin{align}
\label{Eq:Generator}
    \mathcal{G}:=\frac{1}{6}\sum_{i j k}\eta_{ijk}\, a_i^\dagger a_j^\dagger a_k^\dagger a_0^3 
\end{align}
of a unitary group $U_s:=e^{s\mathcal{G}^\dagger-s\mathcal{G}}$ and $U:=U_1$. As we show this Section, the unitary $U$ has the property $U^{-1}c_k U\approx a_k$ and $U^{-1}\psi_{ijk}U\approx a_i a_j a_k$ in a suitable sense. Denoting with
\begin{align}
\label{Eq:Gamma_0}
\Gamma_0(x_1,\dots ,x_N):=1    
\end{align}
the constant function in $L^2_\mathrm{sym}\! \left(\Lambda^N\right)$, i.e. $a_k \Gamma_0=0$ for $k\neq 0$, we observe that $\Gamma:=U \Gamma_0$ is a suitable trial state for the (approximate) annihilation of $c_k$, given $k\neq 0$, and $\psi_{\ell m n}$, given $(\ell,m,n)\neq 0$. Note that the action of the unitary $U$ introduces a three-particle correlation structure on the completely uncorrelated wave function $\Gamma_0$. 

For the purpose of verifying that $U^{-1}\psi_{ijk}U$ is approximately identical to $a_i a_j a_k$, we first apply Duhamel's formula, which yields
\begin{align}
\label{Eq:Duhamel_I}
    U^{-1} a_i a_j a_k U=a_i a_j a_k  - \int_0^1 U_{-s} \big[a_i a_j a_k ,\mathcal{G}\big]U_s \, \mathrm{d}s+\int_0^1 U_{-s} \big[a_i a_j a_k ,\mathcal{G}^\dagger \big]U_s \, \mathrm{d}s.
\end{align}
Furthermore, note that we can write 
\begin{align*}
    \big[a_i a_j a_k ,\mathcal{G}\big]=\eta_{ijk}a_0^3+(\delta_1 \psi)_{ijk}+(\delta_2 \psi)_{ijk}
\end{align*}
using the definition
\begin{align*}
(\delta_1 \psi)_{i_1 i_2 i_3}:  & =\frac{1}{2}\sum_{\sigma\in S_3}\sum_j \eta_{i_{\sigma_1}i_{\sigma_2} j}\mathds{1}(i_{\sigma_3}=0)\, a_{j}^\dagger a_0^4,\\
    (\delta_2 \psi)_{i_1 i_2 i_3}:  & =\frac{1}{2}\sum_{\sigma\in S_3}\sum_j \eta_{i_{\sigma_1}i_{\sigma_2} j}\mathds{1}(i_{\sigma_3}\neq 0)\, a_{j}^\dagger a_{i_{\sigma_3}}a_0^3+\frac{1}{4}\sum_{\sigma\in S_3}\sum_{j k} \eta_{i_{\sigma_1}j k}\, a_k^\dagger a_{j}^\dagger a_{i_{\sigma_2}}a_{i_{\sigma_3}}a_0^3.
\end{align*}
Therefore we can identify the transformed operators $U^{-1}\psi_{ijk}U$ as
   \begin{align}
    \nonumber
        U^{-1}\psi_{ijk} U \! & = \! a_i a_j a_k \! +  \! \!  \! \int_0^1    \! \! \!  \! U_{-s} \! \left\{ \! \big[a_i a_j a_k ,\mathcal{G}^\dagger \big] \! - \! (\delta_1 \psi)_{ijk} \! - \! (\delta_2 \psi)_{ijk} \! - \! \eta_{ijk}a_0^3 \!\right\} \! U_s \mathrm{d}s \! + \!\eta_{ijk} U^{-1}   a_0^3  U\\
        \nonumber
        &=\! a_i a_j a_k  +   \!  \! \int_0^1    \! \! \!  \! U_{-s}   \big[a_i a_j a_k ,\mathcal{G}^\dagger \big]   U_s \mathrm{d}s  -  \!  \! \int_0^1    \! \! \!  \! U_{-s}   (\delta_1 \psi)_{ijk}   U_s \mathrm{d}s -  \!  \! \int_0^1    \! \! \!  \! U_{-s}   (\delta_2 \psi)_{ijk}   U_s \mathrm{d}s\\
        \label{Eq:Double_Duhamel}
        & \ \ \ \ \ \ \ \ \ \ +\int_0^1\int_s^1 U_{-t} \eta_{ijk} \big[a_0^3 ,\mathcal{G}^\dagger \big]U_t\mathrm{d}t\mathrm{d}s,
    \end{align}
where we have used Duhamel's formula to express $U^{-1}\eta_{ijk}a_0^3U-U_{-s}\eta_{ijk}a_0^3U_s$. The following Lemma \ref{Lem:First_Order_Upper_Bound_Useful} demonstrates that we can treat the quantities $\delta_1 \psi$ and $\delta_2 \psi$ in Eq.~(\ref{Eq:Double_Duhamel}) as error terms. In order to formulate Lemma \ref{Lem:First_Order_Upper_Bound_Useful}, recall the set $\mathcal{L}_0:=\{(0,0,0)\}$ from Eq.~(\ref{Eq:Def_Low_Momentum_Set_0}) and let us define 
\begin{align*}
    A := (2\pi \mathbb{Z})^9\setminus \mathcal{L}_0 =(2\pi \mathbb{Z})^9\setminus \{(0,0,0)\},
\end{align*}
and the potential energy $\mathcal{E}_\mathcal{P}$ of an operator valued three particle vector $\Theta_{i_1 i_2 i_3}$ as
\begin{align}
\label{Eq:Potential_Energy_Operator}
    \mathcal{E}_\mathcal{P}(\Theta):=\sum_{(i_1 i_2 i_3),(i_1' i_2' i_3')\in A}(V_N)_{i_1 i_2 i_3,i_1' i_2' i_3'}\Theta_{i_1 i_2 i_3}^\dagger \Theta_{i_1' i_2' i_3'}.
\end{align}
To keep the notation light, we will occasionally write $\mathcal{E}_\mathcal{P}(\Theta_{i_1 i_2 i_3})$ for $\mathcal{E}_\mathcal{P}(\Theta)$ with dummy indices $i_1 i_2 i_3$.
\begin{lem}
\label{Lem:First_Order_Upper_Bound_Useful}
There exists a constant $C>0$, such that 
   \begin{align}
   \label{Eq:First_Order_Upper_Bound_Useful_I}
        \mathcal{E}_\mathcal{P}(\delta_1 \psi) & \leq CN^{\frac{1}{2}}(\mathcal{N}+1),\\
   \label{Eq:First_Order_Upper_Bound_Useful_II}
        \mathcal{E}_\mathcal{P}(\delta_2 \psi) & \leq C (\mathcal{N}+1)^4.
   \end{align}
   Furthermore, $\mathcal{E}_P\! \left(\big[a_{i_1} a_{i_2} a_{i_3} ,\mathcal{G}^\dagger \big]\right)\leq C N^{-\frac{3}{2}}(\mathcal{N}+1)^5$.
\end{lem}
\begin{proof}
   Let us define $A_j$ as the set of all $s$ such that $(-s, j-s, 0)\in A$ and 
   \begin{align*}
      \alpha_j:=9\sum_{s,t\in A_j}\eta_{-s (s-j) j}\overline{\eta_{-t (t-j) j}}(V_N)_{-s (j-s) 0,-t (j-t) 0} .
   \end{align*}
Making use of the fact that $V_N\geq 0$, we obtain by the Cauchy-Schwarz inequality
    \begin{align*}
        \sum_{(i_1 i_2 i_3),(i'_1 i'_2 i'_3)\in A} \!  \! (V_N)_{i_1 i_2 i_3,i'_1 i'_2 i'_3}(\delta_1 \psi)_{i_1 i_2 i_3}^\dagger (\delta_1 \psi)_{i'_1 i'_2 i'_3}\leq (a_0^\dagger)^4 a_0^4\sum_{j}\alpha_j\,  a_{j}a_j^\dagger\leq N^4 \left(\sum_j \alpha_j\right) \! (\mathcal{N} \! + \! 1).
    \end{align*}
    By Lemma \ref{Lem:Coefficient_Control} and the fact that $|(V_N)_{-s (j-s) 0,-t (j-t) 0}|\lesssim \frac{1}{N^2}$, we have $N^4\sum_j \alpha_j\lesssim C N^{\frac{1}{2}}$, which concludes the proof of Eq.~(\ref{Eq:First_Order_Upper_Bound_Useful_I}). In order to verify Eq.~(\ref{Eq:First_Order_Upper_Bound_Useful_II}), let us first define 
    \begin{align*}
        (\widetilde{\delta}_{2,\sigma} \psi)_{i_1 i_2 i_3}:=\sum_{j k} \eta_{i_{\sigma_1}j k}\, a_k^\dagger a_{j}^\dagger a_{i_{\sigma_2}}a_{i_{\sigma_3}}a_0^3.
    \end{align*}
Then we obtain, using the sign $V_N\geq 0$ and a Cauchy-Schwarz estimate,
    \begin{align*}
        \mathcal{E}_\mathcal{P}(\delta_2 \psi) \! \lesssim  \! \sum_{\sigma\in S_3} \! \mathcal{E}_\mathcal{P}\! \left(\widetilde{\delta}_{2,\sigma} \psi\right) \! + \! \mathcal{E}_\mathcal{P}\! \left(\delta_2 \psi \! - \! \frac{1}{4} \! \sum_{\sigma\in S_3} \! \widetilde{\delta}_{2,\sigma} \psi\right) \! = \! 6\mathcal{E}_\mathcal{P}\! \left(\widetilde{\delta}_{2,\mathrm{id}} \psi\right) \! + \! \mathcal{E}_\mathcal{P}\! \left(\delta_2 \psi \! - \! \frac{1}{4}\sum_{\sigma\in S_3}\widetilde{\delta}_{2,\sigma} \psi\right).
    \end{align*}
    Proceeding as in the proof of Eq.~(\ref{Eq:First_Order_Upper_Bound_Useful_I}), we obtain
    \begin{align*}
       \mathcal{E}_\mathcal{P}\! \left(\delta_2 \psi \! - \! \frac{1}{4}\sum_{\sigma\in S_3}\widetilde{\delta}_{2,\sigma} \psi\right)\lesssim N^{-\frac{1}{2}}(\mathcal{N}+1)^2 .
    \end{align*}
   Regarding the term $\mathcal{E}_\mathcal{P}\! \left(\widetilde{\delta}_{2,\mathrm{id}} \psi\right)$, let us define
    \begin{align*}
       (G)_{i_1.. i_4, j_1..j_4}:=(V_N)_{i_1 i_2 (-j_3-j_4),j_1 j_2 (-i_3-i_4)}\eta_{i_3 i_4 (-i_3-i_4)}\overline{\eta_{j_3 j_4 (-j_3-j_4)}}
    \end{align*}
   for $(i_1,i_2,-j_3-j_4)\in A$ and $(j_1,j_2,-i_3-i_4)\in A$, and $(G)_{i_1.. i_4, j_1..j_4}:=0$ otherwise. Then
    \begin{align}
    \nonumber
        & \ \ \ \  \ \ \ \ \mathcal{E}_\mathcal{P}\! \left(\widetilde{\delta}_{2,\mathrm{id}} \psi\right)= \sum_{i_1.. i_4, j_1..j_4}(G)_{i_1.. i_4, j_1..j_4}(a_0^\dagger)^3 a_{i_1}^\dagger a_{i_2}^\dagger a_{j_3}a_{j_4}a_{i_3}^\dagger a_{i_4}^\dagger a_{j_1}a_{j_2} a_0^3\\
        \nonumber
        & = \!  \!  \!  \! \sum_{i_1.. i_4, j_1..j_4} \!  \!  \!  \!  \!  \! (G)_{i_1.. i_4, j_1..j_4}(a_0^\dagger)^3 a_{i_1}^\dagger a_{i_2}^\dagger a_{i_3}^\dagger a_{i_4}^\dagger a_{j_3}a_{j_4} a_{j_1}a_{j_2} a_0^3 \! + \! 2 \!  \!  \!  \!  \!  \!  \!  \! \sum_{i_1.. i_3, j_1..j_3,k} \!  \!  \!  \!  \!  \!  \! (G')_{i_1.. i_3 , j_1..j_3 }(a_0^\dagger)^3 a_{i_1}^\dagger a_{i_2}^\dagger a_{i_3}^\dagger  a_{j_3} a_{j_1}a_{j_2} a_0^3\\
         \label{Eq:First_Order_Upper_Bound_Useful_Aux_II}
        & \ \  \ \  \ \  \ \  \ \  \ \  \ \ + \sum_{i_1  i_2, j_1 j_2} \!  \!  \!  (G'')_{i_1  i_2 , j_1 j_2}(a_0^\dagger)^3 a_{i_1}^\dagger a_{i_2}^\dagger a_{j_1}a_{j_2} a_0^3,
    \end{align}
   with $(G')_{i_1.. i_3, j_1..j_3}:=\sum_k G_{i_1.. i_3 k, j_1..j_3 k}$ and $(G'')_{i_1 i_2, j_1 j_2}:=\sum_{k_1 k_2} G_{i_1 i_2 k_1 k_2, j_1 j_2 k_1 k_2}$. In the following let us study the term involving $G''$, which is responsible for the largest contribution. Since $(V_N)_{i_1 i_2 (-k_1-k_2),j_1 j_2 (-k_1-k_2)}=(V_N)_{i_1 i_2 0,j_1 j_2 0}$, see Eq.~(\ref{Eq:Coefficients_in_Fourier}), we obtain
   \begin{align*}
     & \sum_{i_1  i_2, j_1 j_2} \!   \! \!  \!  (G'')_{i_1  i_2 , j_1 j_2}(a_0^\dagger)^3 a_{i_1}^\dagger a_{i_2}^\dagger a_{j_1}a_{j_2} a_0^3 \! = \!  \! \left(\sum_{k_1 k_2}|\eta_{k_1 k_2 (-k_1-k_2)}|^2 \! \right) \!  \!  \! \sum_{i_1 i_2,j_1 j_2} \!  \!  \!  \! (V_N)_{i_1 i_2 0,j_1 j_2 0}(a_0^\dagger)^3 a_{i_1}^\dagger a_{i_2}^\dagger a_{j_1}a_{j_2} a_0^3\\
      & \ \  \ \  \  \lesssim  N^{-3}\sum_{i_1 i_2,j_1 j_2} \!  \!  \!  \! (V_N)_{i_1 i_2 0,j_1 j_2 0}(a_0^\dagger)^3 a_{i_1}^\dagger a_{i_2}^\dagger a_{j_1}a_{j_2} a_0^3 \lesssim N^{-3} (a_0^\dagger)^3  (\mathcal{N}+1)^2 a_0^3\leq (\mathcal{N}+1)^2,
   \end{align*}
   where we have used 
   \begin{align*}
       & \sum_{i_1 i_2,j_1 j_2} \!   (V_N)_{i_1 i_2 0,j_1 j_2 0} a_{i_1}^\dagger a_{i_2}^\dagger a_{j_1}a_{j_2}\lesssim (\mathcal{N}+1)^2,\\
       & \ \ \ \ \sum_{k_1 k_2}|\eta_{k_1 k_2 (-k_1-k_2)}|^2\lesssim N^{-3},
   \end{align*}
 see Lemma \ref{Lem:Coefficient_Control}. Proceeding similarly for the other terms in Eq.~(\ref{Eq:First_Order_Upper_Bound_Useful_Aux_II}), concludes the proof of Eq.~(\ref{Eq:First_Order_Upper_Bound_Useful_II}). Regarding the bound on $\mathcal{E}_P\! \left(\big[a_{i_1} a_{i_2} a_{i_3} ,\mathcal{G}^\dagger \big]\right)$, let us identify 
   \begin{align*}
      & \big[a_{i_1} a_{i_2} a_{i_3} ,\mathcal{G}^\dagger \big]  =\bigg\{\frac{3}{2}\mathds{1}(i_1 \! = \! 0) (a_0^\dagger )^2 a_{i_2} a_{i_3}\, \mathbb{A}+ 3\mathds{1}(i_1=i_2=0) a_0^\dagger a_{i_3}\, \mathbb{A}\\
       & \ \ \ \ \ \ \ + \mathds{1}(i_1=i_2=i_3=0)\,  \mathbb{A}\bigg\}+ \{\mathrm{Permutations}\},
  \end{align*}
  where $\mathbb{A}:=\frac{1}{6}\sum_{ijk}\eta_{ijk}a_i a_j a_k$. Due to the sign $V_N\geq 0$ and the permutations symmetry of $V_N$, as well as to the fact that there are $6$ permutations of the set $\{1,2,3\}$, we can bound the operator $ \mathcal{E}_P\! \left(\big[a_{i_1} a_{i_2} a_{i_3} ,\mathcal{G}^\dagger \big]\right)$ from above by
  \begin{align*}
   &  6\, \mathcal{E}_P\! \left(\frac{3}{2}\mathds{1}(i_1 \! = \! 0)(a_0^\dagger )^2 a_{i_2} a_{i_3}\, \mathbb{A}+ 3\mathds{1}(i_1=i_2=0) a_0^\dagger a_{i_3}\, \mathbb{A} + \mathds{1}(i_1=i_2=i_3=0) \, \mathbb{A}\right)\\
    & \ \ \ \ \ \ \ \ \ \ \ \ \leq 18\, \mathcal{E}_P\! \left(\frac{3}{2}\mathds{1}(i_1 \! = \! 0)(a_0^\dagger )^2 a_{i_2} a_{i_3}\, \mathbb{A}\right) + 18 \, \mathcal{E}_P\! \left(3\mathds{1}(i_1=i_2=0)  a_0^\dagger a_{i_3}\, \mathbb{A}\right)\\
    & \ \ \ \ \ \ \ \ \ \ \ \ \ \ \ \ \ \ \ \ \ \ \ \ \ \ \ \ + 18\, \mathcal{E}_P\! \left(  \mathds{1}(i_1=i_2=i_3=0) \, \mathbb{A}\right).
  \end{align*}
  In the following we focus on $\mathcal{E}_P\! \left(\mathds{1}(i_1 \! = \! 0)(a_0^\dagger )^2 a_{i_2} a_{i_3}\, \mathbb{A}\right)$, the other terms can be treated in a similar fashion. By Lemma \ref{Lem:Coefficient_Control} we have $\mathbb{A}^\dagger \mathbb{A}\lesssim N^{-3}(\mathcal{N}+1)^3$, and therefore
  \begin{align*}
      & \ \ \ \  \ \ \ \  \ \ \ \  \mathcal{E}_P\! \left(\mathds{1}(i_1 \! = \! 0)(a_0^\dagger )^2 a_{i_2} a_{i_3}\, \mathbb{A}\right) = \! \mathbb{A}^\dagger \!  \!    \!  \!  \!  \!   \!  \!  \! \sum_{(ij),(\ell m )\in A^0} \!  \! \!  \!   \!  \!    \!  \! (V_N)_{0jk,0 m n}a_k^\dagger a_j^\dagger a_0^2 (a_0^\dagger )^2  a_{m} a_{n} \mathbb{A}\\
     & \ \ \ \lesssim  a_0^2 (a_0^\dagger )^2 \mathbb{A}^\dagger  (\mathcal{N}+1)^2 \mathbb{A} \! = \! a_0^2(a_0^\dagger )^2 (\mathcal{N} \! + \! 4)\mathbb{A}^\dagger   \mathbb{A}(\mathcal{N} \! + \! 4)  \! \lesssim  \! a_0^2(a_0^\dagger )^2 N^{-3}(\mathcal{N} \! + \! 1)^5 \! \leq \!  N^{-1}(\mathcal{N} \! + \! 1)^5,
  \end{align*}
  where $A^0$ contains all pairs $(jk)$ such that $(0jk)\in A$.
\end{proof}

As a consequence of Lemma \ref{Lem:First_Order_Upper_Bound_Useful}, we obtain that the trial state $\Gamma$ defined below Eq.~(\ref{Eq:Gamma_0}) has a potential energy $\mathcal{E}_\mathcal{P}(\psi)$ of the order $O_{N\rightarrow \infty}\! \left(\sqrt{N}\right)$, see the following Corollary \ref{Cor:First_Order_Trial_State_V}.

\begin{cor}
\label{Cor:First_Order_Trial_State_V}
 There exists a constant $C>0$, such that $ \left\langle \Gamma, \mathcal{E}_\mathcal{P}(\psi+\delta_1 \psi )\Gamma\right\rangle\leq C$ and
    \begin{align*}
        \left\langle \Gamma, \mathcal{E}_\mathcal{P}(\psi)\Gamma\right\rangle\leq C\sqrt{N}.
    \end{align*}
\end{cor}
\begin{proof}
   Recall that we can express the transformed quantity $U^{-1}\psi_{ijk} U$ by Eq.~(\ref{Eq:Double_Duhamel}) as
    \begin{align*}
        U^{-1}\psi_{ijk} U \! & =\! a_i a_j a_k  +   \!  \! \int_0^1    \! \! \!  \! U_{-s}   \big[a_i a_j a_k ,\mathcal{G}^\dagger \big]   U_s \mathrm{d}s  -  \!  \! \int_0^1    \! \! \!  \! U_{-s}   (\delta_1 \psi)_{ijk}   U_s \mathrm{d}s -  \!  \! \int_0^1    \! \! \!  \! U_{-s}   (\delta_2 \psi)_{ijk}   U_s \mathrm{d}s\\
        & \ \ \ \ \ \ \ \ \ \ +\int_0^1\int_s^1 U_{-t} \eta_{ijk} \big[a_0^3 ,\mathcal{G}^\dagger \big]U_t\mathrm{d}t\mathrm{d}s,
    \end{align*}
    where we have used Duhamel's formula to express $U^{-1}\eta_{ijk}a_0^3U-U_{-s}\eta_{ijk}a_0^3U_s$. Using the sign $V_N\geq 0$ and Lemma \ref{Lem:First_Order_Upper_Bound_Useful}, we estimate using the Cauchy-Schwarz inequality
    \begin{align}
\nonumber
       \mathcal{E}_\mathcal{P}\! \left(\int_0^1    \! \! \!  \! U_{-s}   (\delta_1 \psi)_{ijk}   U_s \mathrm{d}s\right)     \leq \!  \! \int_0^1   \! \! \! U_{-s}\mathcal{E}_\mathcal{P}(\delta_1 \psi)U_s\,  \mathrm{d}s  \! \leq  \! CN^{\frac{1}{2}} \!  \! \int_0^1 \!  \!  U_{-s}(\mathcal{N} \! + \! 1)^4 U_s  \mathrm{d}s \! \leq  \! C'N^{\frac{1}{2}}(\mathcal{N} \! + \! 1),\\
       \label{Eq:Upper_Bound_CS_I}
           \mathcal{E}_\mathcal{P}\! \left(\int_0^1    \! \! \!  \! U_{-s}   (\delta_2 \psi)_{ijk}   U_s \mathrm{d}s\right)     \leq \!  \! \int_0^1   \! \! \! U_{-s}\mathcal{E}_\mathcal{P}(\delta_2 \psi)U_s\,  \mathrm{d}s  \! \leq  \! CN^{\frac{1}{2}} \!  \! \int_0^1 \!  \!  U_{-s}(\mathcal{N} \! + \! 1)^4 U_s  \mathrm{d}s \! \leq  \! C'(\mathcal{N} \! + \! 1)^4.
    \end{align}
    for suitable $C,C'$, where we utilize Lemma \ref{Lem:Particle_Number_Creation} in order to estimate $U_{-s}(\mathcal{N}+1)^4 U_s$. Similarly
    \begin{align}
    \label{Eq:Upper_Bound_CS_II}
        \mathcal{E}_\mathcal{P}\! \left(\int_0^1    \! \! \!  \! U_{-s}   \big[a_i a_j a_k ,\mathcal{G}^\dagger \big]  U_s \mathrm{d}s\right)\leq C'N^{-\frac{3}{2}}(\mathcal{N} \! + \! 1)^5
    \end{align}
    follows from Lemma \ref{Lem:First_Order_Upper_Bound_Useful}. Regarding the term in the last line of Eq.~(\ref{Eq:Double_Duhamel}), we note that 
    \begin{align*}
        \big[a_0^3 ,\mathcal{G}^\dagger \big]^\dagger \big[a_0^3 ,\mathcal{G}^\dagger \big]\lesssim N(\mathcal{N}+1)^3
    \end{align*}
follows from an analogous argument as we have seen in the proof of Lemma \ref{Lem:First_Order_Upper_Bound_Useful} and
\begin{align*}
   \sum_{(i j k),(\ell m n)\in A} \! \left( V_N\right)_{i j k,\ell m n} \overline{\eta_{ijk}}\eta_{\ell m n}\lesssim N^{-2} 
\end{align*}
by Lemma \ref{Lem:Coefficient_Control}. Therefore
    \begin{align}
     \label{Eq:Upper_Bound_CS_III}
       & \mathcal{E}_\mathcal{P}\!\left(\int_0^1\int_s^1 U_{-t} \eta_{ijk} \big[a_0^3 ,\mathcal{G}^\dagger \big]U_t\mathrm{d}t\mathrm{d}s\right)\leq \frac{1}{2}\int_0^1\int_s^1 \! \! \! U_{-t}\mathcal{E}_\mathcal{P}\! \left(\eta_{ijk} \big[a_0^3 ,\mathcal{G}^\dagger \big]\right) \! U_t\, \mathrm{d}t\mathrm{d}s\lesssim N^{-1}(\mathcal{N}+1)^3,
    \end{align}
    where we have used Lemma \ref{Lem:Particle_Number_Creation} again. Using $a_i a_j a_k \Gamma_0=0$ in case $(ijk)\in A$, we obtain by Eq.~(\ref{Eq:Double_Duhamel}) together with Eq.~(\ref{Eq:Upper_Bound_CS_I}), Eq.~(\ref{Eq:Upper_Bound_CS_II}) and  Eq.~(\ref{Eq:Upper_Bound_CS_III}) for a suitable constant $C$
    \begin{align*}
        & \left\langle \Gamma, \mathcal{E}_\mathcal{P}\! \left(\psi\right)\Gamma\right\rangle=\left\langle \Gamma_0, \mathcal{E}_\mathcal{P}\! \left(U^{-1}\psi U\right)\Gamma_0\right\rangle \leq CN^{\frac{1}{2}}\left\langle \Gamma_0, (\mathcal{N}+1)^5 \Gamma_0\right\rangle=CN^{\frac{1}{2}}.
    \end{align*}
    Analogously we obtain $ \left\langle \Gamma, \mathcal{E}_\mathcal{P}(\psi+\delta_1 \psi )\Gamma\right\rangle\leq C$.
\end{proof}

Regarding the variable $c_k=a_k+[a_k,\mathcal{G}]$ from Eq.~(\ref{Eq:c_for_K=0}), let us apply Duahmel's formula 
\begin{align}
\label{Eq:Decomposition_c}
    U^{-1}c_k U & =a_k-\int_0^1 U_{-s}[a_k,\mathcal{G}]U_s \, \mathrm{d}s+U^{-1}[a_k,\mathcal{G}]U=a_k+\int_0^1 \! \int_s^1 U_{-t} \left[[a_k,\mathcal{G}],\mathcal{G}^\dagger \right]U_t \, \mathrm{d}t\mathrm{d}s,
\end{align}
where we have used $\left[a_k,\mathcal{G}^\dagger\right]=0$ for $k\neq 0$ and $[[a_k,\mathcal{G}],\mathcal{G}]=0$, which follows from the observation that $\eta_{ijk}=0$ in case one of the indices in $\{i,j,k\}$ is zero. The following Lemma \ref{Lem:First_Order_Upper_Bound_Useful_II} provides useful estimates on the quantity $\left[[a_k,\mathcal{G}],\mathcal{G}^\dagger \right]$. In order to formulate Lemma \ref{Lem:First_Order_Upper_Bound_Useful_II} let us define the kinetic energy of an operator valued one particle vector $\Theta_k$, written as $\mathcal{E}_\mathcal{K}(\Theta)$ or $\mathcal{E}_\mathcal{K}(\Theta_k)$ with $k$ being a dummy index, as
\begin{align}
\label{Eq:Kinetic_Energy_Operator}
    \mathcal{E}_\mathcal{K}(\Theta):=\sum_k |k|^2 \Theta_k^\dagger \Theta_k.
\end{align}

\begin{lem}
\label{Lem:First_Order_Upper_Bound_Useful_II}
For $m\geq 0$ there exists a constant $C_m>0$, such that 
\begin{align*}
  \mathcal{E}_\mathcal{K}\! \left(\mathcal{N}^m \left[[a_k,\mathcal{G}],\mathcal{G}^\dagger \right]\right)\leq C_m N^{-1}(\mathcal{N}+1)^{5+2m}  .
\end{align*}
\end{lem}
\begin{proof}
Let us write the double commutator as $\left[[a_k,\mathcal{G}],\mathcal{G}^\dagger \right]= (\delta_1 c)_k+(\delta_2 c)_k+(\delta_3 c)_k$, where
    \begin{align*}
    (\delta_1 c)_k : & =(a_0^\dagger)^3 a_0^3 \sum_{ij}|\eta_{ijk}|^2 a_k,\\
    (\delta_2 c)_k: & =(a_0^\dagger)^3 a_0^3 \sum_{ij,j' k'}\overline{\eta_{i j' k'}}\eta_{ijk} a_j^\dagger a_{j'}a_{k'},\\
        (\delta_3 c)_k: & =\left[a_0^3, (a_0^\dagger)^3\right] \! \left(\sum_{ij}\eta_{ijk}a_j^\dagger a_k^\dagger\right) \! \left(\sum_{i'j'k'}\overline{\eta_{i'j'k'}}a_{i'} a_{j'}a_{k'}\right).
    \end{align*}
    By Lemma \ref{Lem:Coefficient_Control} it is clear that 
    \begin{align*}
     \sum_{ij}|\eta_{ijk}|^2\lesssim \frac{1}{N^4}\sum_{t}\frac{1}{(|k|^2+|t|^2)^2}\lesssim \frac{1}{N^4 |k|},   
    \end{align*}
and therefore
    \begin{align*}
        \mathcal{E}_\mathcal{K} \! \left(\mathcal{N}^m \delta_1 c\right)=\sum_k |k|^2 (\delta_1 c)_k^\dagger \mathcal{N}^{2m} (\delta_1 c)_k\lesssim \frac{1}{N^8}\sum_{k\neq 0}a_k^\dagger \left((a_0^\dagger)^3 a_0^3\right)^2  \mathcal{N}^{2m}  a_k\leq \frac{1}{N^2}\mathcal{N}^{2m+1}.
    \end{align*}
   Using $J_{p_1 p_2 p_3,p_1' p_2' p_3'}:=\sum_{q q' k}|k|^2\overline{\eta_{q' p_2' p_3'}\eta_{q p_1' k}}\eta_{q' p_1 k}\eta_{q p_2 p_3}$ and $\widetilde{J}_{p_2 p_3,p_2' p_3'}:=\sum_{p_1}J_{p_1 p_2 p_3,p_1 p_2' p_3'}$
    \begin{align*}
       \mathcal{E}_\mathcal{K} \! \left(\mathcal{N}^{m} \delta_2 c\right) & =\left((a_0^\dagger)^3 a_0^3\right)^2\sum_{j p' n' , p j' k'}J_{j p' n' , p j' k'} \, a_j^\dagger a_{p'}^\dagger a_{n'}^\dagger (\mathcal{N}+2)^{2m} a_p a_{j'}a_{k'}\\
       & \ \  \ \  \ + \left((a_0^\dagger)^3 a_0^3\right)^2\sum_{ p' n' , j'  k'}\widetilde{J}_{p' n' , j' k'} \, a_{p'}^\dagger a_{n'}^\dagger  (\mathcal{N}+1)^{2m} a_{j'}a_{k'}.
    \end{align*}
    Utilizing the operator $X_{jk,j'k'}:=\sum_q |k|\eta_{q j k}\overline{\eta_{q j' k'}}$ acting on $L^2(\Lambda)^{\otimes 2}$ and the permutation operator $(S\Psi)(x_1,x_2,x_3):=\Psi(x_2,x_1,x_3)$ acting on $L^2(\Lambda)^{\otimes 3}$, we can write 
    \begin{align*}
    J=(1\otimes X^\dagger )S (1\otimes X),    
    \end{align*}
 and $\widetilde{J}=X^\dagger X$. Consequently
\begin{align*}
    \|J\|\leq \|S\|\, \|1\otimes X\|^2=\|X\|^2=\|\widetilde{J}\|.
\end{align*}
    By Lemma \ref{Lem:Coefficient_Control} we have $\|\widetilde{J}\|\leq C N^{-\frac{15}{2}}$ for a suitable constant $C$. Consequently 
    \begin{align*}
        \mathcal{E}_\mathcal{K} \! \left(\mathcal{N}^{m} \delta_2 c\right)\leq CN^{-\frac{15}{2}}\left((a_0^\dagger)^3 a_0^3\right)^2(\mathcal{N}+2)^{3+2m}\leq C N^{-\frac{3}{2}}(\mathcal{N}+2)^{3+2m}.
    \end{align*}
   Similarly one can show that $\mathcal{E}_\mathcal{K} \! \left(\mathcal{N}^{m} \delta_3 c\right)\lesssim N^{-1}(\mathcal{N}+1)^{5+2m}$, and therefore
   \begin{align*}
     &  \ \ \ \ \ \ \ \ \mathcal{E}_\mathcal{K}\! \left(\mathcal{N}^m \left[[a_k,\mathcal{G}],\mathcal{G}^\dagger \right]\right)=\mathcal{E}_\mathcal{K}\! \left(\mathcal{N}^m \delta_1 + \mathcal{N}^m \delta_2 + \mathcal{N}^m \delta_3\right)\\
       & \leq 3\mathcal{E}_\mathcal{K}\! \left(\mathcal{N}^m \delta_1\right)+3\mathcal{E}_\mathcal{K}\! \left(\mathcal{N}^m \delta_2\right)+3\mathcal{E}_\mathcal{K}\! \left(\mathcal{N}^m \delta_3\right)\lesssim C_m N^{-1}(\mathcal{N}+1)^{5+2m}.
   \end{align*}
\end{proof}

As a consequence of Lemma \ref{Lem:First_Order_Upper_Bound_Useful_II}, we obtain that the trial state $\Gamma$ defined below Eq.~(\ref{Eq:Gamma_0}) has a kinetic energy $\mathcal{E}_\mathcal{K}(c)$ of the order $O_{N\rightarrow \infty}\! \left(1\right)$ in the subsequent Corollary \ref{Cor:First_Order_Trial_State_Kinetic}. Since in the residual term $\mathcal{E}$ defined in Lemma \ref{Lem:Many_Body_Block_Diagonal} the term $\mathcal{E}_\mathcal{K}\! \left(\sqrt{\mathcal{N}}c\right)\leq \frac{1}{2}\mathcal{E}_\mathcal{K}\! \left(c\right)+\frac{1}{2}\mathcal{E}_\mathcal{K}\! \left(\mathcal{N} c\right)$ appears, it will be convenient to estimate the expectation value in the state $\Gamma$ of $\mathcal{E}_\mathcal{K}\! \left(\mathcal{N}^m c\right)$ for $m\geq 1$ as well.

\begin{cor}
\label{Cor:First_Order_Trial_State_Kinetic}
   Lt $\Gamma$ be the state defined below Eq.~(\ref{Eq:Gamma_0}) and $m\geq 0$. Then there exists a constant $C>0$, such that $\left\langle \Gamma, \! \mathcal{E}_\mathcal{K}\! \left(\mathcal{N}^m c\right)\Gamma\right\rangle\leq \frac{C_m}{N}$.
\end{cor}
\begin{proof}
By Lemma \ref{Lem:Particle_Number_Creation} we have 
\begin{align*}
  U^{-1}\mathcal{N}^{2m}U=(U^{-1}\mathcal{N}^{m}U)^\dagger  U^{-1}\mathcal{N}^{m}U\lesssim (\mathcal{N}^{m}+1)^2  
\end{align*}
and hence
\begin{align*}
    U^{-1}\mathcal{E}_\mathcal{K}\! \left(\mathcal{N}^m c\right)U=\mathcal{E}_\mathcal{K}\! \left(U^{-1}\mathcal{N}^m U\, U^{-1} c\, U\right)\lesssim \mathcal{E}_\mathcal{K}\! \left((\mathcal{N}^m+1) U^{-1} c\, U\right),
\end{align*}
where we have used that for operators $f_k$ and $A,B$ satisfying $A^\dagger A\leq C B^\dagger B$ we have 
\begin{align*}
    \mathcal{E}_\mathcal{K}\! \left(A f_k\right)\leq C\mathcal{E}_\mathcal{K}\! \left(B f_k\right).
\end{align*}
Proceeding as in the proof of Corollary \ref{Cor:First_Order_Trial_State_V}, we obtain by Eq.~(\ref{Eq:Decomposition_c}) and Lemma \ref{Lem:First_Order_Upper_Bound_Useful_II} 
\begin{align*}
   &  \mathcal{E}_\mathcal{K}\! \left((\mathcal{N}^m+1) U^{-1} c\, U\right)\lesssim \mathcal{E}_\mathcal{K}\! \left((\mathcal{N}^m+1) a\right)+\int_0^1 \mathcal{E}_\mathcal{K}\! \left((\mathcal{N}^m+1) U_{-t} \left[[a_k,\mathcal{G}],\mathcal{G}^\dagger \right] U_t\right)\mathrm{d}t\\
   & \ \ \ \ \  =\mathcal{E}_\mathcal{K}\! \left((\mathcal{N}^m+1) a\right)+\int_0^1  U_{-t}\mathcal{E}_\mathcal{K}\! \left(U_t(\mathcal{N}^m+1)U_{t} \left[[a_k,\mathcal{G}],\mathcal{G}^\dagger \right] \right)U_t\mathrm{d}t\\
    & \ \ \ \ \  \lesssim \mathcal{E}_\mathcal{K}\! \left((\mathcal{N}^m+1) a\right)+\int_0^1  U_{-t}\mathcal{E}_\mathcal{K}\! \left((\mathcal{N}^m+1) \left[[a_k,\mathcal{G}],\mathcal{G}^\dagger \right] \right)U_t\mathrm{d}t\\
        & \ \ \ \ \  \lesssim \mathcal{E}_\mathcal{K}\! \left((\mathcal{N}^m+1) a\right)+ N^{-1}\int_0^1  U_{-t}(\mathcal{N}+1)^{5+2m}  U_t\mathrm{d}t\\
        & \ \ \ \ \  \lesssim \mathcal{E}_\mathcal{K}\! \left((\mathcal{N}^m+1) a\right)+ N^{-1}(\mathcal{N}+1)^{5+2m}  .
\end{align*}
where we have made use of Lemma \ref{Lem:Particle_Number_Creation} again. Using $a_k\Gamma_0=0$ for $k\neq 0$ therefore yields
\begin{align*}
&  \ \ \ \ \ \ \ \left\langle \Gamma, \! \mathcal{E}_\mathcal{K}\! \left(\mathcal{N}^m c\right)\Gamma\right\rangle=\left\langle \Gamma_0, \!  U^{-1}\mathcal{E}_\mathcal{K}\! \left(\mathcal{N}^m c\right)U \Gamma_0\right\rangle\\
& \lesssim \left\langle \Gamma_0, \!  \mathcal{E}_\mathcal{K}\! \left((\mathcal{N}^m+1) a\right) \Gamma_0\right\rangle +N^{-1} \left\langle \Gamma_0, \!  (\mathcal{N}+1)^{5+2m} \Gamma_0\right\rangle =\frac{1}{N}.
\end{align*}
\end{proof}

Having Corollary \ref{Cor:First_Order_Trial_State_V} and Corollary \ref{Cor:First_Order_Trial_State_Kinetic} at hand, we are in a position to verify the upper bound on the ground state energy $E_N$ in Theorem \ref{Th:First_Order_Upper_Bound}.

\begin{proof}[Proof of Theorem \ref{Th:First_Order_Upper_Bound}]
Let $A:=(2\pi \mathbb{Z})^9 \setminus \{(0,0,0)\}$, and let $\Gamma$ be the state defined below Eq.~(\ref{Eq:Gamma_0}). Using Eq.~(\ref{Eq:First_Algebraic_Representation}) and Eq.~(\ref{Eq:Definition_tilde_V}), and the fact that $\left(\widetilde{V}_N\right)_{i j k,\ell m n}=\left(V_N\right)_{i j k,\ell m n}$ for index triples $(ijk),(\ell m n)\in A$, we obtain
    \begin{align*}
       H_N & =\sum_{k}|k|^2  c_k^\dagger c_k  +\lambda_{0,0} (a_0^\dagger)^3 a_0^3 +  \frac{1}{6}\sum_{(i j k),(\ell m n)\in A}\left(V_N\right)_{i j k,\ell m n}\psi_{i j k}^\dagger \psi_{\ell m n}\\
       & \ \ \ \ \ \ +\left(3 a_0^\dagger a_0^3\sum_{\ell\neq 0}\lambda_{0,\ell}\,  a_{\ell}^\dagger a_{-\ell}^\dagger + \mathrm{H.c.}\right)-\mathcal{E}\\
       & =\lambda_{0,0} (a_0^\dagger)^3 a_0^3 +\mathcal{E}_\mathcal{K}\! \left(c\right) +  \mathcal{E}_\mathcal{P}(\psi)+\left(3 a_0^\dagger a_0^3\sum_{\ell\neq 0}\lambda_{0,\ell}\,  a_{\ell}^\dagger a_{-\ell}^\dagger + \mathrm{H.c.}\right)-\mathcal{E}.
    \end{align*}
    By a symmetry argument it is clear that $\braket{\Gamma,a_0^\dagger a_0^3a_{\ell}^\dagger a_{-\ell}^\dagger\Gamma}=0$. Applying Corollary \ref{Cor:First_Order_Trial_State_V} as well as Corollary \ref{Cor:First_Order_Trial_State_Kinetic}, with $m=0$, yields for suitable constants $C>0$
    \begin{align*}
       &  \left\langle \Gamma, \mathcal{E}_\mathcal{P}(\psi)\Gamma\right\rangle\leq C\sqrt{N}, \ \ \ \ \ \ \ \ \ \ \left\langle \Gamma, \! \mathcal{E}_\mathcal{K}\! \left(c\right)\Gamma\right\rangle\leq \frac{C}{N}.
    \end{align*}
Furthermore, observe that $N^3\lambda_{0,0}\leq \frac{1}{6}b_\mathcal{M}(V)N +C'$ by Lemma \ref{Lem:Coefficient_Control} for a suitable $C'$. In order to estimate the final term $\braket{\Psi,\mathcal{E}\Psi}$, note that we have by Lemma \ref{Lem:Coefficient_Control} for $m\in \mathbb{N}$\begin{align}
\label{Eq:Gamma_Particle_Control}
    \braket{\Gamma,\mathcal{N}^m \Gamma}=\braket{\Gamma_0, U^{-1}\mathcal{N}^m U \Gamma_0}\lesssim \braket{\Gamma_0, (\mathcal{N}+1)^m \Gamma_0}=1.
\end{align}
    Using Lemma \ref{Lem:First_Error_Term_Estimate} together with the estimate from Corollary \ref{Cor:First_Order_Trial_State_Kinetic} for $m=0$ and $m=1$, we therefore obtain $ \left|\braket{\Psi,\mathcal{E}\Psi}\right|\lesssim N^{-\frac{1}{3}}$.
\end{proof}

\section{Refined Correlation Structure}
\label{Section:Second_Order_Lower_Bound}
Utilizing the set of operators defined in Eq.~(\ref{Eq:Definition_c_variable}) and Eq.~(\ref{Eq:Definition_psi}), we where able to identify the ground state energy $E_N$ up to errors of the magnitude $O_{N\rightarrow \infty}(\sqrt{N})$ in the previous Sections \ref{Sec:Condensation with a Rate} and \ref{Sec:First_Order_Upper_Bound}. It is the purpose of this Section to obtain a higher resolution of the energy, which especially captures the subleading term proportional to $\sqrt{N}$ in the asymptotic expansion of $E_N$, using a more refined correlation structure compared to the one introduced in Subsection \ref{Subsec:The three-body Problem}. On a technical level, the new correlation structure is implemented by the new set of operators $d_k$ and $\xi_{ijk}$ defined below in Eq.~(\ref{Eq:Second_Family_of_Variables}) and Eq.~(\ref{Eq:Second_Family_of_Variables_xi}), which constitute a refined version of the operators $c_k$ and $\psi_{ijk}$ respectively. Writing the operator $H_N$ in terms of $d_k$ and $\xi_{ijk}$ will then allow us to verify the lower bound from Theorem \ref{Eq:Main_Theorem_Introduction} in Subsection \ref{Subsection:Proof_of_Theorem_Main_in_Text} and the corresponding upper bound in Section \ref{Sec:Second_Order_Upper_Bound}. \\

The approach presented in Sections \ref{Sec:Condensation with a Rate} and \ref{Sec:First_Order_Upper_Bound} fails to capture the correct term of order $\sqrt{N}$ for two reasons: (I) The following expression appearing in Eq.~(\ref{Eq:First_Lower_Bound})
\begin{align}
\label{Eq:Two-particle_Extra_Correlations}
    3  \! \sum_{|\ell|>K} \!  \! \lambda_{0,\ell}\,  a_{\ell}^\dagger a_{-\ell}^\dagger\, a_0^\dagger a_0^3
\end{align}
is expected to lower the ground state energy by an amount proportional to $\sqrt{N}$, which is consistent with our estimate in Lemma \ref{Lem:Rest_Estimate}. (II) In the pursue of an upper bound on $E_N$ we expressed the unitary conjugated variables $U^{-1} \psi_{ijk} U$ as a sum of $a_i a_j a_k$ and various error terms according to Eq.~(\ref{Eq:Double_Duhamel}) as
\begin{align}
\nonumber
     U^{-1}\psi_{ijk} U \! & =\! a_i a_j a_k  +   \!  \! \int_0^1    \! \! \!  \! U_{-s}   \big[a_i a_j a_k ,\mathcal{G}^\dagger \big]   U_s \mathrm{d}s  -  \!  \! \int_0^1    \! \! \!  \! U_{-s}   (\delta_1 \psi)_{ijk}   U_s \mathrm{d}s -  \!  \! \int_0^1    \! \! \!  \! U_{-s}   (\delta_2 \psi)_{ijk}   U_s \mathrm{d}s\\
     \label{Eq:Duhamel_I_Copy}
        & \ \ \ \ \ \ \ \ \ \ +\int_0^1\int_s^1 U_{-t} \eta_{ijk} \big[a_0^3 ,\mathcal{G}^\dagger \big]U_t\mathrm{d}t\mathrm{d}s.
\end{align}
While most of the terms appearing in Eq.~(\ref{Eq:Duhamel_I_Copy}) give a contribution of the order $o_{N\rightarrow}\! \left(\sqrt N \right)$, the term $\delta_1 \psi$ is expected to increase the ground state energy by an amount proportional to $\sqrt{N}$, which is consistent with our estimate in Eq.~(\ref{Eq:First_Order_Upper_Bound_Useful_I}). In order to extract the energy shift due to the expression in Eq.~(\ref{Eq:Two-particle_Extra_Correlations}), we follow the strategy in Subsection \ref{Subsec:The three-body Problem} and introduce an additional two-particle correlation structure via a map acting on the two-particle space 
\begin{align*}
 T_2:L^2(\Lambda^2)\longrightarrow L^2(\Lambda^2)   
\end{align*}
in Eq.~(\ref{Eq:Definition_T_2}), which will give rise to the negative energy correction $-\mu(V)\sqrt{N}$ from Theorem \ref{Eq:Main_Theorem_Introduction}. Regarding the energy shift associated with $\delta_1 \psi$, it is a natural idea to include this term in the definition of our new operators $\xi_{ijk}$, giving rise to the positive energy correction $\gamma(V)\sqrt{N}$ from Theorem \ref{Eq:Main_Theorem_Introduction}. However a computation in Eq.~(\ref{Eq:Representation_Of_Second_Order_Transformations}) demonstrates that the presence of $\delta_1 \psi$ produces new four-particle correlation terms of the form 
\begin{align*}
    a_u^\dagger a_i^\dagger a_j^\dagger a_k^\dagger a_0^4 +\mathrm{H.c.},
\end{align*}
with $\{u,i,j,k\}$ all different from zero, which are expected to lower the ground state energy by an amount proportional to $\sqrt{N}$. Again we extract the correlation energy by introducing a map, acting this time on the four-particle space
\begin{align*}
 T_4:L^2(\Lambda^4) \longrightarrow L^2(\Lambda^4)   
\end{align*}
in Eq.~(\ref{Eq:Definition_T_4}), which gives rise to the negative energy correction $-\sigma(V)\sqrt{N}$ from Theorem \ref{Eq:Main_Theorem_Introduction}.

In the following let $T:L^2(\Lambda^3)\longrightarrow L^2(\Lambda^3)$ be the map constructed in Eq.~(\ref{Eq:Definition_Feshbach-Schur}), and for now let $T_2:L^2(\Lambda^2)\longrightarrow L^2(\Lambda^2)$ and $T_4:L^2(\Lambda^4)\longrightarrow L^2(\Lambda^4)$ be generic bounded permutation symmetric operators modelling the two-particle and four-particle correlation structure respectively. Following the approach in Section \ref{Sec:Condensation with a Rate}, we are implementing many-particle counterparts to the transformations $T$, $T_2$ and $T_4$ as
\begin{align}
\label{Eq:Second_Family_of_Variables}
    d_k:&=a_k+\sum_{j,m n}(T_2-1)_{jk,m n} \,  a_j^\dagger a_{m}a_{n}+\frac{1}{2}\sum_{i j,\ell m n}(T-1)_{ijk,\ell m n} \, a_i^\dagger a_j^\dagger a_{\ell}a_{m}a_{n}\\
    \nonumber
    &\ \ \ \ +\frac{1}{6}\sum_{u i j,v \ell m n}(T_4-1)_{u ijk,v \ell m n} \, a_u^\dagger a_i^\dagger a_j^\dagger a_v a_{\ell}a_{m}a_{n},\\
    \label{Eq:Second_Family_of_Variables_xi}
\xi_{i j k}:&=\sum_{\ell m n}(T)_{ijk,\ell m n} \,  a_{\ell}a_{m}a_{n}+(\delta_1 \psi)_{ijk}+\sum_{u,v\ell m n}(T_4-1)_{uijk,v\ell m n} \,  a_u^\dagger a_v a_{\ell}a_{m}a_{n}.
\end{align}
Note that $T_2$ is not included in the definition of $\xi_{ijk}$, as it would only give contributions of the order $O_{N\rightarrow \infty}(1)$. Using the Laplace operator $\Delta_s$ acting on the space $L^2(\Lambda)^{\otimes s}$ and the coefficients
\begin{align*}
    (\chi)_{i_1..i_4,j_1..j_4}:  =\frac{1}{2}\sum_{\sigma\in S_3}\mathds{1}(j_1=..=j_4=i_{\sigma_3}=0)\eta_{i_{\sigma_1}i_{\sigma_2}i_4},
\end{align*}
let us furthermore define the operators $X_2: =T_2^\dagger (-\Delta_2)T_2+\Delta_2$ and 
\begin{align}
     \label{Eq:X_4}
   X_4:   & = \bigg(\left\{(-\Delta_4 +4( \widetilde{V}_N\otimes 1))(T_4-1)+ 4( \widetilde{V}_N\otimes 1)\chi\right\}+\mathrm{H.c.}\bigg) \\ 
     \nonumber
     & + (T_4-1)^\dagger (-\Delta_4)(T_4-1)+\left((T\otimes 1 - 1)^\dagger 4( \widetilde{V}_N\otimes 1)(T_4-1+\chi)+\mathrm{H.c.}\right)\\
     \nonumber
     & +(T_4-1+\chi)^\dagger 4( \widetilde{V}_N\otimes 1)(T_4-1+\chi).
\end{align}
A straightforward computation, similar to the one in Eq.~(\ref{Eq:First_Algebraic_Representation}), reveals that up to excess terms involving $X_2$, $X_4$ and an error term $\widetilde{\mathcal{E}}$, we can write the operator $H_N$ as a sum of squares in the variables $d_k$ and $\xi_{ijk}$ according to
 \begin{align}
 \nonumber
    & \ \ \ \ \  \  \ \ \ \  \ \ \ \ \  \  \ \ \ \ \sum_{k}|k|^2  d_k^\dagger d_k \!+ \! \frac{1}{6}\sum_{i j k,\ell m n}\left(\widetilde{V}_N\right)_{i j k,\ell m n}\xi_{i j k}^\dagger \xi_{\ell m n}\\
    \label{Eq:Representation_Of_Second_Order_Transformations}
    & \ =H_N+\frac{1}{2}\sum_{jk,mn} (X_2)_{jk,mn} \, a_j^\dagger a_k^\dagger a_m a_n+\frac{1}{24}\sum_{uijk,v \ell m n}(X_4)_{uijk,v \ell m n} \, a_u^\dagger a_i^\dagger a_j^\dagger a_k^\dagger a_v a_\ell a_m a_n+ \widetilde{\mathcal{E}},
 \end{align}
 where the error $\widetilde{\mathcal{E}}$ contains all the non-fully contracted products appearing in the squares 
 \begin{align}
 \label{Eq:Non-Fully_Contracted}
  &   \sum_{k}|k|^2  (d_k-a_k)^\dagger (d_k - a_k) ,  \frac{1}{6}\sum_{i j k,\ell m n}\left(\widetilde{V}_N\right)_{i j k,\ell m n}(\xi_{i j k}-\psi_{ijk})^\dagger (\xi_{\ell m n} -\psi_{ijk}).
 \end{align}
 In this context we define the fully contracted part of a product of monomials
 \begin{align*}
    \left(a_{i_1}^\dagger \dots a_{i_r}^\dagger a_{j_1}\dots a_{j_t}\right)\left(a_{i'_1}^\dagger \dots a_{i'_{r'}}^\dagger a_{j'_1}\dots a_{j'_{t'}}\right) 
 \end{align*}
as $C_{j_1\dots j_{t},i'_1\dots i'_{r'}} a_{i_1}^\dagger \dots a_{i_r}^\dagger a_{j'_1}\dots a_{j'_{t'}}$ with $C_{j_1\dots j_{t},i'_1\dots i'_{r'}}$ being the expectation of $a_{j_1}\dots a_{j_t}a_{i'_1}^\dagger \dots a_{i_{r'}}^\dagger$ in the vacuum. For a term by term definition of $\widetilde{\mathcal{E}}$ see Eq.~(\ref{Eq:Term-by-Term}) in Subsection \ref{Subsection:E_total}. 

In the following we want to choose $T_4$, such that the term $4( \widetilde{V}_N\otimes 1)\chi$ is cancelled in the expression $\{..\}$ from Eq.~(\ref{Eq:X_4}), at least after symmetrization and projection onto the range of $Q^{\otimes 4}$, i.e. we define 
\begin{align}
\label{Eq:Definition_T_4}
   T_4:=1-R_4 \Pi_{\mathrm{sym}}Q^{\otimes 4}4( \widetilde{V}_N\otimes 1)\chi=1-R_4 \Pi_{\mathrm{sym}}Q^{\otimes 4}4( V_N\otimes 1)\chi, 
\end{align}
where $\Pi_{\mathrm{sym}}$ is the orthogonal projection onto the subspace $L^2_{\mathrm{sym}}(\Lambda^4)\subseteq L^2(\Lambda^4)$ and $R_4$ is the pseudo-inverse of 
\begin{align}
\label{Eq:Def_R_4}
Q^{\otimes 4}\left(-\Delta_4 +4( \widetilde{V}_N\otimes 1)\right)Q^{\otimes 4}=Q^{\otimes 4}\left(-\Delta_4 +4( V_N\otimes 1)\right)Q^{\otimes 4}.    
\end{align}
In order to obtain an improved representation of the operator $X_4$ defined in Eq.~(\ref{Eq:X_4}), let us introduce the constants
\begin{align}
\label{Eq:sigma}
 \sigma_N: & =\frac{N^4}{6}\braket{(\widetilde{V}_N\otimes 1)\chi u_0^{\otimes 4},(1-T_4) u_0^{\otimes 4}},\\
 \nonumber
 &= \frac{N^4}{24} \left\langle (T_4-1) u_0^{\otimes 4},\left(-\Delta_4 +4 V_N\otimes 1\right)(T_4-1) u_0^{\otimes 4}\right\rangle, \\
 \label{Eq:gamma}
 \gamma_N: & =\frac{N^4}{6}\braket{( \widetilde{V}_N \otimes 1)\chi u_0^{\otimes 4},\chi u_0^{\otimes 4}}=\frac{N^4}{6}\braket{( V_N \otimes 1)\chi u_0^{\otimes 4},\chi u_0^{\otimes 4}},
\end{align}
which allow us to write  $\gamma_N-\sigma_N=\frac{N^4}{24}(X_4)_{0000,0000}$. Furthermore, we define the three-particle state $\Theta$ as
\begin{align}
\label{Eq:X_3_Definition}
    (\Theta)_{ijk}: & =4\left(\Pi_{\mathrm{sym}}\frac{X_4}{24} \Pi_{\mathrm{sym}}\right)_{0ijk,000},\\
    \nonumber
    (\Theta)_{0jk}: & =6\left(\Pi_{\mathrm{sym}}\frac{X_4}{24} \Pi_{\mathrm{sym}}\right)_{00jk,000}  
\end{align}
for $\{i,j,k\}$ all different from zero and $ (\Theta)_{ijk}:=0$ otherwise. According to the definition of $T_4$ we have $Q^{\otimes 4}\Pi_{\mathrm{sym}}X_4\Pi_{\mathrm{sym}} P^{\otimes 4}=0$, and therefore
\begin{align}
\nonumber
  & \frac{1}{24} \!  \! \sum_{uijk,v \ell m n} \!  \!  \!  \! (X_4)_{uijk,v \ell m n} \, a_u^\dagger a_i^\dagger a_j^\dagger a_k^\dagger a_v a_\ell a_m a_n \! = \! \frac{1}{24} \!  \! \sum_{uijk,v \ell m n} \!  \!  \!  \! (\Pi_{\mathrm{sym}}X_4 \Pi_{\mathrm{sym}})_{uijk,v \ell m n} \, a_u^\dagger a_i^\dagger a_j^\dagger a_k^\dagger a_v a_\ell a_m a_n \\
  \label{Eq:X_3}
  & \ \ \ \ \ =(a_0^\dagger )^4 a_0^4\, N^{-4}(\gamma_N-\sigma_N) +\left(\sum_{ijk, \ell m n}(\Theta)_{ijk} \,   a_i^\dagger a_j^\dagger a_k^\dagger\,  a_0^\dagger a_0^4 + \mathrm{H.c.}\right).
\end{align}
In order to understand the size of the term in Eq.~(\ref{Eq:X_3}) better, we are going to rewrite it in terms of the variables $\psi_{ijk}$ defined in Eq.~(\ref{Eq:Definition_psi}), respectively the variables
\begin{align}
\label{Eq:widetilde_psi_Definition}
\widetilde{\psi}_{i j k} &=a_ia_ja_k+\eta_{ijk} \,  a_0^3
\end{align}
defined in Eq.~(\ref{Eq:Definition_psi}) for the concrete choice $K:=0$, see Eq.~(\ref{Eq:Psi_for_K=0}), with the corresponding operator $T_{K=0}:=1+RV_N\pi_0$, see Eq.~(\ref{Eq:Definition_pi_K}). Note that
\begin{align*}
   (T_{K=0}^{-1 })^\dagger \Theta = \Theta +2(\sigma_N-\gamma_N) u_0^{\otimes 3} ,
\end{align*}
and therefore
\begin{align*}
 &   \frac{1}{24}\!  \!  \!  \! \sum_{uijk,v \ell m n} \!  \! \!  \! \!  \! (X_4)_{uijk,v \ell m n} \, a_u^\dagger a_i^\dagger a_j^\dagger a_k^\dagger a_v a_\ell a_m a_n \! = \!  N^{-4}(a_0^\dagger )^4 a_0^4\, (  \sigma_N-\gamma_N)  \\
    & \ \ \ \ \ \  \ \ \ \ \ \  \ \ \ \ \ \ +  \left(\sum_{ijk, \ell m n}(\Theta)_{ijk} \,   \widetilde{\psi}_{ijk}^\dagger \,  a_0^\dagger a_0^4 + \mathrm{H.c.}\right).
\end{align*}

In order to address the correlation term in Eq.~(\ref{Eq:Two-particle_Extra_Correlations}), we make the concrete choice for $T_2$
 \begin{align}
 \label{Eq:Definition_T_2}
(T_2-1)_{\ell (-\ell),00}:=(T_2-1)_{00, \ell (-\ell)}:=3N\frac{\lambda_{0,\ell}}{|\ell|^2},     
 \end{align}
for $|\ell|>K$ and $(T_2-1)_{jk,mn}:=0$ otherwise, where $\lambda_{k,\ell}$ is defined below Eq.~(\ref{Eq:First_Lower_Bound}). With 
 \begin{align}
 \label{Eq:mu}
  \mu_N:=\frac{N^2}{2} (X_2)_{00,00},   
 \end{align}
this choice for a transformation $T_2$ yields
 \begin{align*}
     \frac{1}{2}\sum_{jk,mn} (X_2)_{jk,mn} \, a_j^\dagger a_k^\dagger a_m a_n= N^{-2}(a_0^\dagger )^2 a_0^2\, \mu_N +\bigg(3 N a_0^2\sum_{|\ell|>K}\lambda_{0,\ell}\,  a_{\ell}^\dagger a_{-\ell}^\dagger  +  \mathrm{H.c.}\bigg).
 \end{align*}
Summarizing what we have so far, allows us to write the operator $H_N$ in terms of the new variables $d_k$ and $\xi_{ijk}$ as
 \begin{align}
\nonumber
    & \  H_N=\sum_{k}|k|^2  d_k^\dagger d_k \!+ \! \frac{1}{6}\sum_{i j k,\ell m n}\left(\widetilde{V}_N\right)_{i j k,\ell m n}\xi_{i j k}^\dagger \xi_{\ell m n}+ N^{-4}(a_0^\dagger )^4 a_0^4\, (\gamma_N-\sigma_N)-N^{-2}(a_0^\dagger )^2 a_0^2\, \mu_N  \\
 \label{Eq:Pre_Precise_Lower_Bound}
    & \ \ \ \ \  -\bigg(3 N a_0^2\sum_{|\ell|>K}\lambda_{0,\ell}\,  a_{\ell}^\dagger a_{-\ell}^\dagger  +  \mathrm{H.c.}\bigg)  - \left(\sum_{ijk, \ell m n}(\Theta)_{ijk} \,   \widetilde{\psi}_{ijk}^\dagger \,  a_0^\dagger a_0^4 + \mathrm{H.c.}\right) - \widetilde{\mathcal{E}}.
 \end{align}
 Defining the error term
 \begin{align}
 \label{Eq:Definition_E_Star}
     \mathcal{E}_*:= 3  \!  \left(a_0^\dagger a_0 \! - \! N\right)  \! a_0^2 \! \sum_{|\ell|>K} \!  \! \lambda_{0,\ell}\,  a_{\ell}^\dagger a_{-\ell}^\dagger   \! + \!   9 a_0^\dagger a_0^2 \! \sum_{\ell,0<|k|\leq K} \!  \! \lambda_{k,\ell}  a_\ell^\dagger a_{k-\ell}^\dagger a_k,
 \end{align}
 we obtain as a consequence of Eq.~(\ref{Eq:Pre_Precise_Lower_Bound}) the following Corollary \ref{Cor:Many_Body_Block_Diagonal_II}.
 \begin{cor}
 \label{Cor:Many_Body_Block_Diagonal_II}
     Let $d_k$ and $\xi_{ijk}$ be as in Eq.~(\ref{Eq:Second_Family_of_Variables}) and Eq.~(\ref{Eq:Second_Family_of_Variables_xi}), with $T_2$ defined in Eq.~(\ref{Eq:Definition_T_2}) and $T_4$ defined in Eq.~(\ref{Eq:Definition_T_4}), $\gamma_N,\sigma_N$ and $\mu_N$ as in Eq.~(\ref{Eq:gamma}), Eq.~(\ref{Eq:sigma}) and Eq.~(\ref{Eq:mu}), and let $\mathcal{E}_*$ be as in Eq.~(\ref{Eq:Definition_E_Star}), $\Theta$ as in Eq.~(\ref{Eq:X_3_Definition}) and $\widetilde{\psi}_{ijk}$ as in Eq.~(\ref{Eq:widetilde_psi_Definition}). Furthermore, recall the definition of $\lambda_{0,0}$ below Eq.~(\ref{Eq:Identity_Hamiltonian_A}) and $\mathbb{Q}_K$ in Eq.~(\ref{Eq:Definition_V_Quadratic}). Then
     \begin{align}
         \label{Eq:Precise_Lower_Bound}
    & H_N\geq    (a_0^\dagger )^3 a_0^3 \lambda_{0,0}+ N^{-4}(a_0^\dagger )^4 a_0^4\, (\gamma_N - \sigma_N) -N^{-2}(a_0^\dagger )^2 a_0^2\, \mu_N +\sum_{k}|k|^2  d_k^\dagger d_k  +\mathbb{Q}_K\\
    \nonumber 
    & \ \  \ \  \ \  \ \  \ -  \left(\sum_{ijk, \ell m n}(\Theta)_{ijk} \,   \widetilde{\psi}_{ijk}^\dagger \,  a_0^\dagger a_0^4 + \mathrm{H.c.}\right) +  \left(\mathcal{E}_*+\mathcal{E}_*^\dagger \right)-  \widetilde{\mathcal{E}}.
     \end{align}
    Making use of the notation $\mathcal{E}_\mathcal{K}(d)$ from Eq.~(\ref{Eq:Kinetic_Energy_Operator}) and $\mathcal{E}_\mathcal{P}(\xi)$ from Eq.~(\ref{Eq:Potential_Energy_Operator}), we obtain in the case $K=0$ the identity
     \begin{align}
          \nonumber
        &  H_N =\lambda_{0,0}(a_0^\dagger)^3 a_0^3  \! + \! (\gamma_N \! - \! \sigma_N)  N^{-4}(a_0^\dagger)^4 a_0^4 \! - \! \mu_N  N^{-2}(a_0^\dagger)^2 a_0^2 +\mathcal{E}_\mathcal{K}(d)+  \mathcal{E}_\mathcal{P}(\xi) \\
\label{Eq:Precise_Identity}
       & \ \ \ \ \ \ \ \ \ \ -  \left(\sum_{ijk, \ell m n}(\Theta)_{ijk} \,   \widetilde{\psi}_{ijk}^\dagger \,  a_0^\dagger a_0^4 + \mathrm{H.c.}\right)+ \left(\mathcal{E}_*+\mathcal{E}_*^\dagger \right) -  \widetilde{\mathcal{E}}.
     \end{align}
 \end{cor}
 \begin{proof}
    Using Eq.~(\ref{Eq:Pre_Precise_Lower_Bound}) and the definition of $\widetilde{V}_N$ in Eq.~(\ref{Eq:Definition_tilde_V}), as well as the identities in Eq.~(\ref{Eq:Identity_Hamiltonian_A}) and Eq.~(\ref{Eq:Identity_Hamiltonian_B}), we obtain
     \begin{align*}
         & H_N=\lambda_{0,0}(a_0^\dagger)^3 a_0^3  \! + \! (\gamma_N \! - \! \sigma_N)  N^{-4}(a_0^\dagger)^4 a_0^4 \! - \! \mu_N  N^{-2}(a_0^\dagger)^2 a_0^2  +\sum_{k}|k|^2  d_k^\dagger d_k+\mathbb{Q}_K\\
         &  \ \ \ \ \ \ \ \ \ \ + \frac{1}{6}\sum_{i j k,\ell m n}\left((1-\pi_K)V_N(1-\pi_K)\right)_{i j k,\ell m n}\xi_{i j k}^\dagger \xi_{\ell m n}\\
              & \ \ \ \ \ \ \ \ \ \  -\left(\sum_{ijk, \ell m n}(\Theta)_{ijk} \,   \widetilde{\psi}_{ijk}^\dagger \,  a_0^\dagger a_0^4 + \mathrm{H.c.}\right) + \left(\mathcal{E}_*+\mathcal{E}_*^\dagger \right) -  \widetilde{\mathcal{E}}.
     \end{align*}
     Since $\mathbb{Q}_0=0$ and
     \begin{align*}
         & \frac{1}{6}\sum_{i j k,\ell m n}\left((1-\pi_K)V_N(1-\pi_K)\right)_{i j k,\ell m n}\xi_{i j k}^\dagger \xi_{\ell m n}\geq 0,\\
          & \frac{1}{6}\sum_{i j k,\ell m n}\left((1-\pi_0)V_N(1-\pi_0)\right)_{i j k,\ell m n}\xi_{i j k}^\dagger \xi_{\ell m n}=\mathcal{E}_\mathcal{P}(\xi),
     \end{align*}
     we immediately obtain Eq.~(\ref{Eq:Precise_Lower_Bound}), respectively Eq.~(\ref{Eq:Precise_Identity}).
 \end{proof}

\subsection{Analysis of the Error Terms}
\label{Subsection:E_total}

In the following we are providing an explicit representation of the error term $\widetilde{\mathcal{E}}$ introduced in Eq.~(\ref{Eq:Representation_Of_Second_Order_Transformations}). For this purpose let us introduce the two-particle state $(\varphi_2^0)_{jk}:=N(T_2-1)_{jk,00}$, the three-particle states $(\varphi^0_3)_{ijk}:=\frac{N^{\frac{3}{2}}}{2}(T-1)_{ijk,0 0 0}$ and for $m\in (2\pi \mathbb{Z})^3\setminus \{0\}$
\begin{align*}
    (\varphi_3^m)_{ijk}:=\frac{N^{\frac{3}{2}}}{2}(T-1)_{ijk,m 0 0}+\frac{N^{\frac{3}{2}}}{2}(T-1)_{ijk,0 m 0}+\frac{N^{\frac{3}{2}}}{2}(T-1)_{ijk,0 0 m},
\end{align*}
and the four particle state $(\varphi^0_4)_{uijk}:=\frac{N^2}{6} (T_4-1)_{uijk,0000}$ as well as $(\varphi_4)_{uijk}:=N^2 (T_4-1+\chi)_{uijk,0000}$. Furthermore, let us introduce for $\varphi\in L^2(\Lambda^s)$ and $\psi\in L^2(\Lambda^t)$ with $s,t\geq 0$, and $\ell\leq \min\{s,t\}$, the operator
\begin{align}
\label{Eq:G-Kernel}
   G_{\ell}(\varphi,\psi):=\mathrm{Tr}_{1\rightarrow \ell} \! \left[(-\Delta)_{x_1} \varphi \psi^\dagger \right]
\end{align}
acting on $L^2(\Lambda^{t-\ell})\longrightarrow L^2(\Lambda^{s-\ell})$. In coordinates, the operator is given by
\begin{align*}
    \Big(G_{\ell}(\varphi,\psi)\Big)_{i_1\dots i_{s-\ell},j_1\dots j_{t-\ell}} \! :=  \sum_{k_1\dots k_\ell} \! |k_1|^2 \varphi_{k_1\dots k_\ell i_1\dots i_{s-\ell}}\overline{\psi}_{k_1\dots k_\ell j_1\dots j_{t-\ell}}.
\end{align*}
Finally let $\widetilde{G}:=\mathrm{Tr}_{1\rightarrow 3}\! \left[\widetilde{V}_N\otimes 1 \varphi_4 \varphi_4^\dagger \right]$. With this at hand we can write
\begin{align}
\nonumber
   \widetilde{\mathcal{E}} & = \!  \! \sum_{(s,t,\ell,m,n)\in \mathcal{S}}C_{s,t,\ell}\underset{j_1\dots j_{t-\ell}}{\sum_{i_1\dots i_{s-\ell}}}   \!  \!  \Big(G_{\ell}(\varphi_s^m,\varphi_t^n)\Big)_{i_1\dots i_{s-\ell},j_1\dots j_{t-\ell}}a_{i_{s-\ell}}^\dagger \dots a_{i_1}^\dagger  \frac{a_m^\dagger (a_0^\dagger)^{s-1} a_0^{t-1}a_n}{N^{ \frac{s+t}{2}  }} a_{j_1}\dots a_{j_{t-\ell}}\\
   \label{Eq:Term-by-Term}
   & \ \ \ \ \ \  \ \ \ \   \ \  \ \ \  \ \ \ \ +\sum_{i,j}\left(\widetilde{G}\right)_{i,j}a_i^\dagger \frac{(a_0^\dagger)^4 a_0^4}{N^4} a_j,
\end{align}
where $C_{s,t,\ell}:=\frac{(s-1)! (t-1)!}{(\ell-1)! (s-\ell)! (t-\ell)!}$ and the set of allowed configurations $(s,t,\ell,m,n)\in\mathcal{S}$ is defined by the rules $\ell\leq \min\{s,t\}$ and $\ell<\max\{s,t\}$, where $2\leq s,t\leq 4$ and $|n|,|m|\leq K$ with $m=0$, respectively $n=0$, in case $s\neq 3$, respectively $t\neq 3$. Note that the criterion $\ell<\max\{s,t\}$ makes sure that we only include non-fully contracted parts of the first product in Eq.~(\ref{Eq:Non-Fully_Contracted}) and the constant $C_{s,t,\ell}$ counts the various different ways of contracting, while $\widetilde{G}$ is the kernel associated with the non-fully contracted part of the second product in Eq.~(\ref{Eq:Non-Fully_Contracted}). Estimating the various terms appearing in Eq.~(\ref{Eq:Term-by-Term}) individually allows us to prove the following Lemma \ref{Lemma:Main_Error_Precise}.

\begin{lem}
\label{Lemma:Main_Error_Precise}
    There exists a constant $C>0$ and a function $\epsilon:[0,\infty)\longrightarrow (0,C)$ satisfying $\lim_{K\rightarrow \infty}\epsilon(K)=0$, such that we have for $K$ as in the definition of $\pi_K$ below Eq.~(\ref{Eq:Definition_pi_K})
    \begin{align*}
        \pm \widetilde{\mathcal{E}}\leq C N^{-\frac{1}{4}} \! \sum_k |k|^2 c_k^\dagger \! \left(\frac{\mathcal{N}}{\sqrt{N}}+1\right)^2  \!  \! c_k+CN^{-\frac{1}{4}}  \! \left(\frac{\mathcal{N}}{\sqrt{N}}+1\right)^2 \!  \!  \left(\mathcal{N}+\sqrt{N}\right)+\epsilon(K) \mathcal{N}.
    \end{align*}
\end{lem}
\begin{proof}
   In the following let $\tau\leq \frac{1}{2}$. Using the fact that $\|\frac{a_m^\dagger (a_0^\dagger)^{s-1} a_0^{t-1}a_n}{N^{ \frac{s+t}{2}  }}\|\leq 1$, there exists by Corollary \ref{Cor:Operator_Estimates} a constant $C>0$ such that for $\delta>0$ and $s,t\geq \ell+1$ 
    \begin{align}
    \nonumber
        & \pm\left(\underset{j_1\dots j_{t-\ell}}{\sum_{i_1\dots i_{s-\ell}}}   \!  \!  \Big(G_{\ell,\sigma,\tau}(\varphi_s^m,\varphi_t^n)\Big)_{i_1\dots i_{s-\ell},j_1\dots j_{t-\ell}}(\Phi_{\sigma,s})_{i_1\dots i_{s-\ell}}^\dagger \frac{a_m^\dagger 
        (a_0^\dagger)^{s-1} a_0^{t-1}a_n}{N^{ \frac{s+t}{2}  }}(\Phi_{\tau,t})_{j_1\dots j_{t-\ell}}+\mathrm{H.c.}\right)\\
        \nonumber
        &  \ \ \ \  \ \  \ \ \ \  \leq C \left\|\mathcal{K}_{\tau,s-\ell}^{-\frac{1}{2}} G_{\ell}(\varphi_s^m,\varphi_t^n)\mathcal{K}_{\tau,t-\ell}^{-\frac{1}{2}}\right\| \! \Bigg( \sum_k |k|^{2 }c_k^\dagger\left( \delta \mathcal{N}^{s-\ell-1}+\delta^{-1}\mathcal{N}^{t-\ell-1}\right)c_k  \\
        \label{Eq:First_Estimate_tilde_Error}
        & \ \ \ \  \ \  \ \ \ \  \ \ \ \  \ \ \ \  \ \ \ \ + \left(\mathcal{N}+\sqrt{N}\right) \left(\delta \mathcal{N}^{s-\ell-1}+\delta^{-1}\mathcal{N}^{t-\ell-1}\right) \! \Bigg).
    \end{align}
    For $\tau=\sigma=0$, we have the improved bound on the left hand side of Eq.~(\ref{Eq:First_Estimate_tilde_Error})
    \begin{align}
    \label{Eq:First_Estimate_tilde_Error'}
        C \left\|G_{\ell}(\varphi_s^m,\varphi_t^n)\right\|\left(\delta \mathcal{N}^{s-\ell}+\delta^{-1}\mathcal{N}^{t-\ell}\right) .
    \end{align}
    In the case that either $s$ or $t$ is equal to $\ell$, e.g. $t=\ell$ we obtain by Corollary \ref{Cor:Operator_Estimates}
    \begin{align}
    \nonumber
         &  \ \ \ \  \pm\left(\sum_{i_1\dots i_{s-\ell}}  \!  \!  \Big(G_{\ell,\sigma,\tau}(\varphi_s^m,\varphi_t^n)\Big)_{i_1\dots i_{s-\ell}}(\Phi_{\sigma,s})_{i_1\dots i_{s-\ell}}^\dagger \frac{a_m^\dagger (a_0^\dagger)^{s-1} a_0^{t-1}a_n}{N^{ \frac{s+t}{2}  }}+\mathrm{H.c.}\right)\\
                 \label{Eq:Second_Estimate_tilde_Error}
        & \leq C \left\|\mathcal{K}_{\tau,s-\ell}^{-\frac{1}{2}} G_{\ell}(\varphi_s^m,\varphi_t^n)\right\| \! \Bigg( \delta^{-1}+\delta \sum_k |k|^{2 }c_k^\dagger \mathcal{N}^{s-\ell-1}c_k   + \delta \left(\mathcal{N}+\sqrt{N}\right)  \mathcal{N}^{s-\ell-1} \! \Bigg).
    \end{align}
    In order to obtain good estimates on the operator norms $\left\|\mathcal{K}_{\sigma,s-\ell}^{-\frac{1}{2}} G_{\ell}(\varphi_s^m,\varphi_t^n)\mathcal{K}_{\tau,t-\ell}^{-\frac{1}{2}}\right\|$, observe that we obtain by a Cauchy-Schwarz argument 
    \begin{align*}
        \left\|\mathcal{K}_{\sigma,s-\ell}^{-\frac{1}{2}} G_{\ell}(\varphi_s^m,\varphi_t^n)\mathcal{K}_{\tau,t-\ell}^{-\frac{1}{2}}\right\|\leq \sqrt{\left\|\mathcal{K}_{\sigma,s-\ell}^{-\frac{1}{2}} G_{\ell}(\varphi_s^m,\varphi_s^m)\mathcal{K}_{\sigma,s-\ell}^{-\frac{1}{2}}\right\|\left\|\mathcal{K}_{\tau,t-\ell}^{-\frac{1}{2}} G_{\ell}(\varphi_t^n,\varphi_t^n)\mathcal{K}_{\tau,t-\ell}^{-\frac{1}{2}}\right\|},
    \end{align*}
    i.e. it is enough to control the norm of the symmetric ones. In the following we choose $\sigma=\tau=\frac{1}{2}$, except for the case $s=t=2$ where we choose $\sigma=\tau=0$. By Lemma \ref{Lem:Coefficient_Control} and Lemma \ref{Lem:Coefficient_Control_III} we obtain for a suitable $C>0$
    \begin{align*}
       &  \left\|\mathcal{K}_{\frac{1}{2},3}^{-\frac{1}{2}} G_{1}(\varphi_4^0,\varphi_4^0)\mathcal{K}_{\frac{1}{2},3}^{-\frac{1}{2}}\right\|\leq C N^{-\frac{3}{2}},
       & \left\|\mathcal{K}_{\frac{1}{2},2}^{-\frac{1}{2}} G_{1}(\varphi_3^m,\varphi_3^m)\mathcal{K}_{\frac{1}{2},2}^{-\frac{1}{2}}\right\|\leq C N^{-1},\\
       & \left\|\mathcal{K}_{\frac{1}{2},2}^{-\frac{1}{2}} G_{2}(\varphi_4^0,\varphi_4^0)\mathcal{K}_{\frac{1}{2},2}^{-\frac{1}{2}}\right\|\leq C N^{-\frac{3}{2}},
              & \left\|\mathcal{K}_{\frac{1}{2},1}^{-\frac{1}{2}} G_{2}(\varphi_3^m,\varphi_3^m)\mathcal{K}_{\frac{1}{2},1}^{-\frac{1}{2}}\right\|\leq C N^{-\frac{1}{2}},\\
       & \left\|\mathcal{K}_{\frac{1}{2},1}^{-\frac{1}{2}} G_{3}(\varphi_4^0,\varphi_4^0)\mathcal{K}_{\frac{1}{2},1}^{-\frac{1}{2}}\right\|\leq C N^{-1},
       & \left\|G_{3}(\varphi_3^m,\varphi_3^m)\right\|\leq C N.
    \end{align*}
    Furthermore, $\left\| G_{1}(\varphi_2^0,\varphi_2^0)\right\|\leq \epsilon$ in case $K$ is large enough and $\left\| G_{2}(\varphi_2^0,\varphi_2^0)\right\|\leq C \sqrt{N}$, as well as $\left\| \widetilde{G}\right\|\leq C N^{-1}$. Choosing $\delta:=N^{\frac{t-s}{4}}$, and combining the estimates on the operator norms with Eq.~(\ref{Eq:First_Estimate_tilde_Error}), respectively Eq.~(\ref{Eq:First_Estimate_tilde_Error'}), and Eq.~(\ref{Eq:Second_Estimate_tilde_Error}) concludes the proof by Eq.~(\ref{Eq:Term-by-Term}).
\end{proof}

Following the ideas in the proof of Lemma \ref{Lemma:Main_Error_Precise}, we can furthermore compare the operator $\sum_{k}|k|^{2\tau}a_k^\dagger a_k$ with the corresponding operator $\sum_{k}|k|^{2\tau}d_k^\dagger d_k$ in the variables $d_k$ defined in Eq.~(\ref{Eq:Second_Family_of_Variables}). This is the content of the subsequent Lemma \ref{Lem:Comparison_a_d}.

\begin{lem}
    \label{Lem:Comparison_a_d}
    Let $0\leq \tau< \frac{1}{4}$. Then there exists  $K_0,C>0$ such that for $K\geq K_0$, with $K$ as in the definition of $\pi_K$ below Eq.~(\ref{Eq:Definition_pi_K}),
    \begin{align*}
        \sum_{k}|k|^{2\tau}a_k^\dagger a_k\leq C\sum_{k}|k|^{2}d_k^\dagger d_k+CN^{-\frac{3}{2}}\mathcal{N}^3+CN^{\tau}.
    \end{align*}
\end{lem}
\begin{proof}
    By Lemma \ref{lem:Comparison}, there exists a constant $C>0$ such that
    \begin{align*}
         \sum_{k}|k|^{2\tau}a_k^\dagger a_k\leq C\sum_{k}|k|^{2\tau}c_k^\dagger c_k+CN^{\tau}.
    \end{align*}
    Furthermore, we have by Cauchy-Schwarz the estimate
    \begin{align*}
        \sum_{k}|k|^{2\tau}c_k^\dagger c_k\leq 2\sum_{k}|k|^{2\tau}d_k^\dagger d_k+\sum_{k}|k|^{2\tau}(d_k-c_k)^\dagger (d_k-c_k).
    \end{align*}
    Similar to the definition of $G_\ell$ in Eq.~(\ref{Eq:G-Kernel}) let us introduce 
    \begin{align*}
      G'_{\ell}: & =\mathrm{Tr}_{1\rightarrow \ell} \! \left[(-\Delta)_{x_1} \varphi^0_2 (\varphi^0_2)^\dagger \right],\\
   G''_{\ell}: & =\mathrm{Tr}_{1\rightarrow \ell} \! \left[(-\Delta)_{x_1} \varphi_4^0 (\varphi_4^0)^\dagger \right].
\end{align*}
A similar computation as in Eq.~(\ref{Eq:Term-by-Term}) together with a Cauchy-Schwarz estimate yields 
\begin{align*}
  & \sum_{k}|k|^{2\tau}(d_k-c_k)^\dagger (d_k-c_k) \leq CG'_2 \frac{ (a_0^\dagger)^{2} a_0^{2}}{N^{ 2  }}  +C\sum_{i,j}(G'_1)_{i,j}a_i^\dagger \frac{ (a_0^\dagger)^{2} a_0^{2}}{N^{ 2  }} a_j\\
  & \ \ \  \ \ \ +C\sum_{\ell=1}^4\underset{j_1\dots j_{4-\ell}}{\sum_{i_1\dots i_{4-\ell}}}   \!  \!  \Big(G''_{\ell}\Big)_{i_1\dots i_{4-\ell},j_1\dots j_{4-\ell}}a_{i_{4-\ell}}^\dagger \dots a_{i_1}^\dagger  \frac{ (a_0^\dagger)^{4} a_0^{4}}{N^{ 4  }} a_{j_1}\dots a_{j_{4-\ell}}.  
\end{align*}
for a suitable constant $C>0$. Utilizing the estimates in Lemma \ref{Lem:Coefficient_Control_III} we obtain that $|G'_2|\lesssim 1$, $\|G'_2\|\lesssim \frac{1}{K^2}$, $|G''_4|\lesssim N^{\tau-\frac{1}{2}}\leq 1$ and $\|G''_\ell\|\lesssim N^{\tau-2}\leq N^{-\frac{3}{2}}$ for $\ell\leq 3$. Consequently there exists a $C>0$ such that
\begin{align*}
     \sum_{k}|k|^{2\tau}(d_k-c_k)^\dagger (d_k-c_k) \leq C+\frac{C}{K^2}\mathcal{N}+CN^{-\frac{3}{2}}(\mathcal{N}+1)^3.
\end{align*}
Using $\mathcal{N}\leq  \sum_{k}|k|^{2\tau}a_k^\dagger a_k$ and $\sum_{k}|k|^{2\tau}d_k^\dagger d_k\leq \sum_{k}|k|^{2}d_k^\dagger d_k$ we therefore obtain
\begin{align*}
    \sum_{k}|k|^{2\tau}a_k^\dagger a_k\leq \frac{C}{K^2}\sum_{k}|k|^{2\tau}a_k^\dagger a_k+C\sum_{k}|k|^{2}d_k^\dagger d_k+CN^{-\frac{3}{2}}\mathcal{N}^3+CN^{\tau} .
\end{align*}
Choosing $K$ large enough such that $\frac{C}{K^2}<1$ concludes the proof.
\end{proof}

Before we come to the proof of the lower bound in Theorem \ref{Eq:Main_Theorem_Introduction} in the following Subsection \ref{Subsection:Proof_of_Theorem_Main_in_Text}, we are going to derive sufficient estimates on 
\begin{align*}
   \sum_{ijk, \ell m n}(\Theta)_{ijk} \,   \widetilde{\psi}_{ijk}^\dagger \,  a_0^\dagger a_0^4 + \mathrm{H.c.}
\end{align*}
in Lemma \ref{Lemma:X_3}. Lemma \ref{Lem:General_Comparison} is an auxiliary result required in the proof of Lemma \ref{Lemma:X_3}.

\begin{lem}
\label{Lem:General_Comparison}
  Then there exists a $C>0$, such that
    \begin{align*}
\sum_k |k|^{2} \tilde{c}_k^\dagger \tilde{c}_k\leq C\left(\sum_k |k|^2 c_k^\dagger c_k+\mathcal{N}+1\right),        
    \end{align*}
where we define $ \tilde{c}_k:=a_k +\frac{1}{2}\sum_{ij} (T-1)_{ijk,000}a_i^\dagger a_j^\dagger a_0^3$.
\end{lem}
\begin{proof}
    Similar to Eq.~(\ref{Eq:Increment_Analysis}), we can write
    \begin{align*}
       \sum_{k}|k|^{2}(c_k \! - \! \tilde{c}_k)^\dagger  (c_k \! - \! \tilde{c}_k) \! = \! \sum_{I,I'\neq 0} \! \left( \! f^{I,I'} X^{I,I'}_0 \! + \! \sum_{k\neq 0} \! g^{I,I'}_k a_k^\dagger X^{I,I'}_1 a_k \! + \! \frac{1}{2}\Phi_{1}^\dagger \left(\widetilde{\Upsilon}^{(I,I')}_{1,1} \! + \! \mathrm{H.c.} \! \right) \!  \Phi_1 \! \right) ,
    \end{align*}
    where $f^{I,I'},g_k^{I,I'}, \widetilde{\Upsilon}^{(I,I')}_{1,1}, X^{I,I'}_0$ and $X^{I,I'}_1$ are defined below Eq.~(\ref{Eq:Increment_Analysis}) for the concrete choice $s:=0$ and $\Phi_1$ is defined at the start of the proof of Lemma \ref{lem:Comparison}. Using $I,I'\neq 0$, we obtain the improved estimates $\pm X^{I,I'}_0\lesssim N^2 \mathcal{N}$ and $\pm X^{I,I'}_1\lesssim N^2 \mathcal{N}$. Consequently
    \begin{align*}
        & \pm (f^{I,I'} X^{I,I'}_0+\mathrm{H.c.})\lesssim \mathcal{N},\\
       & (\sum_{k\neq 0} \! g^{I,I'}_k a_k^\dagger X^{I,I'}_1 a_k+\mathrm{H.c.})\lesssim N^{-\frac{3}{2}}\mathcal{N}^2\leq \mathcal{N}.
    \end{align*}
   Furthermore,
    \begin{align*}
        \pm \frac{1}{2}\Phi_1^\dagger \left(\widetilde{\Upsilon}^{(I,I')}_{1,1} \! + \! \mathrm{H.c.} \! \right) \!  \Phi_1\lesssim \sum_k |k|^2 a_k^\dagger \frac{\mathcal{N}}{N} a_k\lesssim \sum_k |k|^2 c_k^\dagger c_k +\mathcal{N}+1,
    \end{align*}
    where we have used Lemma \ref{lem:Comparison} in the last estimate.
\end{proof}

Having Lemma \ref{Lem:General_Comparison} at hand, we are in a position to verify the subsequent Lemma \ref{Lemma:X_3}.

\begin{lem}
\label{Lemma:X_3}
    Let $0\leq \gamma<\frac{1}{4}$. Then there exists a $C>0$ such that
    \begin{align*}
       \pm \left(\sum_{ijk, \ell m n}(\Theta)_{ijk} \,   \widetilde{\psi}_{ijk}^\dagger \,  a_0^\dagger a_0^4 + \mathrm{H.c.}\right)\leq N^{-\frac{1}{4}}\sum_{k}|k|^2 c_k^\dagger c_k+N^{-\frac{1}{4}} \mathcal{N}+CN^{\frac{1}{4}}.
    \end{align*}
\end{lem} 
\begin{proof}
  Let us define for $ijk\neq 0$ 
  \begin{align}
  \label{Eq:Def_lambda_coef}
    \zeta_{ijk}: &  = \!  \frac{1}{24} \! \Big(  \! \Pi_{\mathrm{sym}} 4( \widetilde{V}_N\otimes 1)(T_4 \! - \! 1 \! + \! \chi )\Big)_{0 i j k,0 0 0 0}  \!  \! =\frac{1}{24} \! \Big(  \! \Pi_{\mathrm{sym}} 4( V_N\otimes 1)(T_4 \! - \! 1 \! + \! \chi )\Big)_{0 i j k,0 0 0 0} ,\\
\nonumber 
\zeta_k: & = \frac{1}{24}\Big( \Pi_{\mathrm{sym}} 4( T^\dagger  \widetilde{V}_N\otimes 1)(T_4-1+\chi )\Big)_{0 0 k (-k),0 0 0 0}.
  \end{align}
 Then we have the decomposition
   \begin{align}
   \nonumber
       &  \ \ \     \ \ \    \ \ \    \ \ \    \ \ \  \left(\sum_{ijk, \ell m n}(\Theta)_{ijk} \,   \widetilde{\psi}_{ijk}^\dagger \,  a_0^\dagger a_0^4 + \mathrm{H.c.}\right) \\
       \label{Eq:Split_X_3}
       &= 4\left(\sum_{ijk\neq 0}\overline{\zeta_{ijk}}   (a_0^\dagger )^4  a_0 \, \widetilde{\psi}_{ijk} +\mathrm{H.c.}\right)+6 \left(\sum_{k\neq 0}\overline{\zeta_{k}}   (a_0^\dagger )^4 a_0^2\,   a_k a_{-k} +\mathrm{H.c.}\right).
   \end{align}
   Note that $\zeta_k= \frac{1}{24}\Big( \Pi_{\mathrm{sym}} 4(  V_N\otimes 1)(T_4-1+\chi )\Big)_{0 0 k (-k),0 0 0 0}$ for $|k|>K$. Using the regularity of $V$ and the bounds derived in Lemma \ref{Lem:Coefficient_Control} and Lemma \ref{Lem:Coefficient_Control_III}, we observe that we have $|N^3 \zeta_k|\lesssim N^{-\frac{1}{2}}\left(1+\frac{|k|^2}{N}\right)^{-1}$, and therefore  
 \begin{align*}
     \pm \left( \! \sum_{k\neq 0} \!  \overline{ \zeta_{k}}  (a_0^\dagger )^4 a_0^2  a_k a_{-k}  \! + \! \mathrm{H.c.} \!  \!  \! \right) \! \!  \lesssim    \! \epsilon \!  \sum_{k\neq 0}   \! |k|^{2\tau} \! a_k^\dagger a_k \! + \! \epsilon^{-1} N^{\frac{1}{2}-\tau} \!  \! \lesssim    \! \epsilon \sum_{k\neq 0}   \! |k|^{2\tau} \! c_k^\dagger c_k \! + \! \epsilon \mathcal{N} \! + \! \epsilon N^{\tau} \! + \! \epsilon^{-1}N^{\frac{1}{2}-\tau},
 \end{align*}
 where we have used Lemma \ref{lem:Comparison}. Choosing $\epsilon$ of the order $N^{-\frac{1}{4}}$ and $\tau=\frac{1}{2}$ concludes the analysis of the the second term in Eq.~(\ref{Eq:Split_X_3}). Regarding the first term, we use the definition of $\tilde{c}_k$ from Lemma \ref{Lem:General_Comparison}, in order to identify $\sum_{ijk}\overline{\zeta_{ijk}}   (a_0^\dagger )^4  \widetilde{\psi}_{ijk}a_0$ as
  \begin{align}
  \nonumber
     & \sum_{ijk}\overline{\zeta_{ijk}}(a_0^\dagger )^4 a_i a_j \tilde{c}_k a_0 - \sum_{ijk}\overline{\zeta_{ijk}}(a_0^\dagger )^4 a_i a_j (\tilde{c}_k-a_k) a_0+ \sum_{ijk}\zeta_{ijk}(a_0^\dagger )^4  a_0^4 (T-1)_{ijk,000}\\
     \nonumber
     & \ \ \ = \sum_{ijk}\overline{\zeta_{ijk}}(a_0^\dagger )^4 a_i a_j \tilde{c}_k a_0 - \frac{1}{2}\sum_{ijk,i' j'}\overline{\zeta_{ijk}}(T-1)_{i'j'k,000}a_{i'}^\dagger a_{j'}^\dagger a_0^{4\dagger}a_0^4  a_i a_j\\
     \label{Eq:Split_c_Insertion}
     & \ \ \ \ \ \ \ \ \ \ -2\sum_{ijk,i' }\overline{\zeta_{ijk}}(T-1)_{i'j k,000}a_{i'}^\dagger  a_0^{4\dagger}a_0^4  a_i .
  \end{align}
  In the following we are going to verify that the most significant term $\sum_{ijk}\overline{\zeta_{ijk}}(a_0^\dagger )^4 a_i a_j \tilde{c}_k a_0 $ in Eq.~(\ref{Eq:Split_c_Insertion}) satisfies the desired bound. By Cauchy-Schwarz, we have for $\epsilon>0$
    \begin{align*}
&\left(\sum_{ijk}\overline{\zeta_{ijk}}(a_0^\dagger )^4 a_i a_j \tilde{c}_k a_0 +\mathrm{H.c.}\right)\leq \epsilon \sum_{k}|k|^{2} \tilde{c}_k^\dagger \tilde{c}_k+\epsilon^{-1}\sum_k \frac{1}{|k|^{2}}\left|\sum_{ij}\zeta_{ijk}  a_0^\dagger a_0^4 a_i^\dagger a_j^\dagger   \right|^2\\
& \ =\epsilon \sum_{k}|k|^{2} \tilde{c}_k^\dagger \tilde{c}_k+\epsilon^{-1}G^{(0)} X+\epsilon^{-1}\sum_{i,i'} (G^{(1)})_{i,i'}a_i^\dagger a_{i'} X+\epsilon^{-1}\sum_{ij,i'j'}(G^{(2)})_{ij,i'j'}a_i^\dagger a^\dagger_j X a_{j'} a_{i'}
    \end{align*}
    with $G^{(0)}:=N^5\sum_{ijk}\frac{|\zeta_{ijk}|^2}{|k|^{2}}$, $G^{(1)}_{i,i'}:=N^5\delta_{i,i'}\sum_{jk}\frac{|\zeta_{ijk}|^2+\overline{\zeta_{ijk}}\zeta_{jik}}{|k|^{2}}$ and $G^{(2)}_{ij,i' j'}:=N^5\sum_k \frac{\zeta_{ijk}\zeta_{i' j' k}}{|k|^{2}}$, and $X:=N^{-5}a_0^{4\dagger}a_0 a_0^\dagger a_0^4$. Using again the regularity of $V$ and Lemma \ref{Lem:Coefficient_Control}, as well as the bounds on $T_4$ from Lemma \ref{Lem:Coefficient_Control_III}, yields
    \begin{align*}
        |\zeta_{ijk}|\leq C N^{-\frac{7}{2}}\delta_{i+j+k=0}\left(1+\frac{|i|^2+|j|^2+|k|^2}{N}\right)^{-3},
    \end{align*}
    and therefore $|G^{(0)}|\lesssim 1 $ and $\|G^{(1)}\|\lesssim N^{-\frac{3}{2}}$. The choice $\epsilon:=N{-\frac{1}{4}}$ then yields
    \begin{align*}
        \epsilon^{-1}G^{(0)}X+\epsilon^{-1}\sum_{i,i'} G^{(1)}_{i,i'}a_i^\dagger a_{i'}X\lesssim N^{\frac{1}{4}}+N^{-\frac{5}{4}}\mathcal{N}.
    \end{align*}
    Finally $\|G^{(2)}\|\leq N^{-\frac{3}{2}}$, and therefore
    \begin{align*}
        \sum_{ij,i'j'}G^{(2)}_{ij,i'j'}a_i^\dagger a^\dagger_j X a_{j'} a_{i'}\lesssim N^{-\frac{3}{2}}\mathcal{N}^2\leq N^{-\frac{1}{2}}\mathcal{N}.
    \end{align*}
   This concludes the proof together with Lemma \ref{Lem:General_Comparison}.
\end{proof}

\subsection{Proof of the lower Bound in Theorem \ref{Eq:Main_Theorem_Introduction}}
\label{Subsection:Proof_of_Theorem_Main_in_Text}

In this Subsection, we are going to verify the lower bound in Theorem \ref{Eq:Main_Theorem_Introduction} making use of the sequence of states $\Phi_N$ constructed in Corollary \ref{Cor:Condensation}, which simultaneously satisfies
\begin{align*}
& \mathds{1}\! \left(\mathcal{N}\leq C \sqrt{N}\right)\Phi_N=\Phi_N,\\
&     \braket{\Phi_N,H_N \Phi_N}\leq E_N+C,\\
& \left\langle \Phi_N,\sum_k |k|^2 c_k^\dagger c_k \Phi_N\right\rangle\leq C\sqrt{N}.
\end{align*}
Starting point for our investigations is then the lower bound
\begin{align*}
    & H_N\geq \sum_{k}|k|^{2}d_k^\dagger d_k+   (a_0^\dagger )^3 a_0^3 \lambda_{0,0}+N^{-4}(a_0^\dagger )^4 a_0^4\, (\mu_N - \sigma_N) -N^{-2}(a_0^\dagger )^2 a_0^2\, \mu_N +\mathbb{Q}_K\\
    & \ \  \  \ \  \  \ \  \ - \left(\sum_{ijk, \ell m n}(\Theta)_{ijk} \,   \widetilde{\psi}_{ijk}^\dagger \,  a_0^\dagger a_0^4 + \mathrm{H.c.}\right) +  \left(\mathcal{E}_*+\mathcal{E}_*^\dagger \right)-  \widetilde{\mathcal{E}},
\end{align*}
see Eq.~(\ref{Eq:Precise_Lower_Bound}). Given $\epsilon>0$, assume that $K$ is large enough such that the function $\epsilon(K)$ from Lemma \ref{Lemma:Main_Error_Precise} satisfies $\epsilon(K)\leq \epsilon$. Making use of the fact that 
\begin{align*}
& \ \ \ \mathds{1}\! \left(\mathcal{N}\leq C\sqrt{N}\right)\Phi_N=\Phi_N  ,\\
& \left\langle \Phi_N, \sum_{k}|k|^2 c_k^\dagger  c_k \Phi_N \right\rangle\lesssim \sqrt{N},
\end{align*}
we immediately obtain for $C$ and $\widetilde{C}$ large enough
\begin{align*}
    |\braket{\Phi_N,\widetilde{\mathcal{E}} \Phi_N}| \! \leq \!  C N^{-\frac{1}{4}} \! \left\langle  \! \Phi_N,  \! \sum_k |k|^2 c_k^\dagger   c_k  \Phi_N  \! \right\rangle  \! + \! \epsilon \!  \left\langle   \Phi_N, \mathcal{N}  \Phi_N   \right\rangle \! + \! CN^{\frac{1}{4}} \! \leq  \! \widetilde{C}N^{\frac{1}{4}} \! + \! \epsilon \! \left\langle   \Phi_N, \mathcal{N}  \Phi_N  \right\rangle \! .
\end{align*}
Similarly we obtain by Lemma \ref{Lemma:X_3} and Lemma \ref{Lem:Rest_Estimate} for suitable $C,\widetilde{C}>0$
\begin{align*}
   &  \ \ \ \ \   \ \ \ \ \ \left|\braket{\Phi_N, \left(\sum_{ijk, \ell m n}(\Theta)_{ijk} \,   \widetilde{\psi}_{ijk}^\dagger \,  a_0^\dagger a_0^4 + \mathrm{H.c.}\right) \Phi_N}\right|\\
     & \leq C N^{-\frac{1}{4}}\left(\left\langle \Phi_N, \sum_k |k|^2 c_k^\dagger   c_k  \Phi_N \right\rangle+\sqrt{N}\right)+C N^{\frac{1}{4}}\leq \widetilde{C}N^{\frac{1}{4}},\\
     &  \ \ \ \ \   \ \ \ \ \  \left|\braket{\Phi_N, \left(\mathcal{E}_*+\mathcal{E}_*^\dagger \right) \Phi_N}\right|\leq CN^{\frac{1}{4}}.
\end{align*}
By Lemma \ref{Lem:Quadratic_Potential_Estimate} we furthermore obtain for $\tau,\epsilon>0$ and $K$ large enough, and a suitable $C>0$,
\begin{align*}
    & \braket{\Phi_N, \left((a_0^\dagger )^3 a_0^3 \lambda_{0,0}+\mathbb{Q}_K\right)\Phi_N}\geq \frac{1}{6}b_{\mathcal{M}}(V)N-\epsilon\left\langle\Phi_N, \sum_k |k|^{2\tau}a_k^\dagger a_k\Phi_N \right\rangle -C.
\end{align*}
Moreover we note that we have by Lemma \ref{Lemma:Correction_Coefficients} 
\begin{align*}
   &  \braket{\Phi_N, N^{-4}(a_0^\dagger )^4 a_0^4\, (\sigma_N-\gamma_N) \Phi_N}\leq (\sigma_N-\gamma_N) + |\sigma_N-\gamma_N|\braket{\Phi_N, \left(1-N^{-4}(a_0^\dagger )^4 a_0^4\right)  \Phi_N}\\
   & \ \ \ \ \ \ \ \leq \sigma_N-\gamma_N + |\sigma_N-\gamma_N|\left\langle \Phi_N,\frac{\mathcal{N}}{N}\Phi_N\right\rangle  \leq (\sigma(V)-\gamma(V))\sqrt{N} + o_{N\rightarrow \infty}\! \left(N^{\frac{1}{4}}\right),
\end{align*}
and similarly $\braket{\Phi_N, N^{-2}(a_0^\dagger )^2 a_0^2\, \mu_N \Phi_N}\leq \mu(V) \sqrt{N}+ o_{N\rightarrow \infty}\! \left(1\right)$. Finally by Lemma \ref{Lem:Comparison_a_d}
\begin{align*}
    \left\langle   \Phi_N, \mathcal{N}  \Phi_N  \right\rangle\leq \left\langle\Phi_N, \sum_k |k|^{2\tau}a_k^\dagger a_k\Phi_N \right\rangle\leq C\left\langle\Phi_N, \sum_k |k|^{2}d_k^\dagger d_k\Phi_N \right\rangle+CN^\tau.
\end{align*}
Choosing $\tau<\frac{1}{4}$ and $\epsilon<\frac{1}{2C}$ concludes the proof, since 
\begin{align*}
    E_N +C & \geq\braket{\Phi_N, H_N \Phi_N}\geq \frac{1}{6}b_{\mathcal{M}}(V)N+\left(\gamma-\sigma-\mu\right)\! \sqrt{N}\\
    & \ \ \ \ -CN^{\frac{1}{4}}+\left(1-2C\epsilon\right)\left\langle\Phi_N, \sum_k |k|^{2}d_k^\dagger d_k\Phi_N \right\rangle.
\end{align*}

\section{Second Order Upper Bound}
\label{Sec:Second_Order_Upper_Bound}
It is the goal of this Section to introduce a trial state $\Phi$, which simultaneously annihilates the variables $d_k$ for $k\neq 0$ and $\xi_{\ell m n}$ in case $(\ell,m ,m)\neq 0$, at least in an approximate sense. We are then going to use this trial state $\Phi$ to verify the upper bound in Theorem \ref{Eq:Main_Theorem_Introduction}. For the rest of this Section we specify the parameter $K$ introduced above the definition of $\pi_K$ in Eq.~(\ref{Eq:Definition_pi_K}) as $K:=0$. In order to find $\Phi$, we define $\alpha_{jk}:=(T_2-1)_{jk,00}$ and $\beta_{u i j k}:=(T_4-1)_{uijk,0000}$, and the generator
\begin{align}
\nonumber
  &  \mathcal{G}_2:=\frac{1}{2}\sum_{j k}\alpha_{jk} a_j^\dagger a_k^\dagger a^2_0,\\
  \label{Eq:Four_particle_generator}
   & \mathcal{G}_4:=\frac{1}{24}\sum_{u i j k}\beta_{u i j k}a_u^\dagger a_i^\dagger a_j^\dagger a_k^\dagger a^4_0
\end{align}
of a unitary group $W_s:=e^{s(\mathcal{G}_2+\mathcal{G}_4)^\dagger - s (\mathcal{G}_2+\mathcal{G}_4)}$ and $W:=W_1$. Applying Duhamel's formula, we can express $W^{-1}a_{i_1} a_{i_2} a_{i_3} W$ as
\begin{align}
\label{Eq:Duhamel_Xi_I}
   & W^{-1}a_{i_1} a_{i_2} a_{i_3} W=a_{i_1} a_{i_2} a_{i_3} \! - \! \int_0^1 W_{-s} [a_{i_1} a_{i_2} a_{i_3},\mathcal{G}_4]W_s\mathrm{d}s \! + \! \int_0^1 W_{-s} [a_{i_1} a_{i_2} a_{i_3},\mathcal{G}_2^\dagger  \! + \!  \mathcal{G}_4^\dagger \!  - \! \mathcal{G}_2]W_s\mathrm{d}s.
\end{align}
Furthermore, note that we can write
\begin{align}
\label{Eq:Duhamel_Xi_II}
    [a_{i_1} a_{i_2} a_{i_3},\mathcal{G}_4]=\sum_u \beta_{u i_1 i_2 i_3} a_u^\dagger a_0^4+(\delta \xi)_{i_1 i_2 i_3},
\end{align}
where we define the error term
\begin{align*}
    (\delta \xi)_{i_1 i_2 i_3}:=\frac{1}{4} \sum_{\sigma \in S_3}\sum_{jk} \beta_{i_{\sigma_2}i_{\sigma_3}j k}a_j^\dagger a_k^\dagger a_{i_{\sigma_1}}a_0^4 + \frac{1}{12} \sum_{\sigma \in S_3}\sum_{ijk} \beta_{i_{\sigma_3}i j k}a_i^\dagger a_j^\dagger a_k^\dagger a_{i_{\sigma_1}}a_{i_{\sigma_2}}a_0^4 .
\end{align*}
Therefore we can write the transformed operators $W^{-1}\xi_{i_1 i_2 i_3}W$ as
  \begin{align}
  \nonumber
       & W^{-1}\xi_{i_1 i_2 i_3}W=( \psi \! + \! \delta_1 \psi)_{i_1 i_2 i_3} \! - \! \int_0^1 W_s^{-1 }(\delta \xi)_{i_1 i_2 i_3} W_s\, \mathrm{d}s \! + \! \int_0^1 W_{-s} [a_i a_j a_k,\mathcal{G}_2^\dagger  \! +  \! \mathcal{G}_4^\dagger  \! -  \! \mathcal{G}_2]W_s\mathrm{d}s\\
       \nonumber
        & \ \ \   \ \ \  +\int_0^1 \int_0^s W_t^{-1}\left[\sum_u \beta_{u i_1 i_2 i_3} a_u^\dagger a_0^4,\, \mathcal{G}_2^\dagger +\mathcal{G}_4^\dagger - \mathcal{G}_2 - \mathcal{G}_4 \right]W_t\,   \mathrm{d}t \mathrm{d}s\\   
        \label{Eq:W_transform_of_xi}
    & \ \ \   \ \ \   + \int_0^1 W_{-s} [\xi_{i_1 i_2 i_3} - a_{i_1}a_{i_2}a_{i_3},\, \mathcal{G}_2^\dagger +\mathcal{G}_4^\dagger - \mathcal{G}_2 - \mathcal{G}_4] W_s\, \mathrm{d}s.
    \end{align}
Recall the definition of $\mathcal{E}_\mathcal{P}$ defined in Eq.~(\ref{Eq:Potential_Energy_Operator}). The following Lemma \ref{Lem:delta_xi} provides sufficient bounds on the various error terms appearing in Eq.~(\ref{Eq:W_transform_of_xi}).

\begin{lem}
\label{Lem:delta_xi}
There exists a constant $C>0$, such that 
   \begin{align}
   \label{Eq:delta_xi}
       \mathcal{E}_\mathcal{P}(\delta \xi)  \leq C (\mathcal{N}+1)^6.
   \end{align}
   Furthermore, we have $\mathcal{E}_\mathcal{P}\! \left([a_{i_1} a_{i_2} a_{i_3},\mathcal{G}_2^\dagger  \! +  \! \mathcal{G}_4^\dagger  \! -  \! \mathcal{G}_2]\right)\leq C (\mathcal{N}+1)^6$ and
   \begin{align}
   \label{Eq:Potential_Evolution_Difference}
  & \mathcal{E}_\mathcal{P} \! \left([\xi_{i_1 i_2 i_3} - a_{i_1}a_{i_2}a_{i_3},\, \mathcal{G}_2^\dagger +\mathcal{G}_4^\dagger - \mathcal{G}_2 - \mathcal{G}_4]\right)\leq C (\mathcal{N}+1)^6,\\ 
    \label{Eq:Potential_Evolution_four_particle}
      & \mathcal{E}_\mathcal{P}\! \left(\left[\sum_u \beta_{u i_1 i_2 i_3} a_u^\dagger a_0^4,\, \mathcal{G}_2^\dagger +\mathcal{G}_4^\dagger - \mathcal{G}_2 - \mathcal{G}_4 \right]\right)\leq C (\mathcal{N}+1)^6.
   \end{align}
\end{lem}
\begin{proof}
    Let us define 
    \begin{align*}
    (\delta_1 \xi)_{i_1 i_2 i_3}: & =\frac{1}{4} \sum_{jk} \beta_{i_{2}i_{3}j k}a_j^\dagger a_k^\dagger a_{i_{1}}a_0^4,\\
        C_N: & =\sup_{i_1}\sum_{jk,i_2 i_3,i'_1 i'_2 i'_3}\left|(V_N)_{i_1 i_2 i_3,i'_1 i'_2 i'_3} \beta_{i_{2}i_{3}j k}\beta_{i'_{2}i'_{3}j k}\right|\lesssim N^{-5},
    \end{align*}
    where we have used Lemma \ref{Lem:Coefficient_Control_III} to estimate $C_N$. Applying Cauchy-Schwarz yields
    \begin{align}
    \label{Eq:delta_1_of_xi}
        \sum_{(i_1 i_2 i_3),(i_1' i_2' i_3')\in A}(V_N)_{i_1 i_2 i_3,i_1' i_2' i_3'}(\delta_1 \xi)_{i_1 i_2 i_3}^\dagger (\delta_1 \xi)_{i_1' i_2' i_3'}\leq C_N (a_0^\dagger )^4a_0^4 \left(\sum_{i_1} a_{i_1}^\dagger a_{i_1}\right)(\mathcal{N}+1)^2.
    \end{align}
    Using the fact that $C_N(a_0^\dagger )^4a_0^4 \left(\sum_{i_1} a_{i_1}^\dagger a_{i_1}\right)\leq C_N N^5\lesssim 1$, we observe that the quantity in Eq.~(\ref{Eq:delta_1_of_xi}) is bounded by the right hand side of Eq.~(\ref{Eq:delta_xi}). Let us furthermore define
    \begin{align*}
        (\delta_2 \xi)_{i_1 i_2 i_3}:=\frac{1}{12} \sum_{ijk} \beta_{i_{3}i j k}a_i^\dagger a_j^\dagger a_k^\dagger a_{i_{1}}a_{i_{2}}a_0^4.
    \end{align*}
In the following we want to distinguish between the cases $A':=\{(i_1 i_2 i_3)\in A: i_1,i_2\neq 0\}$ and $A'':=A\setminus A'$, leading to the definition
    \begin{align*}
        C'_N:=\sup_{i_1 i_2}\sum_{ijk,i_3, i_1' i_2' i_3'} \mathds{1}\big((i_1,i_2,i_3),(i'_1 i'_2 i'_3)\in A'\big) \left|(V_N)_{i_1 i_2 i_3,i'_1 i'_2 i'_3} \beta_{i_{3} ij k}\beta_{i'_{3} ij k}\right|\lesssim N^{-5},\\
         C''_N:=\sup_{i_1 i_2}\sum_{ijk,i_3, i_1' i_2' i_3'} \mathds{1}\big((i_1,i_2,i_3),(i'_1 i'_2 i'_3)\in A''\big) \left|(V_N)_{i_1 i_2 i_3,i'_1 i'_2 i'_3} \beta_{i_{3} ij k}\beta_{i'_{3} ij k}\right|\lesssim N^{-\frac{13}{2}},
    \end{align*}
    where we have again used Lemma \ref{Lem:Coefficient_Control_III}. Applying Cauchy-Schwarz leads to the estimate
    \begin{align*}
        \sum_{(i_1 i_2 i_3),(i_1' i_2' i_3')\in A}(V_N)_{i_1 i_2 i_3,i_1' i_2' i_3'}(\delta_2 \xi)_{i_1 i_2 i_3}^\dagger (\delta_2 \xi)_{i_1' i_2' i_3'}\leq C_N' N^4 (\mathcal{N}+1)^5 + C_N'' N^6 (\mathcal{N}+1)^3,
    \end{align*}
    which is bounded by the right hand side of Eq.~(\ref{Eq:delta_xi}). Finally we use that $V_N$ is permutation symmetric and non-negative, and therefore the left hand side of Eq.~(\ref{Eq:delta_xi}) is bounded by
    \begin{align*}
        & \ \ \ \ \ \ \  \ \ \ \ \ \ \ 6 \sum_{(i_1 i_2 i_3),(i_1' i_2' i_3')\in A}(V_N)_{i_1 i_2 i_3,i_1' i_2' i_3'}(\delta_1 \xi+\delta_2 \xi)_{i_1 i_2 i_3}^\dagger (\delta_1 \xi+\delta_2 \xi)_{i_1' i_2' i_3'}\\
        & \leq 12 \! \! \!  \! \! \!  \! \! \! \sum_{(i_1 i_2 i_3),(i_1' i_2' i_3')\in A} \! \! \!  \! \! \! (V_N)_{i_1 i_2 i_3,i_1' i_2' i_3'}(\delta_1 \xi)_{i_1 i_2 i_3}^\dagger (\delta_1 \xi)_{i_1' i_2' i_3'}+ 12 \! \! \!   \! \! \!  \! \! \! \sum_{(i_1 i_2 i_3),(i_1' i_2' i_3')\in A} \! \! \!  \! \! \! (V_N)_{i_1 i_2 i_3,i_1' i_2' i_3'}(\delta_2 \xi)_{i_1 i_2 i_3}^\dagger (\delta_2 \xi)_{i_1' i_2' i_3'}.
    \end{align*}
    Regarding the term $\mathcal{E}_\mathcal{P}\! \left([a_{i_1} a_{i_2} a_{i_3},\mathcal{G}_2^\dagger  \! +  \! \mathcal{G}_4^\dagger  \! -  \! \mathcal{G}_2]\right)$, let us analyse the term involving the commutator with $\mathcal{G}_2$, the terms involving $\mathcal{G}_2^\dagger  \! +  \! \mathcal{G}_4^\dagger$ can be analysed in a similar fashion as has been done in Lemma \ref{Lem:First_Order_Upper_Bound_Useful}. We compute
    \begin{align}
    \label{Eq:Three_Particle_two_particle_conjugation}
       [a_{i_1} a_{i_2} a_{i_3},\mathcal{G}_2]=\alpha_{i_2 i_3}a_{i_1} a_0^2+\sum_u \alpha_{u i_3}a_u^\dagger a_{i_2} a_{i_3} +\{\mathrm{Permutations}\}.
    \end{align}
    In order to analyse the first term on the right hand side of Eq.~(\ref{Eq:Three_Particle_two_particle_conjugation}), let us define $D_N:=\sup_{i_1}\sum_{i_2 i_3, i_1' i_2' i_3'}\left|(V_N)_{i_1 i_2 i_3, i_1' i_2' i_3'} \alpha_{i_2 i_3}\alpha_{i_2' i_3'}\right|$ and note that $D_N\lesssim N^{-3}$ by Lemma \ref{Lem:Coefficient_Control_III}. Hence
    \begin{align*}
        \sum_{(i_1 i_2 i_3),(i_1' i_2' i_3')\in A}(V_N)_{i_1 i_2 i_3,i_1' i_2' i_3'}\left(\alpha_{i_2 i_3}a_{i_1} a_0^2\right)^\dagger \alpha_{i'_2 i'_3}a_{i'_1} a_0^2\leq D_N N^3 \lesssim 1.
    \end{align*}
    Regarding the second term on the right hand side of Eq.~(\ref{Eq:Three_Particle_two_particle_conjugation}), we use again the split $A=A'\cup A''$ and define
    \begin{align*}
       & D'_N:=\sup_{i_1 i_2}\sum_{u, i_3, i'_1 i'_2 i'_3}\mathds{1}\big((i_1,i_2,i_3),(i'_1 i'_2 i'_3)\in A'\big)\left| (V_N)_{i_1 i_2 i_3,i'_1 i'_2 i'_3}\alpha_{u i_3} \alpha_{u i'_3} \right|\lesssim N^{-\frac{5}{2}},\\
       & D'_N:=\sup_{i_1 i_2}\sum_{u, i_3, i'_1 i'_2 i'_3}\mathds{1}\big((i_1,i_2,i_3),(i'_1 i'_2 i'_3)\in A''\big)\left| (V_N)_{i_1 i_2 i_3,i'_1 i'_2 i'_3}\alpha_{u i_3} \alpha_{u i'_3} \right|\lesssim N^{-4},
    \end{align*}
    where we have used Lemma \ref{Lem:Coefficient_Control_III}. Consequently
    \begin{align*}
           \sum_{(i_1 i_2 i_3),(i_1' i_2' i_3')\in A}(V_N)_{i_1 i_2 i_3,i_1' i_2' i_3'}\left(\sum_u \alpha_{u i_3}a_u^\dagger a_{i_2} a_{i_3}\right)^\dagger \left(\sum_u \alpha_{u i_3}a_u^\dagger a_{i_2} a_{i_3}\right)\lesssim N^{-\frac{1}{2}}(\mathcal{N}+1)^3+1,
    \end{align*}
    and therefore $\mathcal{E}_\mathcal{P}\! \left([a_{i_1} a_{i_2} a_{i_3},\mathcal{G}_2]\right)\leq 12 \mathcal{E}_\mathcal{P}\! \left(\alpha_{i_2 i_3}a_{i_1} a_0^2\right)+ 12 \mathcal{E}_\mathcal{P}\! \left(\sum_u \alpha_{u i_3}a_u^\dagger a_{i_2} a_{i_3}\right)\lesssim (\mathcal{N}+1)^3$. The inequalities in Eq.~(\ref{Eq:Potential_Evolution_Difference}) and Eq.~(\ref{Eq:Potential_Evolution_four_particle}) can be verified in similarly.
\end{proof}

With Lemma \ref{Lem:delta_xi} at hand, we show in the subsequent Corollary \ref{Corollary:Xi_Upper_Bound} that after conjugation with the unitary $W$, the potential energy of the operators $\xi_{ijk}$ is comparable to the potential energy of $(\psi+\delta_1)_{ijk}$.

\begin{cor}
\label{Corollary:Xi_Upper_Bound}
There exists a constant $C>0$, such that
\begin{align*}
  W^{-1}\mathcal{E}_\mathcal{P}(\xi)W  \leq   C\,    \mathcal{E}_\mathcal{P}\! \left(\psi+\delta_1 \psi\right) + C\,  (\mathcal{N}+1)^6.
\end{align*}
\end{cor}
\begin{proof}
   Using the sign $V_N\geq 0$, we obtain by the Cauchy-Schwarz inequality and the representation of $W^{-1}\xi_{i_1 i_2 i_3} W$ in Eq.~(\ref{Eq:W_transform_of_xi}) the estimate
    \begin{align*}
        & W^{-1} \mathcal{E}_\mathcal{P}(\xi)W=\mathcal{E}_\mathcal{P}\! \left(W^{-1}\xi W\right)\leq 5\mathcal{E}_\mathcal{P}\! \left(\psi \! + \! \delta_1 \right)+5 \int_0^1 W_s^{-1 }\mathcal{E}_\mathcal{P}\! \left( \delta \xi\right) W_s \mathrm{d}s\\
        & \ \  + 5\int_0^1 W_{-s} \mathcal{E}_\mathcal{P}\! \left([a_i a_j a_k,\mathcal{G}_2^\dagger  \! +  \! \mathcal{G}_4^\dagger  \! -  \! \mathcal{G}_2]\right)\mathrm{d}s W_s \\
        & \ \  +\frac{5}{2}\int_0^1 \int_0^s W_t^{-1}\mathcal{E}_\mathcal{P}\left(\left[\sum_u \beta_{u i_1 i_2 i_3} a_u^\dagger a_0^4,\, \mathcal{G}_2^\dagger +\mathcal{G}_4^\dagger - \mathcal{G}_2 - \mathcal{G}_4 \right]\right)W_t \mathrm{d}t \mathrm{d}s\\
        & \ \ + 5\int_0^1 W_{-s} \mathcal{E}_\mathcal{P} \left([\xi_{i_1 i_2 i_3} - a_{i_1}a_{i_2}a_{i_3},\, \mathcal{G}_2^\dagger +\mathcal{G}_4^\dagger - \mathcal{G}_2 - \mathcal{G}_4] \right) W_s\mathrm{d}s\\
        & \lesssim \! \mathcal{E}_\mathcal{P}\! \left(\psi \! + \! \delta_1 \right) \! + \! \int_0^1 \!  W_{-s}(\mathcal{N} \! + \! 1)^6 W_s \, \mathrm{d}s  \! +  \! \int_0^1 \!  \int_0^s  \! W_{-t}(\mathcal{N} \! + \! 1)^6 W_t\, \mathrm{d}t\mathrm{d}s \! \lesssim \!  \mathcal{E}_\mathcal{P}\! \left(\psi \! + \! \delta_1 \right) \! + \! (\mathcal{N} \! + \! 1)^6,
    \end{align*}
    where we have first used Lemma \ref{Lem:delta_xi} and subsequently Lemma \ref{Lem:Particle_Number_Creation} in the last line.
\end{proof}

Regarding the variable $d_k$, Duhamel's formula yields for $k\neq 0$
\begin{align}
\label{Eq:Duhamel_d}
    W^{-1}d_kW \! = \! c_k \! + \! \int_0^1 \!  \int_0^s \!  W_{-t}\left[\left[a_k,G_2 \! + \! G_4\right],\mathcal{G}_2^\dagger  \! + \!  \mathcal{G}^\dagger_4\right]W_t \mathrm{d}t\mathrm{d}s \! + \! \int_0^1  \! W_{-s}\left[d_k \! - \!  a_k,G_2^\dagger  \! + \!  G_4^\dagger \right]W_s \mathrm{d}s.
\end{align}
Recall the kinetic energy of an operator valued one particle vector $\Theta_k$ defined in Eq.~(\ref{Eq:Kinetic_Energy_Operator}). Then the following Lemma \ref{Lem:delta_d} provides sufficient bounds on the various error terms appearing in Eq.~(\ref{Eq:Duhamel_d}).

\begin{lem}
\label{Lem:delta_d}
    There exists a constant $C>0$, such that for $m\in \mathbb{N}$ 
    \begin{align*}
   & \mathcal{E}_\mathcal{K}\! \left(\mathcal{N}^m \left[d_k- a_k,G_2^\dagger + G_4^\dagger \right]\right)\leq \frac{C}{N}(\mathcal{N}+1)^{5+2m},\\
    &    \mathcal{E}_\mathcal{K}\! \left(\mathcal{N}^m\left[\left[a_k,G_2+G_4\right],\mathcal{G}_2^\dagger + \mathcal{G}^\dagger_4\right]\right)\leq  \frac{C}{N}(\mathcal{N}+1)^{5+2m}.
    \end{align*}
\end{lem}
\begin{proof}
    Let us compute as an example for $k\neq 0$
    \begin{align*}
        & \ \ \ \ \ \left[\left[a_k,G_4\right],\mathcal{G}^\dagger_4\right]= \frac{1}{24}\sum_{i_1 i_2 i_3}\beta_{k i_1 i_2 i_3}\left[ a_{i_1}^\dagger a_{i_2}^\dagger a_{i_3}^\dagger a^4_0,\mathcal{G}^\dagger_4\right]\\
        &  =\frac{1}{24}\sum_{i_1 i_2 i_3}\beta_{k i_1 i_2 i_3} a_{i_1}^\dagger a_{i_2}^\dagger \left[a_{i_3}^\dagger,\mathcal{G}^\dagger_4\right] a^4_0+\frac{1}{24}\sum_{i_1 i_2 i_3}\beta_{k i_1 i_2 i_3} a_{i_1}^\dagger \left[a_{i_2}^\dagger,\mathcal{G}^\dagger_4\right] a_{i_3}^\dagger a^4_0\\
        & \ \ +\frac{1}{24}\sum_{i_1 i_2 i_3}\beta_{k i_1 i_2 i_3}\left[ a_{i_1}^\dagger,\mathcal{G}^\dagger_4\right] a_{i_2}^\dagger a_{i_3}^\dagger a^4_0+\frac{1}{24}\sum_{i_1 i_2 i_3}\beta_{k i_1 i_2 i_3} a_{i_1}^\dagger a_{i_2}^\dagger a_{i_3}^\dagger \left[a^4_0,\mathcal{G}^\dagger_4\right],
    \end{align*}
     and let us focus on the term
     \begin{align*}
         \sum_{i_1 i_2 i_3}\beta_{k i_1 i_2 i_3} a_{i_1}^\dagger a_{i_2}^\dagger \left[a_{i_3}^\dagger,\mathcal{G}^\dagger_4\right] a^4_0=\frac{1}{6}  \!  \sum_{i_1 i_2 i_3, j_1 j_2 j_3}\beta_{k i_1 i_2 i_3}\overline{\beta_{i_3 j_1 j_2 j_3}} (a_0^\dagger)^4 a_0^4 a_{i_1}^\dagger a_{i_2}^\dagger a_{j_1}a_{j_2}a_{j_3}.
     \end{align*}
     Defining 
     \begin{align*}
         C_N:=\sum_{k,i_1 i_2,j_1 j_2 j_3}|k|^2 \left|\sum_{i_3}\beta_{k i_1 i_2 i_3}\overline{\beta_{i_3 j_1 j_2 j_3}}\right|^2\lesssim N^{-\frac{19}{2}},
     \end{align*}
    where we have used Lemma \ref{Lem:Coefficient_Control_III}, we obtain
     \begin{align*}
         \mathcal{E}_\mathcal{K}\! \left(\sum_{i_1 i_2 i_3}\beta_{k i_1 i_2 i_3} \mathcal{N}^m a_{i_1}^\dagger a_{i_2}^\dagger \left[a_{i_3}^\dagger,\mathcal{G}^\dagger_4\right] a^4_0\right)\lesssim N^{-\frac{19}{2}}\left((a_0^\dagger)^4 a_0^4\right)^2 (\mathcal{N}+1)^{5+2m}\leq N^{-\frac{3}{2}}(\mathcal{N}+1)^{5+2m}.
     \end{align*}
     The other estimates in Lemma \ref{Lem:delta_d} can be verified similarly.
\end{proof}

Similar to Corollary \ref{Corollary:Xi_Upper_Bound}, we show in the following Corollary \ref{Corollary:d_Upper_Bound} that, after conjugation with the unitary $W$, the kinetic energy of the operators $ d_k$ is comparable to the one of $c_k$.

\begin{cor}
\label{Corollary:d_Upper_Bound}
    There exists a constant $C>0$, such that for $m\in \mathbb{N}$
    \begin{align*}
        W^{-1} \mathcal{E}_K\! \left(\mathcal{N}^m d\right) W\leq C\mathcal{E}_K\! \left( c\right) + C \mathcal{E}_K\! \left(\mathcal{N}^m c\right) + \frac{C}{N} (\mathcal{N}+1)^{5+2m}.
    \end{align*}
\end{cor}
\begin{proof}
By Lemma \ref{Lem:Particle_Number_Creation} we have $W^{-1}\mathcal{N}^{2m}W=(W^{-1}\mathcal{N}^{m}W)^* W^{-1}\mathcal{N}^{m}W\lesssim \mathcal{N}^{2m}+1$, hence
\begin{align*}
    W^{-1}\mathcal{E}_\mathcal{K}\! \left(\mathcal{N}^m d\right)W=\mathcal{E}_\mathcal{K}\! \left(W^{-1}\mathcal{N}^m W\, W^{-1} d\, W\right)\lesssim \mathcal{E}_\mathcal{K}\! \left( W^{-1} d\, W\right)+\mathcal{E}_\mathcal{K}\! \left(\mathcal{N}^m W^{-1} d\, W\right).
\end{align*}
    Following the ideas in Corollary \ref{Corollary:Xi_Upper_Bound}, we estimate using Eq.~(\ref{Eq:Duhamel_d})
    \begin{align*}
       & \ \ \ \   \ \ \ \ \ \ \   \ \ \ \  \ \ \ \   \mathcal{E}_\mathcal{K}\! \left(\mathcal{N}^m W^{-1} d\, W\right)\leq 3\, \mathcal{E}_\mathcal{K}\left(\mathcal{N}^m c\right)\\
        &  \ \ \  +\frac{3}{2}\int_0^1 \!  \int_0^s \!   W_{-t}\mathcal{E}_\mathcal{K}\! \left(\mathcal{N}^m\left[\left[a_k,G_2 \!  + \!  G_4\right],\mathcal{G}_2^\dagger  \!  + \!   \mathcal{G}^\dagger_4\right]\right)W_t\mathrm{d}s \\
        & \ \ \  + \!  3\int_0^1 W_{-s}\mathcal{E}_\mathcal{K}\! \left(\mathcal{N}^m\left[d_k \!  -  \!  a_k,G_2^\dagger  \!  + \!   G_4^\dagger \right]\right)W_s\mathrm{d}s\\
        & \lesssim \!  \mathcal{E}_\mathcal{K}\left(\mathcal{N}^m c\right) \! +   \! \frac{1}{N}\int_0^1 \!  \int_0^s  \! \!   W_{-t}(\mathcal{N} \! + \! 1)^{5+2m} W_t\mathrm{d}s \! +  \! \frac{1}{N}\int_0^1  \!  \! W_{-s}(\mathcal{N} \! + \! 1)^5W_s\mathrm{d}s \\
        & \lesssim  \! \mathcal{E}_\mathcal{K}\left(\mathcal{N}^m c\right) \! + \!  \frac{1}{N}(\mathcal{N} \! + \! 1)^{5+2m},
    \end{align*}
     where we have first used Lemma \ref{Lem:delta_d} and subsequently Lemma \ref{Lem:Particle_Number_Creation} in the last line.
\end{proof}

Before we come to the proof of the upper bound in Theorem \ref{Eq:Main_Theorem_Introduction}, we are showing in the following Lemma \ref{Lem:c_d_comparision} that even without a unitary conjugation, the kinetic energy of $c_k$ is comparable with the one of $d_k$. The price for dropping the unitary $W$ is that we obtain an order $O_{N\rightarrow \infty}\! \left(\sqrt{N}\right)$ pre-factor in front of the excess term $(\mathcal{N}+3)^{3+2m}$, instead of a pre-factor of the order $O_{N\rightarrow \infty}\! \left(1\right)$.

\begin{lem}
\label{Lem:c_d_comparision}
    There exists a constant $C>0$, such that for $m\in \mathbb{N}$
    \begin{align*}
        \mathcal{E}_K\! \left(\mathcal{N}^m c\right)\leq C \mathcal{E}_K\! \left(\mathcal{N}^m d\right)+C\sqrt{N}(\mathcal{N}+1)^{3+2m}.
    \end{align*}
\end{lem}
\begin{proof}
    Note that we can write $c_k$ as
    \begin{align*}
       c_k=d_k-2\sum_j \alpha_{jk}\alpha_{jk} a_j^\dagger a_0^2-4\sum_{uij}\beta_{uijk}a^\dagger_u a^\dagger_i a^\dagger_j a_0^4, 
    \end{align*}
and therefore
    \begin{align*}
       \mathcal{E}_K\! \left(\mathcal{N}^m c\right)\leq 3\, \mathcal{E}_K\! \left(\mathcal{N}^m d\right)+12\, \mathcal{E}_K\!\left(\sum_j \alpha_{jk}\, \mathcal{N}^m a_j^\dagger a_0^2\right)+48\, \mathcal{E}_K\!\left(\sum_{uij}\beta_{uijk}\, \mathcal{N}^m a^\dagger_u a^\dagger_i a^\dagger_j a_0^4\right).
    \end{align*}
    Defining the constant $C_N:=\sum_{jk}|k|^2 |\alpha_{jk}|^2\lesssim N^{-\frac{3}{2}}$, which follows from Lemma \ref{Lem:Coefficient_Control_III}, we obtain
    \begin{align*}
        \mathcal{E}_K\!\left(\sum_j \alpha_{jk}\, \mathcal{N}^m a_j^\dagger a_0^2\right)\leq C_N (a_0^\dagger)^2 a_0^2 (\mathcal{N}+1)^{m+1}\lesssim \sqrt{N}(\mathcal{N}+1)^{m+1}.
    \end{align*}
    Similarly we obtain 
    \begin{align*}
        \mathcal{E}_K\!\left(\sum_{uij}\beta_{uijk}\, \mathcal{N}^m a^\dagger_u a^\dagger_i a^\dagger_j a_0^4\right)\lesssim \sqrt{N}(\mathcal{N}+3)^3\lesssim \sqrt{N}(\mathcal{N}+1)^3.
    \end{align*}
\end{proof}

\begin{proof}[Proof of the upper bound in Theorem \ref{Eq:Main_Theorem_Introduction}]
    Let us define the trial state $\Phi:=W\Gamma$, where $\Gamma$ is the state defined below Eq.~(\ref{Eq:Gamma_0}), and recall the representation of $H_N$ in Eq.~(\ref{Eq:Precise_Identity})
    \begin{align*}
       & \braket{\Phi,H_N \Phi}=\lambda_{0,0}N^3 \braket{\Phi,N^{-3}(a_0^\dagger)^3 a_0^3 \Phi} \! + \! (\gamma_N \! - \! \sigma_N) \braket{\Phi, N^{-4}(a_0^\dagger)^4 a_0^4 \Phi} \! - \! \mu_N \braket{\Phi, N^{-2}(a_0^\dagger)^2 a_0^2 \Phi}\\
       & \ \ \ \ \ \ \ \ \ \  + \braket{\Phi,  \mathcal{E}_\mathcal{K}(d) \Phi}+\braket{\Phi,  \mathcal{E}_\mathcal{P}(\xi) \Phi}+ 2\mathfrak{Re}\left\langle \Phi , \mathcal{E}_* \Phi \right\rangle - \left\langle \Phi , \widetilde{\mathcal{E}}\Phi \right\rangle \\
     &   \ \ \ \ \ \ \ \ \ \  - \left\langle \Phi , \left(\sum_{ijk, \ell m n}(\Theta)_{ijk} \,   \widetilde{\psi}_{ijk}^\dagger \,  a_0^\dagger a_0^4 + \mathrm{H.c.}\right) \Phi\right\rangle.
    \end{align*}
    By Lemma \ref{Lem:Particle_Number_Creation} and the fact that $\pm \left(N^{-m}(a_0^\dagger )^m a_0^m-1\right)\lesssim N^{-1}\mathcal{N}$, we obtain that
    \begin{align*}
        \left| \braket{\Phi,N^{-3}(a_0^\dagger)^3 a_0^3 \Phi}-1\right|\lesssim N^{-1}\braket{\Phi, \mathcal{N}\Phi}=N^{-1}\braket{\Gamma, W^{-1}\mathcal{N}W\Gamma}\lesssim N^{-1}\braket{\Gamma, (\mathcal{N}+1)\Gamma}\lesssim N^{-1},
    \end{align*}
    see Eq.~(\ref{Eq:Gamma_Particle_Control}) for the last estimate. Making use of Lemma \ref{Lemma:Main_Error_Precise}, Lemma \ref{Lem:Rest_Estimate} and Lemma \ref{Lemma:X_3} yields
    \begin{align*}
      &  \left|\left\langle \Phi , \mathcal{E}_* \Phi \right\rangle\right|\lesssim N^{-\frac{1}{4}}\braket{\Phi,\mathcal{E}_\mathcal{K}(c)\Phi}+N^{\frac{1}{4}}\braket{\Phi,\mathcal{N}\Phi},\\
      & \Big| \Big\langle \! \Phi , \left(\sum_{ijk, \ell m n}(\Theta)_{ijk} \,   \widetilde{\psi}_{ijk}^\dagger \,  a_0^\dagger a_0^4 + \mathrm{H.c.}\right) \Phi \! \Big\rangle\Big|\lesssim N^{-\frac{1}{4}}\braket{\Phi,\mathcal{E}_\mathcal{K}(c)\Phi}+N^{-\frac{1}{4}}\braket{\Phi,\mathcal{N}\Phi}+N^{\frac{1}{4}},\\
       & \left| \! \left\langle \Phi , \widetilde{\mathcal{E}} \Phi  \! \right\rangle\right| \! \lesssim  \! N^{-\frac{1}{4}}\left\langle  \! \Phi,\mathcal{E}_\mathcal{K}\! \left(\left(\frac{\mathcal{N}}{\sqrt{N}} \! + \! 1\right)c\right)\Phi \! \right\rangle \! + \! N^{\frac{1}{4}}\left\langle  \! \Phi, \left(\frac{\mathcal{N}}{\sqrt{N}} \! + \! 1\right)^2  \! \! \left(\mathcal{N} \! + \! \sqrt{N}\right) \! \Phi\right\rangle  \! + \! \braket{\Phi,\mathcal{N}\Phi}.
    \end{align*}
    Observe that $\braket{\Phi, \mathcal{N}^m \Phi}\lesssim 1$ and furthermore we have by Lemma \ref{Lem:c_d_comparision} and Corollary \ref{Corollary:d_Upper_Bound}
    \begin{align*}
       & \left\langle \Phi, \mathcal{E}_K\! \left(\mathcal{N}^m c\right)\Phi \right\rangle\lesssim \left\langle \Phi, \mathcal{E}_K\! \left(\mathcal{N}^m d\right)\Phi \right\rangle + \sqrt{N}=\left\langle \Gamma , W^{-1}\mathcal{E}_K\! \left(\mathcal{N}^m d\right)W\Gamma \right\rangle + \sqrt{N}\\
       & \ \ \ \ \ \ \ \ \  \lesssim \left\langle \Gamma , \mathcal{E}_K\! \left( c\right)\Gamma \right\rangle + \left\langle \Gamma , \mathcal{E}_K\! \left(\mathcal{N}^m c\right)\Gamma \right\rangle + \sqrt{N}\lesssim \sqrt{N},
    \end{align*}
  where we used Corollary \ref{Cor:First_Order_Trial_State_Kinetic} in the last estimate. Putting together what we have so far, and utilizing Lemma \ref{Lemma:Correction_Coefficients}, yields
  \begin{align*}
      \braket{\Phi,H_N \Phi} \! = \! \frac{1}{6}b_{\mathcal{M}}(V)N \! + \! \big(\gamma(V)  \! - \!   \sigma(V)  \! - \!  \mu(V) \big) \sqrt{N} \! +  \! \braket{\Phi,  \mathcal{E}_\mathcal{K}(d) \Phi} \! + \! \braket{\Phi,  \mathcal{E}_\mathcal{P}(\xi) \Phi} \! + \! O_{N\rightarrow \infty}\! \left( \! N^{\frac{1}{4}} \! \right) \! .
  \end{align*}
  Using Corollary \ref{Corollary:Xi_Upper_Bound}, Corollary \ref{Corollary:d_Upper_Bound}, Corollary \ref{Cor:First_Order_Trial_State_V} and Corollary \ref{Cor:First_Order_Trial_State_Kinetic}, we further have
  \begin{align*}
     & 0\leq \braket{\Phi,  \mathcal{E}_\mathcal{P}(\xi) \Phi}=\braket{\Gamma ,  W^{-1} \mathcal{E}_\mathcal{P}(\xi)W \Gamma}\lesssim \braket{\Gamma ,   \mathcal{E}_\mathcal{P}(\psi+\delta_1 \psi) \Gamma}+1\lesssim 1,\\
    &  0\leq \braket{\Phi,  \mathcal{E}_\mathcal{K}(d) \Phi}=\braket{\Gamma, W^{-1} \mathcal{E}_\mathcal{K}(d)W \Gamma}\lesssim \braket{\Gamma,  \mathcal{E}_\mathcal{K}(c) \Gamma}+\frac{1}{N}\lesssim \frac{1}{N}.
  \end{align*}
\end{proof}

\section{Analysis of the Scattering Coefficients}
\label{Sec:Analysis of the Scattering Coefficients}
This Section is devoted to the study of the variational problems in the definition of $b_{\mathcal{M}}(V)$ in Eq.~(\ref{Eq:Definition_of_b}) and the definition of $\sigma(V)$ in Eq.~(\ref{Eq:Definition_of_sigma}) as well the study of their corresponding minimizers $\omega$ and $\eta$. Especially we want to compare $\gamma(V),\mu(V)$ and $\sigma(V)$ defined in Eq.~(\ref{Eq:Definition_of_gamma}), Eq.~(\ref{Eq:Definition_of_mu}) and Eq.~(\ref{Eq:Definition_of_sigma}) with $\gamma_N,\mu_N$ and $\sigma_N$ defined in Eq.~(\ref{Eq:sigma}), Eq.~(\ref{Eq:mu}) and Eq.~(\ref{Eq:gamma}), see Lemma \ref{Lemma:Correction_Coefficients}. The proof will be based on the observation that the $N$-dependent quantities can be seen as a counterpart on the three dimensional torus $\Lambda$ to the $N$-independent quantities defined in terms of variational problems on the full space $\mathbb{R}^3$. Similarly we will compare in Lemma \ref{Lem:Coefficient_Control_II} the modified scattering length $b_{\mathcal{M}}(V)$, which can be expressed in terms of the minimizer $\omega$ as
\begin{align*}
 b_{\mathcal{M}}(V)=\int_{\mathbb{R}^6}(1-\omega)V\, \mathrm{d}x , 
\end{align*}
see \cite{NRT1}, with its counterpart on the torus $\Lambda$ defined below Eq.~(\ref{Eq:Identity_Hamiltonian_A}) as
\begin{align*}
    6\lambda_{0,0}=\braket{u_0 u_0 u_{0},(V_N -V_N R V_N)u_0 u_0 u_0}=\left(V_N -V_N R V_N\right)_{000,000}.
\end{align*}
The proof of Lemma \ref{Lem:Coefficient_Control_II} is based on the observation that $(1-\omega)V=V-\omega V$ is the full space counterpart to the renormalized potential $V_N -V_N R V_N$. \\

In the following Lemma \ref{Lemma:Existence_of_solutions} we want to derive properties of $\mathcal{Q}$ defined in Eq.~(\ref{Eq:Q-Functional}) as
\begin{align*}
\mathcal{Q}(\varphi)=  \int_{\mathbb{R}^9}\left\{2\big|\mathcal{M}_*\nabla \varphi(x)\big|^2+\mathbb{V}(x)\! \left|\frac{f(x)}{\mathbb{V}(x)}-\varphi(x)\right|^2 \! \right\}\mathrm{d}x  
\end{align*}
and its minimizers. For this purpose it will be useful to introduce for a given cut-off parameter $\ell$ and a smooth function $\chi:\mathbb{R}^6\longrightarrow \mathbb{R}$ function with $\chi(x)=1$ for $|x|_\infty \leq \frac{1}{3}$ and $\chi(x)=0$ for $|x|_\infty >\frac{1}{2}$, the modified function
\begin{align*}
 f_\ell(x_1,x_2,x_3):=V(x_1,x_2)\chi(\ell^{-1} x_2,\ell^{-1}x_3)\omega(x_2,x_3).
\end{align*}
Furthermore, we define the corresponding functional, acting on $\dot{H}^1(\mathbb{R}^9)$, as
\begin{align*}
    \mathcal{Q}_\ell(\varphi):=\int_{\mathbb{R}^9}\left\{2\big|\mathcal{M}_*\nabla \varphi(x)\big|^2+\mathbb{V}(x)\! \left|\frac{f_\ell(x)}{\mathbb{V}(x)}-\varphi(x)\right|^2 \! \right\}\mathrm{d}x,
\end{align*}
and $\sigma_\ell(V):=\mathcal{Q}_\ell(0)-\inf_{\varphi\in \dot{H}^1(\mathbb{R}^9)}\mathcal{Q}_\ell(\varphi)$.

\begin{lem}
\label{Lemma:Existence_of_solutions}
    There exists a unique minimizer $\eta$ of the functional $\mathcal{Q}$ in $\dot{H}^1(\mathbb{R}^9)$, and $\eta$ satisfies the point-wise bounds $0\leq \eta\leq \frac{1}{-2\Delta_{\mathcal{M}_*}}f$ and $\sigma(V)=\int_{\mathbb{R}^9}f(x)\eta(x)\mathrm{d}x$, as well as
    \begin{align*}
       ( -2\Delta_{\mathcal{M}_*}+\mathbb{V})\eta=f
    \end{align*}
    in the sense of distributions. Furthermore, $\mathcal{Q}_\ell$ has a unique minimizer $\eta_\ell$, and $\eta_\ell$ satisfyies $0\leq \eta_\ell \leq \frac{1}{-2\Delta_{\mathcal{M}_*}}f_\ell$ and $\sigma_\ell(V)=\int_{\mathbb{R}^9}f_\ell(x)\eta_\ell(x)\mathrm{d}x$, as well as $( -2\Delta_{\mathcal{M}_*}+\mathbb{V})\eta_\ell=f_\ell$ and 
    \begin{align*}
        \sigma(V)=\lim_{\ell\rightarrow \infty}\sigma_\ell(V).
    \end{align*}
\end{lem}
\begin{proof}
  Following the proof of \cite{NRT1}, we observe that since $\mathcal{Q}(0)<\infty$, there exists a minimizing sequence $\varphi_n$ for $\mathcal{Q}$ with $\sup_{n}\|\nabla \varphi_n\|<\infty$ and $\sup_{n}\left\|\sqrt{\mathbb{V}}\left(\frac{f}{\mathbb{V}}-\varphi_n\right)\right\|<\infty$, and therefore there exists by Banach-Alaoglu a subsequence $\varphi_n$ and an element $\eta\in \dot{H}^1(\mathbb{R}^9)$, such that $\nabla \varphi_n\rightharpoonup \nabla \eta$, $\sqrt{\mathbb{V}}\left(\frac{f}{\mathbb{V}}-\eta \right)\in L^2(\mathbb{R}^9)$ and $\sqrt{\mathbb{V}}\left(\frac{f}{\mathbb{V}}-\varphi_n\right)\rightharpoonup \sqrt{\mathbb{V}}\left(\frac{f}{\mathbb{V}}-\eta \right)$. We observe that $\eta$ is a minimizer of $\mathcal{Q}$, since
  \begin{align*}
      \mathcal{Q}(\eta) \! = \! 2\|\mathcal{M}_* \!  \nabla \eta\|^2 \! + \! \Big\| \! \sqrt{\mathbb{V}}\Big(\frac{f}{\mathbb{V}} \! - \! \eta \Big) \! \Big\|^2  \!  \! \! \leq  \! \liminf_n  \! \left\{ \! 2\|\mathcal{M}_* \! \nabla \varphi_n\|^2 \! + \! \Big\|\! \sqrt{\mathbb{V}}\Big(\frac{f}{\mathbb{V}} \! - \! \varphi_n \Big) \! \Big\|^2 \! \right\} \! = \! \liminf_n \mathcal{Q}(\varphi_n).
  \end{align*}
  Computing $0=\frac{\mathrm{d}}{\mathrm{d}t}\mathcal{Q}(\eta+t\varphi)$ for $\varphi\in C^\infty_{0}(\mathbb{R}^9)$ immediately gives in the sense of distributions
  \begin{align*}
      (-2\Delta_{\mathcal{M}_*}+\mathbb{V})\eta=f,
  \end{align*}
and computing $0=\frac{\mathrm{d}}{\mathrm{d}t}\mathcal{Q}(\eta+t\eta)$ yields $\sigma(V)=\int_{\mathbb{R}^9}f(x)\eta(x)\mathrm{d}x$. Regarding the uniqueness, we note that $\varphi\mapsto \|\mathcal{M}_* \!  \nabla \varphi\|^2$ is strictly convex on $\dot{H}^1(\mathbb{R}^9)$, and therefore $\mathcal{Q}$ is strictly convex too. Consequently the minimizer $\eta$ is unique. Using that $\frac{f}{\mathbb{V}}\geq 0$, we obtain $\mathcal{Q}(|\varphi|)\leq \mathcal{Q}(\varphi)$ for all $\varphi\in \dot{H}^1(\mathbb{R}^9)$ and by the uniqueness of the minimizer $\eta=|\eta|\geq 0$. For the purpose of obtaining an upper bound on $\eta$, we observe that $\frac{1}{\mathcal{M}_*\nabla}f\in L^2(\mathbb{R}^9)$ and define the functional
\begin{align*}
    \widetilde{\mathcal{Q}}(\varphi):=  \int_{\mathbb{R}^9}\left\{2\Big|\mathcal{M}_*\nabla \varphi(x)+\frac{1}{2\mathcal{M}_*\nabla}f(x)\Big|^2+\mathbb{V}(x)\! \left|\varphi(x)\right|^2 \! \right\}\mathrm{d}x .
\end{align*}
Since $\widetilde{\mathcal{Q}}(\varphi)=\mathcal{Q}(\varphi)+\widetilde{\mathcal{Q}}(0)-\mathcal{Q}(0)$, we observe that $\eta$ is the unique minimizer of $\widetilde{Q}$. It is furthermore clear that $\widetilde{\mathcal{Q}}(\min\{\varphi,\frac{1}{-2\Delta_{\mathcal{M}_*}}f\})\leq \widetilde{\mathcal{Q}}(\varphi)$, and therefore we obtain by the uniqueness of minimizer for $\widetilde{\mathcal{Q}}$
\begin{align*}
    \eta=\min \! \left\{\eta,\frac{1}{-2\Delta_{\mathcal{M}_*}}f\right\} \! \leq \frac{1}{-2\Delta_{\mathcal{M}_*}}f.
\end{align*}
The properties of $\mathcal{Q}_\ell$ can verified analogously.

In order to compare $\sigma(V)$ with $\sigma_\ell(V)$, let us first verify the point-wise bounds 
\begin{align}
\label{Eq:Comparison_Minima}
    \eta_\ell\leq \eta\leq \eta_\ell + \frac{1}{-2\Delta_{\mathcal{M}_*}}(f-f_\ell).
\end{align}
For this purpose we introduce the additional functionals $\mathcal{Q}'_\ell$ and $\mathcal{Q}''_\ell$ as
\begin{align*}
    \mathcal{Q}'_\ell(\varphi): & =\int_{\mathbb{R}^9}\left\{2\Big|\mathcal{M}_*\nabla \varphi-\mathcal{M}_*\nabla \eta\Big|^2+\mathbb{V}\! \left|\varphi-\eta+\frac{f-f_\ell}{\mathbb{V}}\right|^2 \! \right\}\mathrm{d}x,\\
     \mathcal{Q}''_\ell(\varphi): & =\int_{\mathbb{R}^9}\left\{2\Big|\mathcal{M}_*\nabla \varphi-\mathcal{M}_*\nabla \eta_\ell+\frac{1}{2\mathcal{M}_*\nabla}f-\frac{1}{2\mathcal{M}_*\nabla}f_\ell \Big|^2+\mathbb{V}\! \left|\varphi-\eta_\ell\right|^2 \! \right\}\mathrm{d}x.
\end{align*}
By a straightforward computation, we observe that $\mathcal{Q}'_\ell(\varphi)=\mathcal{Q}_\ell(\varphi)+\mathcal{Q}'_\ell(0)-\mathcal{Q}_\ell(0)$ and similarly $\mathcal{Q}''_\ell(\varphi)=\mathcal{Q}(\varphi)+\mathcal{Q}''_\ell(0)-\mathcal{Q}(0)$. Therefore $\eta_\ell$ is the unique minimizer of $\mathcal{Q}'_\ell$, and since $f(x)\geq f_\ell(x)$ we further have
\begin{align*}
    \mathcal{Q}'_\ell\! \left(\min\left\{\varphi,\eta\right\}\right)\leq \mathcal{Q}'_\ell(\varphi).
\end{align*}
Consequently $\eta_\ell\leq \eta$. The second inequality in Eq.~(\ref{Eq:Comparison_Minima}) follows analogously, using that $\eta$ is the unique minimizer of $\mathcal{Q}''_\ell$ and that 
\begin{align*}
    \mathcal{Q}''_\ell\! \left(\min\left\{\varphi,\eta_\ell+\frac{1}{-2\Delta_{\mathcal{M}_*}}(f-f_\ell)\right\}\right)\leq \mathcal{Q}''_\ell(\varphi) .
\end{align*}
Using the fact that $|f(x)-f_\ell(x)|\lesssim \frac{\mathds{1}(|x|>\frac{\ell}{3})}{|x|^4}$, see \cite{NRT1}, we obtain the estimate
\begin{align*}
    \frac{1}{-2\Delta_{\mathcal{M}_*}}(f-f_\ell)(x)=\frac{\Gamma\! \left(\frac{9}{2}\right)}{28 \pi^{\frac{9}{2}}\mathrm{det}[M_*]}\int_{\mathbb{R}^9}\frac{f(y)-f_\ell(y)}{|\mathcal{M}_*^{-1} (x-y)|^{7}}\mathrm{d}y\lesssim \int_{|y|\geq \frac{\ell}{3}}\frac{1}{|y|^4 |x-y|^{7}}\mathrm{d}y\lesssim \frac{1}{\ell},
\end{align*}
i.e. $\frac{1}{-2\Delta_{\mathcal{M}_*}}(f-f_\ell)$ converges point-wise to zero and consequently $\eta_\ell$ converges point-wise to $\eta$ by Eq.~(\ref{Eq:Comparison_Minima}). Using Fatou's Lemma and $f_\ell(x)\eta_\ell(x)\geq 0$, as well as the fact that $f_\ell$ converges point-wise to $f$, therefore yields 
\begin{align*}
   & \sigma(V) \! = \! \int_{\mathbb{R}^9}f(x)\eta(x)\mathrm{d}x\leq \liminf_{\ell\rightarrow \infty}\int_{\mathbb{R}^9}f_\ell(x)\eta_\ell(x)\mathrm{d}x = \liminf_{\ell\rightarrow \infty}\sigma_\ell (V) \leq \limsup_{\ell\rightarrow \infty}\sigma_\ell (V)\leq \sigma(V),
\end{align*}
where we have used in the last inequality that $f_\ell \eta_\ell \leq f \eta$ by Eq.~(\ref{Eq:Comparison_Minima}).
\end{proof}

Before we can compare the modified scattering length $b_\mathcal{M}(V)$ with its counterpart on the torus in Lemma \ref{Lem:Coefficient_Control_II}, we need the following auxiliary result Lemma \ref{Lem:Coefficient_Control}.

\begin{lem}
\label{Lem:Coefficient_Control}
  Recall the definition of the coefficients $\lambda_{k,\ell}$ below Eq.~(\ref{Eq:Identity_Hamiltonian_A}) and the definition of $T$ in Eq.~(\ref{Eq:Definition_Feshbach-Schur}). Then there exists a constant $C>0$ such that $|\lambda_{k,\ell}|\leq \frac{C}{N^2}\left(1+\frac{|\ell|^2}{N }\right)^{-1}$ and
    \begin{align}
    \label{Eq:T-1_Coefficient_Control}
        \left|(T-1)_{ i j k,\ell 00}\right|\leq \frac{C\, \mathds{1}(i+j+k=\ell)}{N^2(|i|^2+|j|^2+|k|^2)}\left(1+\frac{|i|^2+|j|^2+|k|^2}{N  + |\ell|^2}\right)^{-2}.
    \end{align}
\end{lem}
\begin{proof}
In order to verify Eq.~(\ref{Eq:T-1_Coefficient_Control}), we observe that for $|\ell|\leq K$, $\nabla^n (T-1)e^{i \ell x}=\nabla^n R V_N e^{i \ell x}$ can be written as the sum of terms of the form
\begin{align}
\label{Eq:Split_of_nabla_R}
    Q^{\otimes 3}\nabla^{k_1} \! (V_N)Q^{\otimes 3}\dots  Q^{\otimes 3}\nabla^{k_m} \! (V_N)Q^{\otimes 3}\, \nabla^a R^{1-b}V_N  e^{i \ell x},
\end{align}
where the coefficients satisfy $k_1+\dots + k_m + 2m +a +2b=n$ and either (I) that $b=1$, (II) that $b=0$ and $a=1$ or (III) that $b=0$, $a=0$ and $m\geq 1$ as well as $k_m=0$. In the following we are going to verify individually for the three cases (I)-(III) that the Fourier transform of the expression in Eq.~(\ref{Eq:Split_of_nabla_R}) has an $L^{\infty}$ bound of the order $N^{-2}(\sqrt{N}+|\ell|)^{n-2}$ for $n\geq 2$, which immediately implies Eq.~(\ref{Eq:T-1_Coefficient_Control}). Let us first of all state the useful bounds
\begin{align}
\label{Eq:Lem:Coefficient_Control_Addendum_a}
  \left\|\sqrt{Q^{\otimes 3}\nabla^{k}  (V_N)Q^{\otimes 3}}\right\|  & \leq \sqrt{\|Q^{\otimes 3}\nabla^{k}  (V_N)Q^{\otimes 3}\|}\lesssim \sqrt{N}^{\frac{k}{2}+1},\\
  \label{Eq:Lem:Coefficient_Control_Addendum_b}
  \left\|\sqrt{Q^{\otimes 3}\nabla^{k}  (V_N)Q^{\otimes 3}}e^{i K\cdot X}\right\| &  = \sqrt{\braket{e^{i K\cdot X},Q^{\otimes 3}\nabla^{k}  (V_N)Q^{\otimes 3}e^{i K\cdot X}}} \lesssim \sqrt{N}^{\frac{k}{2}-2}
\end{align}
for $k\geq 0$. Regarding the case (I), we obtain immediately by Eq.~(\ref{Eq:Lem:Coefficient_Control_Addendum_b}) and Eq.~(\ref{Eq:Lem:Coefficient_Control_Addendum_b})
\begin{align*}
   & \left|\left\langle e^{i K\cdot X},Q^{\otimes 3}\nabla^{k_1} \! (V_N)Q^{\otimes 3}\dots  Q^{\otimes 3}\nabla^{k_m} \! (V_N)Q^{\otimes 3}\, \nabla^a V_N  e^{i \ell x} \right\rangle \right|\\
   & \lesssim \sqrt{N}^{k_1+\dots +k_m +2m-4}\big(\sqrt{N}+|\ell|\big)^a\leq N^{-2}(\sqrt{N}+|\ell|)^{n-2}.
\end{align*}
Since the case (II) is similar to the case (III), let us directly have a look at the case (III), where we use the fact that $\|\sqrt{Q^{\otimes 3} V_N Q^{\otimes 3}}R\nabla\|\lesssim 1$ to obtain
\begin{align}
\label{Eq:Fourier_Case_III}
   & \ \ \ \ \   \left|\left\langle e^{i K\cdot X},Q^{\otimes 3}\nabla^{k_1} \! (V_N)Q^{\otimes 3}\dots  Q^{\otimes 3}\nabla^{k_{m-1}} \! (V_N)Q^{\otimes 3}\, Q^{\otimes 3} V_N Q^{\otimes 3}\, R  V_N  e^{i \ell x} \right\rangle \right|\\
   \nonumber
    & \lesssim \left\|\sqrt{Q^{\otimes 3} V_N Q^{\otimes 3}}\, Q^{\otimes 3}\nabla^{k_{m-1}} \! (V_N)Q^{\otimes 3}\dots Q^{\otimes 3}\nabla^{k_{1}} \! (V_N)Q^{\otimes 3}\, e^{i K\cdot X}\right\|\, \left\|\frac{1}{\nabla}Q^{\otimes 3} V_N  e^{i \ell x} \right\|.
\end{align}
As a consequence of Eq.~(\ref{Eq:Lem:Coefficient_Control_Addendum_b}) and Eq.~(\ref{Eq:Lem:Coefficient_Control_Addendum_b}) we have
\begin{align}
\label{Eq:Lem:Coefficient_Control_Addendum}
        \left\|\sqrt{Q^{\otimes 3} V_N Q^{\otimes 3}}\, Q^{\otimes 3}\nabla^{k_{m-1}} \! (V_N)Q^{\otimes 3}\dots Q^{\otimes 3}\nabla^{k_{1}} \! (V_N)Q^{\otimes 3}\, e^{i K\cdot X}\right\|\lesssim \sqrt{N}^{k_1+\dots + k_m+2(m-1)-2}.
\end{align}
This yields the desired estimate for the term in Eq.~(\ref{Eq:Fourier_Case_III}), since $\left\|\frac{1}{\nabla}Q^{\otimes 3} V_N  e^{i \ell x} \right\|\lesssim \sqrt{N}^{-2}$. The bounds on $\lambda_{k,\ell}$ can be verified similarly.
\end{proof}

In the following Lemma \ref{Lem:Coefficient_Control_II}, we show that the renormalized potential $N^2(V_N - V_N R  V_N)$ converges to $b_{\mathcal{M}}(V)\delta(x-y,x-z)$ in a suitable sense. The analogous result for Bose gases with two-particle interactions has been verified in \cite[Lemma 1]{B}.

\begin{lem}
\label{Lem:Coefficient_Control_II}
Let $b_{\mathcal{M}}(V)$ be the modified scattering length introduced in Eq.~(\ref{Eq:Definition_of_b}). Then 
    \begin{align*}
        \left|N^2(V_N - V_N R V_N)_{0 0  0,0 0 0 }-b_{\mathcal{M}}(V)\right|\leq \frac{1}{N}
    \end{align*}
Furthermore, $(V_N - V_N R V_N)_{ijk,\ell m n}=0$ in case $i+j+k\neq \ell+m+n$ and otherwise 
    \begin{align*}
         \left|N^2(V_N - V_N R V_N)_{ijk,\ell m n}-b_{\mathcal{M}}(V)\right|\leq \frac{C_{ijk,\ell m n}}{\sqrt{N}}
    \end{align*}
    for suitable constants $C_{ijk,\ell m n}$.
\end{lem}
\begin{proof}
    Let $\omega$ be the unique minimizer to the variational problem in Eq.~(\ref{Eq:Definition_of_b}), which exists according to \cite{NRT1} and satisfies in the sense of distributions
\begin{align*}
    (-2\Delta_{\mathcal{M}}+V)\omega=V.
\end{align*}
Furthermore, let $\chi$ be a smooth function with $\chi(x)=1$ for $|x|_\infty \leq \frac{1}{3}$ and $\chi(x)=0$ for $|x|_\infty>\frac{1}{2}$, and let us denote for a function $f$ the rescaled version with $f^L(x):=f(Lx)$. Then we define for $n=(n_1,n_2,n_3)\in \left(2\pi \mathbb{Z}\right)^{3\times 3}$ and $0<\ell<{\sqrt{N}}$
\begin{align}
\label{Eq:Definition_psi_n}
    \psi_n(x,y,z):=e^{i n_1 x}e^{i n_2 y}e^{i n_3 z}\chi^{\frac{{\sqrt{N}}}{\ell}} \! \left(x-y,x-z\right)\omega^{\sqrt{N}} (x-y,x-z).
\end{align}
In the following we want to show that $\psi_n$ is an approximation of $RV_Ne^{i n_1 x}e^{i n_2 y}e^{i n_3 z}$. For this purpose we observe that the function $\psi_n$ satisfies the differential equation
\begin{align}
\label{Eq:Differential_Equation_First}
    (-\Delta+V_{N})\psi_n=e^{i n_1 x}e^{i n_2 y}e^{i n_3 z} \Big(V_{N}-\xi_n(x-y,x-z)\Big),
\end{align}
where we define $\xi_n:\mathbb{T}^2\longrightarrow \mathbb{R}$ as 
\begin{align}
\label{Eq:Definition_xi_n}
 &   \ \ \ \ \ \ \ \xi_n:=2\Delta_{\mathcal{M}}(\chi^{\frac{{\sqrt{N}}}{\ell}})\omega^{\sqrt{N}}+4\mathcal{M}^2 \nabla(\chi^{\frac{{\sqrt{N}}}{\ell}})\nabla \omega^{\sqrt{N}}-2(n_1^2+n_2^2+n_3^2)\chi^{\frac{{\sqrt{N}}}{\ell}}\omega^{\sqrt{N}}\\
 \nonumber
   & + \! 4i(n_1 \! - \! n_2)\!  \! \left(\nabla_{x_1}(\chi^{\frac{{\sqrt{N}}}{\ell}})\omega^{\sqrt{N}} \! + \! \chi^{\frac{{\sqrt{N}}}{\ell}}\nabla_{x_1}(\omega^{\sqrt{N}})\right) \!  \! + \!  4i(n_1 \! - \! n_3)\!  \! \left(\nabla_{x_2}(\chi^{\frac{{\sqrt{N}}}{\ell}})\omega^{\sqrt{N}} \! + \! \chi^{\frac{{\sqrt{N}}}{\ell}}\nabla_{x_2}(\omega^{\sqrt{N}})\right) \!  \! .
\end{align}
In order to verify that $\xi_n$ can be treated as an error term, we first note that we have 
\begin{align*}
 & \ \ \widehat{\Delta_{\mathcal{M}}\chi^{\frac{{\sqrt{N}}}{\ell}}}   =\left({\sqrt{N}}^{-1}\ell\right)^{4} \widehat{\Delta_{\mathcal{M}}\chi}^{\frac{\ell}{{\sqrt{N}}}} ,\\
& \widehat{\mathcal{M}^2\nabla\chi^{\frac{{\sqrt{N}}}{\ell}}}   = \left({\sqrt{N}}^{-1}\ell\right)^{5}\widehat{\mathcal{M}^2\nabla\chi}^{\frac{\ell}{{\sqrt{N}}}},
\end{align*}
and utilizing the density $\rho:=-2\Delta_\mathcal{M}\omega$ we obtain
\begin{align*}
   \widehat{\omega^{\sqrt{N}}}(K) & = {\sqrt{N}}^{-6}\widehat{\omega}^{\frac{1}{{\sqrt{N}}}}(K)={\sqrt{N}}^{-6}  \frac{\widehat{\rho}({\sqrt{N}}^{-1}K)}{|{\sqrt{N}}^{-1}K |^2}={\sqrt{N}}^{-4}\frac{\widehat{\rho}({\sqrt{N}}^{-1}K)}{|K|^2},\\
  \widehat{\nabla \omega^{\sqrt{N}}}(K)  & ={\sqrt{N}}^{-4}\frac{\widehat{\rho}({\sqrt{N}}^{-1}K)K}{|K|^2},
\end{align*}
Furthermore, we observe that we can write the Fourier transform of $\xi_n$ as
\begin{align*}
  &  \widehat{\xi}_n=\widehat{\Delta_{\mathcal{M}}\chi^{\frac{{\sqrt{N}}}{\ell}}} * \widehat{\omega^{\sqrt{N}}} +2 \widehat{\mathcal{M}^2 \nabla\chi^{\frac{{\sqrt{N}}}{\ell}}} * \widehat{\nabla \omega^{\sqrt{N}}}-2(n_1^2+n_2^2+n_3^2)\widehat{\chi^{\frac{{\sqrt{N}}}{\ell}}} \! * \! \widehat{\omega^{\sqrt{N}}}\\
    & \ \ \ \ \ \ \ \ \ \ \ + \! 4i(n_1 \! - \! n_2)\!  \! \left( \! \widehat{\nabla_{x_1} \! (\chi^{\frac{{\sqrt{N}}}{\ell}})} \! * \! \widehat{\omega^{\sqrt{N}}} \! + \! \widehat{\chi^{\frac{{\sqrt{N}}}{\ell}}} \! * \! \widehat{\nabla_{x_1}\! (\omega^{\sqrt{N}})}\right) \\
    &  \ \ \ \ \ \ \ \ \ \ \ + \!  4i(n_1 \! - \! n_3)\!  \! \left(\widehat{\nabla_{x_2} \! (\chi^{\frac{{\sqrt{N}}}{\ell}})}\! * \! \widehat{\omega^{\sqrt{N}}} \! + \! \widehat{\chi^{\frac{{\sqrt{N}}}{\ell}}} \! * \! \widehat{\nabla_{x_2} \! (\omega^{\sqrt{N}})}\right) \!  \!.
\end{align*}
Since $\rho\in L^1(\mathbb{R}^6)$, see \cite{NRT1}, we have $\widehat{\rho}\in L^\infty(\mathbb{R}^6)$, and distinguishing between the cases $|K|\lesssim \frac{{\sqrt{N}}}{\ell}$ and $|K|\gg \frac{{\sqrt{N}}}{\ell}$, yields the estimate
\begin{align}
\label{Eq:Delta_chi_Fourier}
    \big|\widehat{\Delta_{\mathcal{M}}\chi^{\frac{{\sqrt{N}}}{\ell}}} * \widehat{\omega^{\sqrt{N}}}(K)\big| \! \lesssim  \! {\sqrt{N}}^{-4} \!  & \left({\sqrt{N}}^{-1}\ell\right)^{4} \!  \! \int_{\mathbb{R}^6} \!  \! \frac{|\widehat{\Delta_{\mathcal{M}}\chi}({\sqrt{N}}^{-1}\ell P)|}{|K+P|^2}\mathrm{d}P \!  \lesssim  \! {\sqrt{N}}^{-4}\min\left\{ \! \left( \! \frac{{\sqrt{N}}}{\ell |K|} \! \right)^2,1 \! \right\} \! .
\end{align}
Using that $\rho$ has compact support as a consequence of the scattering equation, we obtain that $x\nabla \rho(x)\in L^1(\mathbb{R}^6)$ and therefore 
\begin{align*}
  \left|K_1\widehat{\rho}(\sqrt{N}^{-1} K_1)-K_2\widehat{\rho}(\sqrt{N}^{-1} K_2)\right|\lesssim |K_1-K_2|.  
\end{align*}
Since $\widehat{\mathcal{M}^2\nabla\chi}$ is reflection anti-symmetric, we furthermore have
\begin{align*}
 &   \ \ \ \ \ \big|\widehat{\mathcal{M}^2 \nabla\chi^{\frac{{\sqrt{N}}}{\ell}}} * \widehat{\nabla \omega^{\sqrt{N}}}\big| \! \lesssim  \! {\sqrt{N}}^{-4} \! \!  \left({\sqrt{N}}^{-1}\ell\right)^{5}   \! \! \int_{\mathbb{R}^6}   \! \! \! \frac{|\widehat{\mathcal{M}^2 \nabla\chi}({\sqrt{N}}^{-1}\ell P)|}{|K \! + \! P|^2}  \left|  (K \! + \! P) \! - \! (K \! - \! P)  \right| \! \mathrm{d}P\\
    &   + \! {\sqrt{N}}^{-4} \! \left({\sqrt{N}}^{-1}\ell\right)^{5} \!  \int_{\mathbb{R}^6}  \! \! \Big|\widehat{\mathcal{M}^2\nabla\chi}({\sqrt{N}}^{-1}\ell P)\Big|\, \Big|(K \! - \! P)\widehat{\rho}\big(\sqrt{N}^{-1}(K \! - \! P) \big)\Big| \! \left|\frac{1}{|K \! + \! P|^2} \! - \! \frac{1}{|K \! - \! P|^2}\right|\\
    &    \ \ \ \ \     \ \ \ \ \   \ \ \ \ \   \ \ \ \ \ \ \ \ \   \ \ \ \ \  \lesssim {\sqrt{N}}^{-4} \! \min\left\{ \! \left(\frac{{\sqrt{N}}}{\ell |K|}\right)^2 \! ,1\right\}.
\end{align*}
Summarizing what we have so far, we can estimate the Fourier coefficients of $\xi_0$ by
\begin{align}
\label{Eq:Fourier_Xi}
    |\widehat{\xi}_0(K)|\lesssim  N^{-2}  \min\left\{ \! \left(\frac{{\sqrt{N}}}{\ell |K|}\right)^2 \! ,1\right\}.
\end{align}
Proceeding similarly for general $n\neq 0$ we observe the slightly weaker estimate
\begin{align}
\label{Eq:Fourier_Xi_general}
    |\widehat{\xi}_n(K)|\lesssim  N^{-2}  \min\left\{ \frac{{\sqrt{N}}}{\ell |K|}  ,1\right\}.
\end{align}
In the following let $R$ denote the resolvent of the operator $Q^{\otimes 3}(-\Delta+V_N)Q^{\otimes 3}$ on the torus, and note that we obtain as a consequence of the differential equation Eq.~(\ref{Eq:Differential_Equation_First})
\begin{align}
\label{Eq:Three_Scattering_Identity}
    RV_N e^{i n_1 x}e^{i n_2 y}e^{i n_3 z}=\psi_n+R\xi_n(x-y,x-z)e^{i n_1 x}e^{i n_2 y}e^{i n_3 z}+(RV_N-1)(1-Q^{\otimes 3})\psi_n.
\end{align}
Using the fact that $V$ has compact support, there exists a $\ell_0>0$ such that $\chi^{\frac{1}{\ell}}(x)=1$ for $x\in \mathrm{supp}(V)$ and $\ell\geq \ell_0$, and therefore we obtain for $n=(n_1,n_2,n_3)$ and $m=(m_1,m_2,m_3)$ 
\begin{align}
\nonumber
    (V_N)_{m,n}-\braket{V_N e^{i m_1 x}e^{i m_2 y}e^{i m_3 z},\psi_n} & =\frac{\delta_{\overline{n}=\overline{m}}}{N^2}\int_{\mathbb{R}^6}e^{i\frac{m_2-n_2}{\sqrt{N}}x}  e^{i\frac{m_3-n_3}{\sqrt{N}}y} V(x,y)\big(1-\omega(x,y)\big)\mathrm{d}x\mathrm{d}y\\
    \label{Eq:n_m_identity}
    & =\frac{\delta_{\overline{n}=\overline{m}}}{N^2} b_\mathcal{M}(V)+O_{N\rightarrow \infty}\! \left(N^{-\frac{5}{2}}\right),
\end{align}
where $\overline{n}:=n_1+n_2+n_3$ and we have used that we can express the minimum in Eq.~(\ref{Eq:Definition_of_b}) according to \cite{NRT1} as
\begin{align*}
   b_\mathcal{M}(V)=\int_{\mathbb{R}^6} V(x,y)\big(1-\omega(x,y)\big)\mathrm{d}x\mathrm{d}y.
\end{align*}
We observe that in the case $m=0$ and $n=0$, we even have the exact identity
\begin{align*}
  N^2 (V_N)_{000,000}-N^2 \braket{V_N ,\psi_0}=b_\mathcal{M}(V)  .
\end{align*}
Consequently we obtain
\begin{align*}
    N^2 (V_N \! -  \! V_N R V_N)_{000,000} \! =  \!  b_\mathcal{M}(V) \! - \! N^2\braket{RV_N,\xi_0} \! - \! 3N^2  \! \sum_{k} \! \!   \braket{V_N,(RV_N \! - \! 1)e^{ik(x-y)}} \! \braket{e^{ik(x-y)},\psi_0} \! .
\end{align*}
Using Lemma \ref{Lem:Coefficient_Control} and the fact that $(RV)_{ijk,000}=(T-1)_{ijk,000}=0$ in case $i\neq -(j+k)$, we can estimate
\begin{align*}
 N^2 \left|\braket{RV_N,\xi_0}\right| & = N^2\left|\sum_{jk}\overline{(RV_N)_{-(j+k)j k}}\, \widehat{\xi}_0(-(j+k),j,k)\right|\\
 & \lesssim N^{-2}\sum_{j k}\frac{\left(1+\frac{|j|^2+|k|^2}{N}\right)^{-2}}{|j|^2 + |k|^2}\min\left\{ \! \left(\frac{{\sqrt{N}}}{\ell (|j|+|k|)}\right)^2 \! ,1\right\}\\
   &  \leq N^{-2}\sum_{j k}\frac{ \left(1+\frac{|j|^2+|k|^2}{N}\right)^{-\frac{3}{2}}\left(\frac{{\sqrt{N}}}{\ell (|j|+|k|)}\right)^2}{|j|^2 + |k|^2}\lesssim \frac{1}{\ell^2}.
\end{align*}
Again by Lemma \ref{Lem:Coefficient_Control} we have that
\begin{align*}
   \left| \braket{V_N,(RV_N-1)e^{ik(x-y)}}\right|\lesssim N^{-2}\left(1+\frac{|k|^2}{N}\right)^{-1}, 
\end{align*}
and following the proof of Eq.~(\ref{Eq:Delta_chi_Fourier}) we obtain that $|\braket{e^{ik(x-y)},\psi_0}|\lesssim N^{-2}\frac{1}{1+|k|^2}$. Therefore
\begin{align*}
    N^2\sum_{k}\left|\braket{V_N,(RV_N-1)e^{ik(x-y)}}\braket{e^{ik(x-y)},\psi_0}\right|\lesssim N^{-2}\sum_k\frac{\left(1+\frac{|k|^2}{N}\right)^{-1}}{1+|k|^2}\lesssim N^{-\frac{3}{2}}.
\end{align*}
Choosing $\ell$ of the order $\sqrt{N}$ yields
\begin{align}
\label{Eq:Convergence_leading_Coefficients}
    \left|N^2   (V_N - V_N R V_N )_{000,000}-b_\mathcal{M}(V)\right|\lesssim \frac{1}{N}.
\end{align}
For general $n=(n_1, n_2 ,n_3)$ and $m=(m_1 , m_2 ,m_3)$ with $n_1+n_2+n_3=m_1+m_2+m_3$ the estimates in Eq.~(\ref{Eq:Fourier_Xi_general}) and Eq.~(\ref{Eq:n_m_identity}) yield in a similar fashion the desired estimate.
\end{proof}

In Eq.~(\ref{Eq:Three_Scattering_Identity}) we saw that $RV_N$, an object defined on the torus $\Lambda$, is approximated by 
\begin{align*}
    \psi_n(x,y,z):=\chi^{\frac{{\sqrt{N}}}{\ell}} \! \left(x-y,x-z\right)\omega^{\sqrt{N}} (x-y,x-z),
\end{align*}
which involves the corresponding object $\omega$ defined on the full space. In the following Lemma \ref{Lemma:Correction_Coefficients} we make use of this correspondence again, to compare $\gamma_N,\mu_N$ and $\sigma_N$ with $\gamma,\mu$ and $\sigma$.

\begin{lem}
\label{Lemma:Correction_Coefficients}
    Let $\gamma_N,\mu_N$ and $\sigma_N$ be as in Eq.~(\ref{Eq:sigma}), Eq.~(\ref{Eq:mu}) and Eq.~(\ref{Eq:gamma}), and $\sigma(V)$, $\mu(V)$ and $\gamma(V)$ as in Eq.~(\ref{Eq:Definition_of_gamma}), Eq.~(\ref{Eq:Definition_of_mu}) and Eq.~(\ref{Eq:Definition_of_sigma}). Then
    \begin{align*}
          \gamma_N & =\gamma(V)\sqrt{N}+O_{N\rightarrow \infty}\! \left(N^{-\frac{1}{4}}\right),\\
         \mu_N & = \mu(V)\sqrt{N}+O_{N\rightarrow \infty}\left(1\right),\\
               \sigma_N & =\sigma(V)\sqrt{N}+O_{N\rightarrow \infty}\! \left(N^{\frac{1}{4}}\right),
    \end{align*}
    and there exists a constant $\lambda(V)>0$ such that for $0<\lambda\leq \lambda(V)$
    \begin{align}
    \label{Eq:Sign_Of_Coefficients_Together}
        \gamma(\lambda V) - \mu(\lambda V)-  \sigma(\lambda V)<0.
    \end{align}
    Furthermore, $\sigma_N$ and $\gamma_N$ are independent of the parameter $K$ from the definition of $\pi_K$ below Eq.~(\ref{Eq:Definition_pi_K}), and the limit $\mu(V)=\lim_N \frac{\mu_N}{\sqrt{N}}$ is independent of $K$ as well.
\end{lem}
\begin{proof}
    In order to analyse $\gamma_N$, let us denote with $L_i:L^2\! \left(\Lambda^4\right)\longrightarrow L^2\! \left(\Lambda^4\right)$ the linear map that exchanges the fist factor in the tensor product $L^2\! \left(\Lambda^4\right)\cong L^2\! \left(\Lambda\right)^{\otimes 4}$ with the $i$-th factor and observe that
\begin{align*}
   \sqrt{N}^{-1}\gamma_N= \frac{N^{\frac{7}{2}}}{6}\sum_{i=1}^3\braket{L_i\, 1\otimes (RV_N),(V_N\otimes 1)L_j\, 1\otimes (RV_N)}.
\end{align*}
Furthermore, recall the definition of $\psi_0$ from Eq.~(\ref{Eq:Definition_psi_n}) in the proof of Lemma \ref{Lem:Coefficient_Control_II} and define
\begin{align*}
   \gamma^{(\ell)}: & = \int_{\mathbb{R}^9}V(x,y)\! \left( \chi^{\frac{1}{\ell}}(x,z) \omega(x,z)\chi^{\frac{1}{\ell}}(y,z) \omega(y,z)+\frac{1}{2}\left|\chi^{\frac{1}{\ell}}(y,z) \omega(y,z)\right|^2 \right) \! \mathrm{d}x \mathrm{d}y \mathrm{d}z\\
  & =\frac{N^{\frac{7}{2}}}{6}\sum_{i,j=1}^3\braket{L_i\, 1\otimes \psi_0 ,(V_N\otimes 1)L_j \, 1\otimes \psi_0},
\end{align*}
where the second identity holds by a scaling argument for all $0<\ell<N$. We observe that by the permutation symmetry of $V_N$ we have $L_i V_N\otimes 1 L_i=V_N\otimes 1$ and therefore
\begin{align*}
   & N^{\frac{7}{2}} \Big|\braket{L_i\, 1\otimes \psi_0 ,(V_N\otimes 1)L_j \, 1\otimes \psi_0}\Big|\leq \sup_{i\in \{1,2,3\}}\braket{L_i\, 1\otimes \psi_0 ,(V_N\otimes 1)L_i \, 1\otimes \psi_0}\\
   & \ =   N^{\frac{7}{2}} \braket{ 1\otimes \psi_0 ,(V_N\otimes 1) 1\otimes \psi_0}=\int_{\mathbb{R}^9}V(x,y)\! \left|\chi^{\frac{1}{\ell}}(y,z) \omega(y,z)\right|^2\! \mathrm{d}x \mathrm{d}y \mathrm{d}z\lesssim 1.
\end{align*}
Using $L_i V_N\otimes 1 L_i=V_N\otimes 1$ again, together with the identity in Eq.~(\ref{Eq:Three_Scattering_Identity}) and the Cauchy-Schwarz inequality yields
\begin{align}
\label{Eq:Gamma_Top}
  &  \left|\sqrt{N}^{-1}\gamma_N-\gamma^{(\ell)}\right|\lesssim N^{\frac{7}{2}}\sqrt{\left\langle 1\otimes (R\xi_0(x \! - \! y,x \! - \! z)),(V_N\otimes 1) 1\otimes (R\xi_0(x \! - \! y,x \! - \! z))\right\rangle}\\
  \nonumber
  & \ \ \ \ \  +N^{\frac{7}{2}}\sqrt{\left\langle 1\otimes (RV_N-1)(1-Q^{\otimes 3}),(V_N\otimes 1) 1\otimes (RV_N-1)(1-Q^{\otimes 3})\right\rangle}.
\end{align}
Regarding the analysis of the term on the right side of Eq.~(\ref{Eq:Gamma_Top}), we observe that 
\begin{align*}
    \rho:=-2\Delta_{\mathcal{M}}\omega=V(1-\omega)
\end{align*}
satisfies $\nabla^k \rho\in L^1$ due to the regularity assumptions on $V$. Proceeding as in Eq.~(\ref{Eq:Delta_chi_Fourier}) we obtain the improved version of Eq.~(\ref{Eq:Fourier_Xi})
\begin{align}
\label{Eq:xi_estimate_Fourier}
    \left| \widehat{\xi}_0(K)\right|\lesssim N^{-2}  \min\left\{ \! \left(\frac{{\sqrt{N}}}{\ell |K|}\right)^2 \! ,1\right\}\left(1+\frac{|K|^2}{N}\right)^{-m},
\end{align}
Similar to Eq.~(\ref{Eq:Split_of_nabla_R}) we can write $\nabla^n R\xi_0 $, where $\xi_0$ is introduced in Eq.~(\ref{Eq:Definition_xi_n}), as the sum of terms of the form
\begin{align}
\label{Eq:Split_of_nabla_R_copy}
    Q^{\otimes 3}\nabla^{k_1} \! (V_N)Q^{\otimes 3}\dots  Q^{\otimes 3}\nabla^{k_m} \! (V_N)Q^{\otimes 3}\, \nabla^a R^{1-b}\xi_0 ,
\end{align}
where the coefficients satisfy $k_1+\dots + k_m + 2m +a +2b=n$ and either (I) that $b=1$, (II) that $b=0$ and $a=1$ or (III) that $b=0$, $a=0$ and $m\geq 1$ as well as $k_m=0$. In the following we are going to verify individually for the three cases (I)-(III) that the Fourier transform of the expression in Eq.~(\ref{Eq:Split_of_nabla_R_copy}) has an $L^{\infty}$ bound of the order $\frac{\sqrt{N}^{n}}{N^3 \ell^2}$ for $n\geq 4$, and consequently
\begin{align}
\label{Eq:R_xi_Bound}
    \left| \widehat{R\xi}_0(K) \right|\lesssim \frac{1}{ \ell^2  N^2|K|^2}\, \frac{N}{|K|^2} \! \! \left(1+\frac{|K|^2}{N}\right)^{-m}.
\end{align}
Regarding the case (I), we obtain using Eq.~(\ref{Eq:xi_estimate_Fourier}) and our regularity assumptions on $V$ by a direct computation in Fourier space, for $n\geq 4$
\begin{align*}
    \left|\left\langle e^{i K\cdot X},Q^{\otimes 3}\nabla^{k_1} \! (V_N)Q^{\otimes 3}\dots  Q^{\otimes 3}\nabla^{k_m} \! (V_N)Q^{\otimes 3}\, \nabla^a \xi_0 \right\rangle \right|\lesssim \frac{\sqrt{N}^{k_1+\dots +k_m +a+2m}}{N^2 \ell^2}=\frac{\sqrt{N}^{n}}{N^3 \ell^2}.
\end{align*}
Since the case (II) is similar to the case (III), let us directly have a look at the case (III), where we use the fact that $\|\sqrt{Q^{\otimes 3} V_N Q^{\otimes 3}}R\nabla\|\lesssim 1$ to obtain
\begin{align}
\label{Eq:Case_III}
   & \ \ \  \ \ \  \left|\left\langle e^{i K\cdot X},Q^{\otimes 3}\nabla^{k_1} \! (V_N)Q^{\otimes 3}\dots  Q^{\otimes 3}\nabla^{k_{m-1}} \! (V_N)Q^{\otimes 3}\, Q^{\otimes 3} V_N Q^{\otimes 3}\, R \xi_0 \right\rangle \right|\\
   \nonumber
    & \lesssim \left\|\sqrt{Q^{\otimes 3} V_N Q^{\otimes 3}}\, Q^{\otimes 3}\nabla^{k_{m-1}} \! (V_N)Q^{\otimes 3}\dots Q^{\otimes 3}\nabla^{k_{1}} \! (V_N)Q^{\otimes 3}\, e^{i K\cdot X}\right\|\, \left\|\frac{1}{\nabla}Q^{\otimes 3}\xi_0\right\|.
\end{align}
Since we have $\left\|\frac{1}{\nabla}Q^{\otimes 3}\xi_0\right\|\lesssim \frac{1}{N\ell^2 }$, we obtain together with Eq.~(\ref{Eq:Lem:Coefficient_Control_Addendum}) that the term in Eq.~(\ref{Eq:Case_III}) is bounded by $\frac{\sqrt{N}^{k_1+\dots + k_m+2(m-1)-2}}{N\ell^2 }=\frac{\sqrt{N}^{n}}{N^3 \ell^2 }$, which concludes the proof of Eq.~(\ref{Eq:R_xi_Bound}). Consequently 
\begin{align}
\label{Eq:Rxi_V_Rxi}
    N^{\frac{7}{2}}\braket{1\otimes (R\xi_0(x \! - \! y,x \! - \! z)),(V_N\otimes 1)1\otimes (R\xi_0(x \! - \! y,x \! - \! z))}\lesssim \frac{1}{\ell^4}.
\end{align}
Using $N^{\frac{7}{2}}\braket{1\otimes (RV_N-1)e^{ik(x_i-x_j)},(V_N\otimes 1)1\otimes (RV_N-1)e^{ik' (x_i-x_j)}}\lesssim N^{\frac{3}{2}}\left(1+\frac{|k-k'|^2}{N}\right)^{-2}$ by Lemma \ref{Lem:Coefficient_Control} for $i\neq j$, we further have
\begin{align}
\label{Eq:Diff_V_Diff}
 & N^{\frac{7}{2}}\braket{1\otimes (RV_N-1)(1-Q^{\otimes 3})\psi_0 ,(V_N\otimes 1)1\otimes (RV_N-1)(1-Q^{\otimes 3})\psi_0 }\\
 \nonumber
  & \ \ \ \ \lesssim \sum_{k,k'} N^{\frac{3}{2}}\left(1+\frac{|k-k'|^2}{N}\right)^{-2}\big|\braket{e^{ik(x-y)},\psi_0}\big|\, \big|\braket{e^{i{k'}(x-y)},\psi_0}\big|\\
  \nonumber
  & \ \ \ \  \lesssim N^{-\frac{5}{2}}\sum_{k,k'}\left(1+\frac{|k-k'|^2}{N}\right)^{-2}\frac{1}{1+|k|^2}\frac{1}{1+|k'|^2}\lesssim N^{-\frac{3}{2}}.   
\end{align}
By Eq.~(\ref{Eq:Gamma_Top}) we consequently obtain $\left|\sqrt{N}^{-1}\gamma_N-\gamma^{(\ell)}\right|\lesssim \ell^{-\frac{3}{2}}$ for $\ell\leq \sqrt{N}$. Note that for $\ell_1,\ell_2>0$ we can always pick an arbitrary $N\geq \max\{\ell_1,\ell_2\}^2$ yielding
\begin{align*}
 \left|\gamma^{(\ell_1)}-\gamma^{(\ell_2)}\right|\lesssim \frac{1}{\min\{\ell_1,\ell_2\}^{\frac{3}{2}} },   
\end{align*}
i.e. $\gamma^{(\ell)}$ is convergent with rate $\frac{1}{\ell^{\frac{3}{2}} }$, and by monotone convergence the limit is given by $\gamma(V)$.

In order to establish the convergence of $\sigma_N$, let us define $f_{N,\ell}:=(V_N\otimes 1) 1\otimes \psi_0$, where we keep track of the $N$ and $\ell$ dependence in our notation, and
\begin{align*}
    \sigma_{N,\ell}:=N^{\frac{7}{2}}\braket{f_{N,\ell},R_4 f_{N,\ell}},
\end{align*}
for $\ell<\sqrt{N}$ and let $R_4$ be defined above Eq.~(\ref{Eq:Def_R_4}). As a consequence of the operator inequality
\begin{align*}
    (V_N\otimes 1) R^{(N)}_4 (V_N\otimes 1) \leq  (V_N\otimes 1) \frac{1}{-\Delta} (V_N\otimes 1)\lesssim V_N\otimes 1,
\end{align*}
we obtain by Eq.~(\ref{Eq:Rxi_V_Rxi}) and Eq.~(\ref{Eq:Diff_V_Diff})
\begin{align*}
   & N^{\frac{7}{2}}\left\langle 1\otimes R\xi_0, (V_N\otimes 1) R^{(N)}_4 (V_N\otimes 1)1\otimes R\xi_0 \right\rangle\lesssim N^{\frac{7}{2}}\left\langle 1\otimes R\xi_0, (V_N\otimes 1) 1\otimes R\xi_0  \right\rangle\lesssim \frac{1}{\ell^2},\\
    & N^{\frac{7}{2}}\braket{1\otimes (RV_N-1)(1-Q^{\otimes 3})\psi_0,(V_N\otimes 1) R^{(N)}_4 (V_N\otimes 1) 1\otimes (RV_N-1)(1-Q^{\otimes 3})\psi_0}\\
    &  \lesssim N^{\frac{7}{2}}\braket{1\otimes (RV_N-1)(1-Q^{\otimes 3})\psi_0,(V_N\otimes 1)1\otimes (RV_N-1)(1-Q^{\otimes 3})\psi_0}\lesssim N^{-\frac{3}{2}}.
\end{align*}
Using the identity Eq.~(\ref{Eq:Three_Scattering_Identity}), this immediately implies for $\ell<\sqrt{N}$
\begin{align}
\label{Eq:Comparison_sigma}
& \left|\sqrt{N}^{-1}\sigma_N - \sigma_{N,\ell}\right|\lesssim \ell^{-\frac{3}{4}}.
\end{align}
To understand the dependence of $\sigma_{N,\ell}$ on the parameter $N$, recall the function 
\begin{align*}
f_\ell(x_1,x_2,x_3):  =V(x_1,x_2)\chi^{\frac{1}{\ell}}\! (x_2,x_3)\omega(x_2,x_3)    
\end{align*}
and $\eta_\ell:\mathbb{R}^9\longrightarrow \mathbb{R}$ from Lemma \ref{Lemma:Existence_of_solutions}, which solves in the sense of distributions
\begin{align}
\label{Eq:Differential_Equation}
    \left(-2\Delta_{\mathcal{M}_*}+\mathbb{V}\right)\eta_\ell=f_\ell.
\end{align}
By Lemma \ref{Lemma:Existence_of_solutions} we have the point-wise bound $0 \leq \eta_\ell\leq \eta_\ell^*$ with
\begin{align*}
    \eta^*_\ell(x):=\frac{1}{-2\Delta_{\mathcal{M}_*}}f_\ell(x)=\frac{\Gamma\! \left(\frac{9}{2}\right)}{28 \pi^{\frac{9}{2}}\mathrm{det}[M_*]}\int_{\mathbb{R}^9}\frac{f_\ell(y)\mathrm{d}y}{|\mathcal{M}_*^{-1} (x-y)|^{7}}.
\end{align*}
In the following let us write $x_1 g$ for the function $x\mapsto x_1 g(x)$. By Eq.~(\ref{Eq:Differential_Equation}) we obtain that $\rho_{\ell}:=-2\Delta_{\mathcal{M}_*}\eta_\ell$ satisfies the (uniform in $\ell$) bounds
\begin{align}
\label{Eq:Density_l}
\left\|\rho_\ell \right\|_{L^1(\mathbb{R}^9)}   & \leq\left\|f_\ell\right\|_{L^1(\mathbb{R}^9)}+\| \mathbb{V}\eta_\ell^*\|_{L^1(\mathbb{R}^9)}\lesssim 1,\\
\label{Eq:Density_l_position}
    \left\|x_1 \rho_\ell \right\|_{L^1(\mathbb{R}^9)} & \leq\left\|x_1 f_\ell\right\|_{L^1(\mathbb{R}^9)}+\| \mathbb{V}x_1 \eta_\ell^*\|_{L^1(\mathbb{R}^9)}\lesssim 1,
\end{align}
where we have used in the second estimates that $f_\ell(x)$ is compactly supported in the variables $x_1$ and $x_2$, and satisfies $\sup_{x_1,x_2}f_\ell(x)\lesssim \frac{1}{1+|x_3|^4}$, see the estimates on $\omega$ in \cite{NRT1}, and therefore $\left\|(1+|x_1|) f_\ell\right\|_{L^1(\mathbb{R}^9)}\lesssim 1$, as well as the fact that $x\mapsto \frac{1}{|x|^4}\mathbb{V}(x)\in L^1(\mathbb{R}^9)$ and hence
\begin{align*}
     & \| \mathbb{V} x \eta_\ell^*\|_{L^1(\mathbb{R}^9)}\lesssim \||x|^5\eta_\ell^*\|_\infty\lesssim  \sup_{x}|x|^5 \int_{\mathbb{R}^9}\frac{\mathrm{d}y}{|x-y|^{7} (|y_1|+|y_2|)^5 |y_3|^2}\lesssim 1,
\end{align*}
and similarly we obtain $ \| \mathbb{V} \eta_\ell^*\|_{L^1(\mathbb{R}^9)}\lesssim 1$. Using Eq.~(\ref{Eq:Differential_Equation}), we obtain the analogue estimates on the derivatives of $\rho_\ell$
\begin{align}
\label{Eq:Differential_Equation_Derivaties_Estimates}
    \left\|\nabla^k \rho_\ell \right\|_{L^1(\mathbb{R}^9)}+\left\|x_1 \nabla^k \rho_\ell \right\|_{L^1(\mathbb{R}^9)}   \lesssim 1.
\end{align}
Having $\eta_\ell$ at hand, we use a smooth function $\chi_*$ with $\chi_*(x)=1$ for $|x|_\infty\leq \frac{1}{2}$ and $\chi_*(x)=0$ for $|x|_\infty>\frac{2}{3}$, in order to define
\begin{align*}
    \Psi:=\chi_*(x_1-x_2,x_1-x_3 , x_1-x_4)\, \eta_\ell^{\sqrt{N}}\! (x_1-x_2,x_1-x_3 , x_1-x_4).
\end{align*}
Notably, the state $\Psi$ allows us to express
\begin{align}
\label{Eq:Psi_and_zeta}
    R_4 f_{N,\ell}=\Psi+R_4 \zeta+\left(R_4\mathbb{V}_N-1\right)(1-Q^{\otimes 4})\Psi,
\end{align}
with $\zeta:=[2\Delta_{\mathcal{M}_*},\chi_*]  \eta_{\ell}^{\sqrt{N}}=2\Delta_{\mathcal{M}_*}(\chi_*)\eta_{\ell}^{\sqrt{N}}+4\mathcal{M}_*^2 \nabla (\chi_*)\nabla \eta_{\ell}^{\sqrt{N}}$ and
\begin{align*}
    \mathbb{V}_N(x_1,x_2,x_3,x_4):=N\mathbb{V}\Big(\sqrt{N}(x_1-x_2),\sqrt{N}(x_1-x_3),\sqrt{N}(x_1-x_4)\Big).
\end{align*}
Proceeding as in the proof of Eq.~(\ref{Eq:Fourier_Xi}), we have by Eq.~(\ref{Eq:Differential_Equation_Derivaties_Estimates})
\begin{align}
\label{Eq:zeta_estimate}
    \big|\widehat{\zeta}(K)\big|\lesssim N^{-\frac{7}{2}}\min \! \left\{\frac{1}{|K|^2},1\right\}\left(1+\frac{|K|^2}{N}\right)^{-m}.
\end{align}
Using the fact that $\|\sqrt{Q^{\otimes 4}\mathbb{V}_N Q^{\otimes 4}}R_4\nabla \|\leq 1$ we furthermore obtain
\begin{align*}
     \big|\widehat{R_4\zeta}(K)\big|=\frac{ \left|\widehat{\zeta}(K) \! - \! \left\langle K, \mathbb{V}_N R_4  \zeta \right\rangle\right|}{|K|^2}\leq \frac{ \left|\widehat{\zeta}(K)\right| \! +  \! \sqrt{\left\langle K, \mathbb{V}_N K \right\rangle \left\langle \zeta, \frac{1}{-\Delta}   \zeta \right\rangle}}{|K|^2}\lesssim \frac{\max\{N^{-\frac{1}{4}},|K|^{-2}\}}{N^{\frac{7}{2}}|K|^2}.
\end{align*}
Using furthermore Eq.~(\ref{Eq:Split_of_nabla_R_copy}), we can utilize Eq.~(\ref{Eq:zeta_estimate}) to improve this result to
\begin{align}
\label{Eq:R_zeta_estimate}
    \big|\widehat{R_4\zeta}(K)\big|\lesssim \frac{\max\{N^{-\frac{1}{4}},|K|^{-1}\}}{N^{\frac{7}{2}}|K|^2}\! \left(1+\frac{|K|^2}{N}\right)^{-m}.
\end{align}
In analogy to Eq.~(\ref{Eq:xi_estimate_Fourier}), one can show that $\left|\widehat{\chi^{\frac{\sqrt{N}}{\ell}}\omega^{\sqrt{N}}} (k)\right|\lesssim N^{-2}\frac{1}{1+|k|^2}\! \left(1+\frac{|k|^2}{N}\right)^{-m}$, and therefore we have $\left|\widehat{f}_{N,\ell}(K)\right|\lesssim N^{-\frac{7}{2}}\left(1+\frac{|K|^2}{N}\right)^{-m}$, which yields together with Eq.~(\ref{Eq:R_zeta_estimate})
\begin{align}
\label{Eq:Sigma_Main_I}
    N^{\frac{7}{2}}\left|\braket{f_{N,\ell},R_4 \zeta}\right|\lesssim N^{-\frac{1}{4}}.
\end{align}
Furthermore,, in analogy to Eq.~(\ref{Eq:R_zeta_estimate}), we have the estimate
\begin{align*}
    \left|\widehat{\Psi}(K)\right|\lesssim \frac{1}{N^{\frac{7}{2}}|K|^2}\! \left(1+\frac{|K|^2}{N}\right)^{-m}.
\end{align*}
Denoting with $\mathbb{I}$ the set of all indices $K=(k_1,\dots ,k_4)$ such that $k_1+\dots + k_4=0$ and at least one of the indices satisfies $k_\alpha=0$, we obtain
\begin{align}
\nonumber
    & N^{\frac{7}{2}}\left|\braket{f_{N,\ell},\left(R_4\mathbb{V}_N-1\right)(1-Q^{\otimes 4})\Psi}\right|=N^{\frac{7}{2}}\left|\sum_{K\in \mathbb{I}}\widehat{\Psi}(K)\braket{(\mathbb{V}_NR_4-1)f_{N,\ell},e^{iK\cdot X}}\right|\\
    \label{Eq:Sigma_Main_II}
    & \ \ \ \ \ \ \ \ \ \  \lesssim N^{-\frac{7}{2}}\sum_{K\in \mathbb{I}}\frac{N^{\frac{3}{4}}}{|K|^2}\! \left(1+\frac{|K|^2}{N}\right)^{-3}\lesssim N^{-\frac{3}{4}},
\end{align}
where we used 
\begin{align*}
  |\braket{\mathbb{V}_NR_4 f_{N,\ell},e^{iK\cdot X}}|\leq \braket{e^{iK\cdot X},\mathbb{V}_N e^{iK\cdot X}}^{\frac{1}{2}}\braket{f_{N,\ell},\frac{1}{-\Delta}f_{N,\ell}}^{\frac{1}{2}}\lesssim N^{-\frac{7}{2}}N^{\frac{3}{4}}. 
\end{align*}
Applying Eq.~(\ref{Eq:Psi_and_zeta}), Eq.~(\ref{Eq:Sigma_Main_I}) and Eq.~(\ref{Eq:Sigma_Main_II}) therefore yields for $\ell<\frac{\sqrt{N}}{2}$
\begin{align*}
    \sigma_{N,\ell} & =N^{\frac{7}{2}}\braket{f_{N,\ell},R_4 f_{N,\ell}}=N^{\frac{7}{2}}\braket{f_{N,\ell},\Psi}+O_{N\rightarrow \infty}\! \left(N^{-\frac{1}{4}}\right)=\braket{f_\ell,\eta_\ell}+O_{N\rightarrow \infty}\! \left(N^{-\frac{1}{4}}\right)\\
    & =\sigma_\ell(V)+O_{N\rightarrow \infty}\! \left(N^{-\frac{1}{4}}\right).
\end{align*}
In combination with Eq.~(\ref{Eq:Comparison_sigma}) and the fact that $\sigma(V)=\lim_\ell \sigma_\ell(V)$, see Lemma \ref{Lemma:Existence_of_solutions}, we obtain that $|\sigma_\ell(V)-\sigma(V)|\lesssim \frac{1}{\sqrt{\ell}}$ and conclude
\begin{align*}
    \left|\sqrt{N}^{-1}\sigma_N-\sigma(V)\right|\lesssim N^{-\frac{1}{4}}.
\end{align*}

To establish the convergence of $\sqrt{N}^{-1}\mu_N$, let us recall the effective potential 
\begin{align*}
    V_\mathrm{eff}:
\begin{cases}
    & \mathbb{R}^3\longrightarrow \mathbb{R},\\
     & x\mapsto \int_{\mathbb{R}^3}V(x,y)(1-\omega(x,y))\, \mathrm{d}y,
\end{cases}
\end{align*}
and let $\theta$ solve $-2\Delta \theta=V_\mathrm{eff}$ with $\theta(x)\underset{|x|\rightarrow \infty}{\longrightarrow} 0$. Then 
\begin{align*}
    \mu(V)=\int_{\mathbb{R}^3}V_\mathrm{eff}(x)\theta(x)\, \mathrm{d}x.
\end{align*}
Applying the techniques developed in this proof so far, yields furthermore
\begin{align*}
    |\sqrt{N}^{-1}\mu_N - \mu(V)|\lesssim \frac{1}{\sqrt{N}}.
\end{align*}

Finally, in order to establish Eq.~(\ref{Eq:Sign_Of_Coefficients_Together}) let us denote with $\omega_\lambda$ the minimizer in Eq.~(\ref{Eq:Definition_of_b}) for the re-scaled potential $\lambda V$, which satisfies $0\leq \omega_\lambda\leq 1$ and $\omega_\lambda(x,y)\leq \frac{\lambda C(V)}{1+|x|^4+|y|^4}$, for a $V$ dependent, constant $C(V)>0$. Consequently $\lim_{\lambda\rightarrow 0}(1-\omega_\lambda)=1$, and hence we obtain by dominated convergence
\begin{align*}
    \lim_{\lambda\rightarrow 0}\frac{1}{\lambda^2}\mu(\lambda V) & =\lim_{\lambda\rightarrow 0}\int_{\mathbb{R}^{12}}\frac{V(x,u)V(y,v)(1-\omega_\lambda(x,u))(1-\omega_\lambda(y,v))}{8\pi |x-y|}\mathrm{d}u\mathrm{d}v\mathrm{d}x\mathrm{d}y\\
    & =\int_{\mathbb{R}^{12}}\frac{V(x,u)V(y,v)}{8\pi |x-y|}\mathrm{d}u\mathrm{d}v\mathrm{d}x\mathrm{d}y\in (0,\infty).
\end{align*}
This concludes the proof, since $\sigma(V)\geq 0$ and
\begin{align*}
    \frac{1}{\lambda^3}\gamma(\lambda V)\leq \frac{3C(V)^2}{2}\int_{\mathbb{R}^9}\frac{V(x,y)}{(1+|y|^4+|z|^4)^2}\mathrm{d}x\mathrm{d}y\mathrm{d}z<\infty.
\end{align*}
\end{proof}

Making use of Eq.~(\ref{Eq:Psi_and_zeta}) again, we can furthermore verify decay properties for the matrix entries of $T_4-1$ in momentum space in the subsequent Lemma \ref{Lem:Coefficient_Control_III}. 

\begin{lem}
\label{Lem:Coefficient_Control_III}
Recall the definition of the linear map $T_2$ in Eq.~(\ref{Eq:Definition_T_2}) and $T_4$ in Eq.~(\ref{Eq:Definition_T_4}). Then there exists a constant $C>0$ such that $|(T_2-1)_{jk,00}|\leq C N^{-1}\frac{\mathds{1}(j+k=0)}{|j|^2+|k|^2}\left(1+\frac{|j|^2+|k|^2}{N}\right)^{-1}$,
 \begin{align*}
   |(T_4-1)_{\ell i j k ,0000}|\leq C N^{-\frac{7}{2}}\frac{\mathds{1}(\ell+i+j+k=0)}{|\ell|^2+|i|^2+|j|^2+|k|^2}\left(1+\frac{|\ell|^2+|i|^2+|j|^2+|k|^2}{N}\right)^{-3}.  
 \end{align*}
\end{lem}
\begin{proof}
    For the purpose of verifying the bound on
    \begin{align*}
       (T_4-1)_{uijk,0000}=\braket{e^{iK\cdot X},R_4 (V_N\otimes 1) 1\otimes RV_N}
    \end{align*}
 with $K=(uijk)$, let us choose $\ell:=\frac{\sqrt{N}}{3}$ and recall the elements $\zeta$ and $\Psi$ from Eq.~(\ref{Eq:Psi_and_zeta}), and the set $\mathbb{I}$ above Eq.~(\ref{Eq:Sigma_Main_II}), in the proof of Lemma \ref{Lemma:Correction_Coefficients}. With these elements at hand, we can write
\begin{align}
\nonumber
   & \braket{e^{iK\cdot X},R_4 (V_N\otimes 1) 1\otimes RV_N} \! = \!\braket{e^{iK\cdot X},R_4 \zeta}  \! + \!  \braket{e^{iK\cdot X},\Psi}\! + \!  \sum_{K'\in \mathbb{I}}\braket{\mathbb{V}_N R_4 e^{iK\cdot X},e^{iK'\cdot X}}\braket{e^{iK'\cdot X},\Psi}\\
   \label{Eq:T_4_Control}
   & \ \ \ \ \ \ \ \ \ \ \ \ \ \  \ \ \ \ \ \ \ +\left\langle e^{iK\cdot X}, R_4 (V_N\otimes 1)1\otimes \big\{RV_N-\psi \big\}\right\rangle.
\end{align}
In the proof of Lemma \ref{Lemma:Correction_Coefficients} we have established
\begin{align*}
&  |\braket{e^{iK\cdot X},R_4 \zeta}|=|\widehat{R_4 \zeta}(K)|\lesssim \frac{1}{N^{\frac{7}{2}}|K|^2}\! \left(1+\frac{|K|^2}{N}\right)^{-m} ,\\
 & \ \ \ |\braket{e^{iK\cdot X},\Psi}|\lesssim \frac{1}{N^{\frac{7}{2}}|K|^2}\! \left(1+\frac{|K|^2}{N}\right)^{-m}.
\end{align*}
Regarding the sum over $\mathbb{I}$ we have
\begin{align*}
   & \left|\sum_{K'\in \mathbb{I}}\braket{\mathbb{V}_N R_4 e^{iK\cdot X},e^{iK'\cdot X}}\braket{e^{iK'\cdot X},\Psi}\right|=\frac{1}{|K|^2}\left|\sum_{K'\in \mathbb{I}}\braket{(\mathbb{V}_N - \mathbb{V}_N R_4\mathbb{V}_N) e^{iK\cdot X},e^{iK'\cdot X}}\braket{e^{iK'\cdot X},\Psi}\right|\\
    & \ \ \ \lesssim \frac{1}{N^{\frac{7}{2}}|K|^2}\sum_{K'\in \mathbb{I}}\frac{\left|\braket{(\mathbb{V}_N - \mathbb{V}_N R_4\mathbb{V}_N) e^{iK\cdot X},e^{iK'\cdot X}}\right|}{|K'|^2}\left(1+\frac{|K'|^2}{N}\right)^{-3}\lesssim \frac{1}{N^{\frac{7}{2}}|K|^2}. 
\end{align*}
Regarding the final term in Eq.~(\ref{Eq:T_4_Control}), we observe that we have the estimate
\begin{align*}
  &  \left|\left\langle e^{iK\cdot X}, R_4 (V_N\otimes 1)1\otimes \big\{RV_N-\psi \big\}\right\rangle\right|= \frac{1}{|K|^2}\left|\left\langle  (1- R_4\mathbb{V}_N)e^{iK\cdot X}, (V_N\otimes 1)1\otimes \big\{RV_N-\psi \big\}\right\rangle\right|\\
    & \leq \frac{\sqrt{\left\langle 1 \! \otimes  \! \big\{RV_N \! - \! \psi \big\},  (V_N\otimes 1)1 \! \otimes  \! \big\{RV_N \! - \! \psi \big\}\right\rangle }}{|K|^2} \sqrt{\left\langle (1 \! -  \! R_4\mathbb{V}_N )e^{iK\cdot X},V_N\otimes 1 (1 \! -  \!  R_4\mathbb{V}_N) e^{iK\cdot X}\right\rangle }.
\end{align*}
By Eq.~(\ref{Eq:Rxi_V_Rxi}) and Eq.~(\ref{Eq:Diff_V_Diff}) we know that $\left\langle 1 \! \otimes  \! \big\{RV_N \! - \! \psi \big\},  (V_N\otimes 1)1 \! \otimes  \! \big\{RV_N \! - \! \psi \big\}\right\rangle \lesssim N^{-5}$,
\begin{align*}
  & \ \ \  \left\langle (1 \! -  \!  R_4\mathbb{V}_N)e^{iK\cdot X},V_N\otimes 1 (1 \! -  \!R_4 \mathbb{V}_N ) e^{iK\cdot X}\right\rangle\leq 2\left\langle  e^{iK\cdot X},V_N\otimes 1   e^{iK\cdot X}\right\rangle\\
   & + 2\left\langle   R_4\mathbb{V}_N e^{iK\cdot X},V_N\otimes 1 R_4  \mathbb{V}_N  e^{iK\cdot X}\right\rangle\lesssim \frac{1}{N^2}+\left\langle   \mathbb{V}_N e^{iK\cdot X},R_4  \mathbb{V}_N  e^{iK\cdot X}\right\rangle\\
   & \ \ \  \ \ \ \lesssim \frac{1}{N^2}+\left\langle   e^{iK\cdot X},  \mathbb{V}_N  e^{iK\cdot X}\right\rangle \lesssim \frac{1}{N^2}.
\end{align*}
Finally we note that the bound on $T_2$ is an immediate consequence of the regularity of $V$ and the bounds on $RV_N$ established in Lemma \ref{Lem:Coefficient_Control}.
\end{proof}

\appendix 
\section{Auxiliary Results}
\label{Sec:Auxiliary Results}

In the following we establish comparability results between transformed and non-transformed quantities. The first result in this direction, Lemma \ref{Lem:Particle_Number_Creation}, establishes that the unitarily transformed powers of the particle number operator $\mathcal{N}$, w.r.t. the transformations $U_s$ and $W_s$, are again of the same order as the bare powers in $\mathcal{N}$.

\begin{lem}
\label{Lem:Particle_Number_Creation}
   Let $U_s$ be the unitary map defined below Eq.~(\ref{Eq:Generator}) and $W_s$ the one defined below Eq.~(\ref{Eq:Four_particle_generator}). Then there exists for all $m\in \mathbb{N}$ constants $C_m>0$, such that
   \begin{align}
   \label{Eq:Groenwall_I}
       U_{-s} \, \mathcal{N}^m U_s\leq e^{C_m |s|} (\mathcal{N}+1)^m,\\
         \label{Eq:Groenwall_II}
       W_{-s} \, \mathcal{N}^m W_s\leq e^{C_m |s|} (\mathcal{N}+1)^m.
   \end{align}
\end{lem}
\begin{proof}
Let us recall the definition of the generator $\mathcal{G}^\dagger - \mathcal{G}$ with 
\begin{align*}
& \mathcal{G}=\frac{1}{6}\sum_{i j k}\eta_{ijk}\, a_i^\dagger a_j^\dagger a_k^\dagger a_0^3  ,\\ 
& \eta_{ijk}=(T-1)_{i j k, 0 0 0}
\end{align*}
of the unitary group $U_t$ from Eq.~(\ref{Eq:Generator}). As a consequence of the bounds on $T$ from Lemma \ref{Lem:Coefficient_Control}, we have 
\begin{align*}
   \pm \left(\mathcal{G}+\mathcal{G}^\dagger\right)\leq \frac{1}{6}\sum_{i}\left(a_i^\dagger a_i +\left(\sum_{jk}\eta_{ijk}\, a_j^\dagger a_k^\dagger a_0^3\right)^\dagger \left(\sum_{jk}\eta_{ijk}\, a_j^\dagger a_k^\dagger a_0^3\right)\right)\lesssim \mathcal{N}+\frac{\mathcal{N}^2}{N}\lesssim \mathcal{N}+1.
\end{align*}
Together with $0\leq (x+n+3)^k-(x+n)^k\leq C_{n,k} (x+3)^{k-1}$ for a suitable $C_{n,k}>0$, we obtain
\begin{align*}
  &  \ \ \ \ \ \ \ \ \  \ \  [\mathcal{G},(\mathcal{N}+3)^m]+\mathrm{H.c.}=- \left((\mathcal{N}+3)^m- \mathcal{N}^m  \right)\mathcal{G}+\mathrm{H.c.}\\
    & =- \sqrt{(\mathcal{N}+3)^m-\mathcal{N}^m}\left(\mathcal{G}+\mathcal{G}^\dagger \right)\sqrt{(\mathcal{N}+6)^m-(\mathcal{N}+3)^m}+\mathrm{H.c.}\lesssim (\mathcal{N}+3)^m.
\end{align*}
Applying Duhamel's formula then yields
    \begin{align*}
         U_{-t} \, (\mathcal{N}+3)^m U_t-(\mathcal{N}+3)^m=\int_0^t U_{-s}[\mathcal{G},(\mathcal{N}+3)^m]U_s\, \mathrm{d}s+\mathrm{H.c.}\lesssim \int_0^1 U_{-s}(\mathcal{N}+3)^m U_s\, \mathrm{d}s.
    \end{align*}
   Consequently Grönwall's inequality gives us
   \begin{align*}
      U_{-t} \, (\mathcal{N}+3)^m U_t\leq e^{C |t|}(\mathcal{N}+3)^m 
   \end{align*}
 for a suitable constant  $C>0$, which concludes the proof of Eq.~(\ref{Eq:Groenwall_I}). The proof of Eq.~(\ref{Eq:Groenwall_II}) follows analogously from $ \pm \left(\mathcal{G}_2+\mathcal{G}_2^\dagger\right)\lesssim \mathcal{N}+1$ and 
 \begin{align*}
     \pm \left(\mathcal{G}_4+\mathcal{G}_4^\dagger\right)\lesssim N^{-3}\left(N (\mathcal{N}+1)^3+\mathcal{N}^{\frac{5}{2}}\right)\lesssim \mathcal{N}+1,
 \end{align*}
 where we have used Lemma \ref{Lem:Coefficient_Control_III} in order to control the coefficients of $T_2$ and $T_4$.
\end{proof}

In the subsequent Lemma \ref{lem:Comparison} we are going to compare the kinetic energy $ \sum_{k}|k|^{2\tau}a_k^\dagger a_k$ in the operators $a_k$ with a fractional Laplace $(-\Delta)^{\tau}$, with the corresponding expression in the variables $c_k$.

\begin{lem}
\label{lem:Comparison}
    Let $0 \leq \tau\leq 1$ and $0\leq \sigma <\frac{1}{2}$. Then $\sum_{k}|k|^{2\sigma}(c_k-a_k)(c_k-a_k)^\dagger\lesssim \frac{1}{N}\mathcal{N}^2$, and furthermore we have for integers $s\geq 0$
\begin{align}
\label{Eq:Kinetic_Comparison}
        \sum_{k}|k|^{2\tau}a_k^\dagger\, \mathcal{N}^s a_k \lesssim \sum_{k}|k|^{2\tau} c_k^\dagger \mathcal{N}^s c_k+\frac{1}{N}\mathcal{N}^{s+2} +N^{\tau}\left(\mathcal{N}+1\right)^s.
\end{align}
\end{lem}
\begin{proof}
   Let us define $\left(G^{(I,I')}_\tau\right)_{ij,i' j'}:=\frac{1}{4}\sum_{k}|k|^{2\tau} \overline{(T-1)_{i'j'k,I' }}(T-1)_{ijk,I}$ for
   \begin{align*}
   I,I'\in \mathcal{I}:=\{(0,0,0)\}\cup \bigcup_{0<|\ell|\leq K}\{(\ell,0,0),(0,\ell,0),(0,0,\ell)\}    
   \end{align*}
as well as for $0\leq \gamma\leq 1$ the operator valued vector and matrix
\begin{align*}
        (\Phi_\tau)_{jk}: & =\left(|j|^{2\tau}+|k|^{2\tau}\right)^{\frac{1}{2}}a_j a_k,\\
        \left(\Upsilon^{(I,I')}_{ \gamma,\tau}\right)_{jk,j'k'}: & =\left(\mathcal{K}_{ \gamma,2}^{-\frac{1}{2}} G^{(I,I')}_\tau \mathcal{K}_{ \gamma,2}^{-\frac{1}{2}}\right)_{j'k',jk}a_{I'_1} a_{I'_2} a_{I'_3} \mathcal{N}^s a_{I_1}^\dagger a_{I_2}^\dagger a_{I_3}^\dagger ,   
\end{align*}
   with $\mathcal{K}_{ \gamma,2}:=(-\Delta_x)^\gamma + (-\Delta_y)^\gamma$. With these definitions at hand we obtain
\begin{align*}
    \sum_{k}|k|^{2\tau}(c_k-a_k)\mathcal{N}^s (c_k-a_k)^\dagger=\frac{1}{2}\sum_{I,I'}\Phi_\gamma^\dagger \left(\Upsilon^{(I,I')}_{\gamma,\tau}+\mathrm{H.c.}\right) \Phi_\gamma.
\end{align*}
    For $\gamma>\tau-\frac{1}{2}$ we have by the estimates from Lemma \ref{Lem:Coefficient_Control} that 
    \begin{align*}
    \|\mathcal{K}_{\gamma,2}^{-\frac{1}{2}} G^{(I,I')}_\tau \mathcal{K}_{\gamma,2}^{-\frac{1}{2}}\|\lesssim N^{-4}.    
    \end{align*}
Together with 
    \begin{align*}
        \big\|\left(\mathcal{N}+1\right)^{-\frac{s}{2}}a_{I'_1} a_{I'_2} a_{I'_3} \mathcal{N}^s a_{I_1}^\dagger a_{I_2}^\dagger a_{I_3}^\dagger\left(\mathcal{N}+1\right)^{-\frac{s}{2}}\big\|\lesssim N^3 
    \end{align*}
on the $N$ particle sector, we obtain $\left(\Upsilon^{(I,I')}_{\gamma,\tau}+\mathrm{H.c.}\right) \leq \frac{C}{N}\left(\mathcal{N}+1\right)^s$ for $\gamma>\tau-\frac{1}{2}$ and a suitable constant $C$. Using Cauchy-Schwarz we therefore have
    \begin{align*}
        \sum_{k}|k|^{2\tau}(c_k-a_k)\mathcal{N}^s (c_k-a_k)^\dagger\lesssim \frac{1}{N}\Phi_\gamma^\dagger \left(\mathcal{N}+1\right)^s\Phi_\gamma=2C\sum_k |k|^{2\gamma} a_k^\dagger \frac{\mathcal{N}^{s+1}}{N} a_k.
    \end{align*}
    Applying this result for $\tau':=\sigma$, $\gamma':=0$ and $s':=0$, yields the first claim of the Lemma 
    \begin{align*}
    \sum_{k}|k|^{2\sigma}(c_k-a_k)(c_k-a_k)^\dagger\leq \frac{1}{N}\mathcal{N}^2.    
    \end{align*}
Concerning Eq.~(\ref{Eq:Kinetic_Comparison}), we have
    \begin{align*}
        \sum_{k}|k|^{2\tau}a_k^\dagger\, \mathcal{N}^s a_k\leq 2\sum_{k}|k|^{2\tau}c_k^\dagger\, \mathcal{N}^s c_k+2\sum_{k}|k|^{2\tau}(c_k-a_k)^\dagger\, \mathcal{N}^s (c_k-a_k),
    \end{align*}
    and furthermore we can express
    \begin{align}
    \label{Eq:Increment_Analysis}
        \sum_{k}|k|^{2\tau}(c_k \! - \! a_k)^\dagger \mathcal{N}^s (c_k \! - \! a_k) \! = \! \sum_{I,I'} \! \left( \! f^{I,I'} X^{I,I'}_0 \! + \! \sum_{k\neq 0} \! g^{I,I'}_k a_k^\dagger X^{I,I'}_1 a_k \! + \! \frac{1}{2}\Phi_\gamma^\dagger \left(\widetilde{\Upsilon}^{(I,I')}_{\gamma,\tau} \! + \! \mathrm{H.c.} \! \right) \!  \Phi_\gamma \! \right)
    \end{align}
    with 
    \begin{align*}
        X^{I,I'}_0: & =a_{I_1}^\dagger a_{I_2}^\dagger a_{I_3}^\dagger \left(\mathcal{N}^s+2s\mathcal{N}^{s-1}+s(s-1)\mathcal{N}^{s-2}\right) a_{I'_1} a_{I'_2} a_{I'_3},\\
        X^{I,I'}_1: & =a_{I_1}^\dagger a_{I_2}^\dagger a_{I_3}^\dagger \left(2\mathcal{N}^s+4s\mathcal{N}^{s-1}+4s(s-1)\mathcal{N}^{s-2}+2s(s-1)(s-2)\mathcal{N}^{s-3}\right) a_{I'_1} a_{I'_2} a_{I'_3},\\
         X_2: & =a_{I_1}^\dagger a_{I_2}^\dagger a_{I_3}^\dagger \Big(\mathcal{N}^s+4s\mathcal{N}^{s-1}+6s(s-1)\mathcal{N}^{s-2}+4s(s-1)(s-2)\mathcal{N}^{s-3}\\
       & \ \  \ \  \ \  \ \  \ \ +s(s-1)(s-2)(s-3)\mathcal{N}^{s-4}\Big) a_{I'_1} a_{I'_2} a_{I'_3},\\
       f^{I,I'}: & =\sum_{ij}\left[\left(G^{I,I'}_\tau\right)_{ij,ij}+\left(G^{I,I'}_\tau\right)_{ij,ji}\right],\\
       g_j^{I,I'}: & =\sum_{i}\left[\left(G^{I,I'}_\tau\right)_{ij,ij}+\left(G^{I,I'}_\tau\right)_{ij,ji}\right],\\
        \widetilde{\Upsilon}^{(I,I')}_{\gamma,\tau}: & =\left(\mathcal{K}_{\gamma,2}^{-\frac{1}{2}} G^{(I,I')}_\tau \mathcal{K}_{\gamma,2}^{-\frac{1}{2}}\right)_{j'k',jk} X_2.
    \end{align*}
    Following the proof of the first part of the Lemma, we obtain for $\gamma>\tau-\frac{1}{2}$
    \begin{align*}
        \frac{1}{2}\sum_{I,I'}\Phi_\gamma^\dagger \left(\widetilde{\Upsilon}^{(I,I')}_{\gamma,\tau}+\mathrm{H.c.}\right) \Phi_\gamma\lesssim \sum_k |k|^{2\gamma} a_k^\dagger \frac{\mathcal{N}^{s+1}}{N} a_k.
    \end{align*}
    Using Lemma \ref{Lem:Coefficient_Control} again, yields $|f|\lesssim N^{\tau-3}$ and $|g_j|\lesssim N^{\max\{\tau-\frac{1}{2},0\}}-4$, and consequently
    \begin{align*}
        f^{I,I'} X^{I,I'}_0+\sum_{k\neq 0}g^{I,I'}_k a_k^\dagger X^{I,I'}_1 a_k\lesssim N^{\tau}\mathcal{N}^s+N^{\max\{\tau-\frac{1}{2},0\}}\frac{\mathcal{N}^{s+1}}{N}\lesssim N^{\tau}\mathcal{N}^s.
    \end{align*}
Summarizing what we have so far we obtain for $\gamma>\tau-\frac{1}{2}$
\begin{align*}
        \sum_{k}|k|^{2\tau}a_k^\dagger\, \mathcal{N}^s a_k \lesssim \sum_{k}|k|^{2\tau } c_k^\dagger \mathcal{N}^s c_k+\frac{1}{N}\sum_{k}|k|^{2\gamma}a_k^\dagger\, \mathcal{N}^{s+1} a_k+N^{\tau}\left(\mathcal{N}+1\right)^s.
\end{align*}
Choosing $\gamma:=\max\{\tau-\frac{1}{3},0\}$ and iterating this equation at most two times with $\tau':=\max\{\tau-\frac{1}{3},0\}$ and $\gamma':=\max\{\gamma-\frac{1}{3},0\}$, and using $\mathcal{N}\leq N$, yields the desired statement.
\end{proof}

As a consequence of Lemma \ref{lem:Comparison}, we can  estimate monomials in the operators $a_k$ and $a_k^\dagger$ by the kinetic energy in the variables $c_k$ and powers of the particle number operator $\mathcal{N}$, see the subsequent Corollary \ref{Cor:Operator_Estimates}.

\begin{cor}
\label{Cor:Operator_Estimates}
Let $\mathcal{K}_{\tau,t}:=\sum_{i=1}^t (-\Delta_{x_i})^\tau $. Given $0\leq  \tau\leq 1$, and integers $s,t\geq 1$ and $\alpha,\beta\geq 0$, there exist $\delta>0$ and $C>0$, such that for $\epsilon>0$
    \begin{align*}
        &\pm \left(\sum_{i_1\dots i_s,j_1\dots j_t} \!  \!  \!  \!  \! G_{i_1\dots i_s,j_1\dots j_t}a_{j_t}^\dagger\dots a_{j_1}^\dagger X\ a_{i_1}\dots a_{i_s}+\mathrm{H.c.}\right)\leq C\left\|\mathcal{K}_{\tau,t}^{-\frac{1}{2}}G\mathcal{K}_{\tau,s}^{-\frac{1}{2}}\right\| \big\| \mathcal{N}^{-\frac{\beta}{2}}X\mathcal{N}^{-\frac{\alpha}{2}}\big\| \\
        & \ \ \  \ \ \times \left\{\sum_{k}|k|^2 c_k^\dagger \left(\epsilon \mathcal{N}^{s+\alpha-1}+\epsilon^{-1}\mathcal{N}^{t+\beta-1}\right) c_k + \left(\mathcal{N}+N^\tau \right) \! \left(\epsilon \mathcal{N}^{s+\alpha-1}+\epsilon^{-1}\mathcal{N}^{t+\beta-1}\right)\right\},
    \end{align*}
    where $G:\left(\mathrm{ran}Q\right)^{\otimes s}\longrightarrow \left(\mathrm{ran}Q\right)^{\otimes t}$ and $X:\mathcal{F}\!\left(L^2(\Lambda)\right)\longrightarrow \mathcal{F}\!\left(L^2(\Lambda)\right)$. In case $s=0$
    \begin{align*}
       &\pm \left(\sum_{j_1\dots j_t} \!   \! G_{j_1\dots j_t}a_{j_t}^\dagger\dots a_{j_1}^\dagger X +\mathrm{H.c.}\right) \leq C\left\|\mathcal{K}_{\tau,t}^{-\frac{1}{2}}G\right\| \big\| \mathcal{N}^{-\frac{\beta}{2}}X\mathcal{N}^{-\frac{\alpha}{2}}\big\| \\
        & \ \ \ \times \left\{\epsilon \mathcal{N}^\alpha+\epsilon^{-1} \sum_{k}|k|^2 c_k^\dagger  \mathcal{N}^{t+\beta-1} c_k +\epsilon^{-1}  \left(\mathcal{N}+N^\tau \right) \mathcal{N}^{t+\beta-1}\right\}.
    \end{align*}
\end{cor}
\begin{proof}
    Let us define for $s,t\geq 1$ the operator valued vector and operator valued matrix
    \begin{align}
    \label{Eq:Def_Phi}
      \left(\Phi_{\tau,s}\right)_{k_1\dots k_s}: & =\left(|k_1|^{2\tau}+\dots +|k_s|^{2\tau}\right)^{\frac{1}{2}}a_{k_1}\dots a_{k_s},\\
      \nonumber
      \Upsilon_{i_1\dots i_s,j_1\dots j_t}: & =\left(\mathcal{K}_{\tau,t}^{-\frac{1}{2}}G\, \mathcal{K}_{\tau,s}^{-\frac{1}{2}}\right)_{i_1\dots i_s,j_1\dots j_t} \mathcal{N}^{-\frac{\beta}{2}}X\mathcal{N}^{-\frac{\alpha}{2}},
    \end{align}
which allow us to represent
    \begin{align*}
        \sum_{i_1\dots i_s,j_1\dots j_t} \!  \!  \!  \!  \! G_{i_1\dots i_s,j_1\dots j_t}a_{j_t}^\dagger\dots a_{j_1}^\dagger X\ a_{i_1}\dots a_{i_s}=\Phi_{\tau,t}^\dagger \mathcal{N}^{\frac{\beta}{2}}\Upsilon\mathcal{N}^{\frac{\alpha}{2}} \Phi_{\tau,s}.
    \end{align*}
    Using the fact that $\|\Upsilon\|\leq \left\|\mathcal{K}_{\tau,t}^{-\frac{1}{2}}G\, \mathcal{K}_{\tau,s}^{-\frac{1}{2}}\right\|\big\| \mathcal{N}^{-\frac{\beta}{2}}X\mathcal{N}^{-\frac{\alpha}{2}}\big\|$, we obtain by Cauchy-Schwarz
    \begin{align*}
      & \ \ \   \pm  \left(\Phi_t^\dagger \Upsilon \Phi_s+\mathrm{H.c.}\right)\leq \left\|\mathcal{K}_{\tau,t}^{-\frac{1}{2}}G\, \mathcal{K}_{\tau,s}^{-\frac{1}{2}}\right\|\big\| \mathcal{N}^{-\frac{\beta}{2}}X\mathcal{N}^{-\frac{\alpha}{2}}\big\|\left(\epsilon \Phi_{\tau,t}^\dagger \mathcal{N}^\beta \Phi_{\tau,t}+\epsilon^{-1}\Phi_{\tau,s}^\dagger \mathcal{N}^\alpha \Phi_{\tau,s}\right)\\
        &=\left\|\mathcal{K}_{\tau,t}^{-\frac{1}{2}}G\, \mathcal{K}_{\tau,s}^{-\frac{1}{2}}\right\|\big\| \mathcal{N}^{-\frac{\beta}{2}}X\mathcal{N}^{-\frac{\alpha}{2}}\big\|\left(\epsilon s\sum_{k}|k|^{2\tau}a_k^\dagger\, \mathcal{N}^{s+\alpha -1} a_k+\epsilon^{-1}\, t\sum_{k}|k|^{2\tau}a_k^\dagger \, \mathcal{N}^{t+\beta -1} a_k\right).
    \end{align*}
Since $\mathcal{N}\leq N$ we furthermore have by Lemma \ref{lem:Comparison} 
    \begin{align*}
        \sum_{k}|k|^{2\tau}a_k^\dagger\, \mathcal{N}^{s+\alpha-1} a_k \lesssim \sum_{k}|k|^2 c_k^\dagger \mathcal{N}^{s+\alpha -1} c_k +(\mathcal{N}+1)^{s+\alpha-1}\left(N^{\tau}+\mathcal{N}\right).
    \end{align*}
\end{proof}

\section{IMS Localisation}
\label{Sec:IMS Localisation}
Adapting the  localization procedure presented in \cite[Theorem A.1]{LS} in the form state in \cite[Proposition 6.1]{LNSS} for the following Lemma \ref{Lem:IMS} allows us to lift Bose-Einstein condensation in expectation, $\braket{\Psi,\mathcal{N}\Psi}\leq C N^\alpha$ with $\alpha<1$, to Bose-Einstein condensation in the spectral sense $\mathds{1}(\mathcal{N}\leq C N^\alpha)\Phi=\Phi$, without changing the energy significantly, see Eq.~(\ref{Eq:Energy_After_IMS}). Clearly the second notion of Bose-Einstein condensation is a much stronger condition.

\begin{lem}
\label{Lem:IMS}
    Let $\Psi$ satisfy $\braket{\Psi,H_N \Psi}= E_N+\delta$ with $\delta\leq N$. Then there exists a constant $C>0$, such that there exists for all $1\leq M\leq N$ states $\Phi$ satisfying $\mathds{1}(\mathcal{N}\leq M)\Phi=\Phi$ and
    \begin{align}
    \label{Eq:Energy_After_IMS}
    \braket{\Phi, H_N \Phi}\leq E_N+ C \left(1-\frac{2\braket{\Psi,\mathcal{N}\Psi}}{M}\right)^{-1}\! \left(\frac{\sqrt{N}}{M}+\frac{N}{M^2}+\delta\right).
    \end{align}
    Furthermore, there exists a state $\widetilde{\Phi}$ such that $\mathds{1}\! \left(\mathcal{N}>\frac{M}{2}\right)\widetilde{\Phi}=\widetilde{\Phi}$ and
    \begin{align}
        \label{Eq:Energy_After_IMS_II}
    \braket{\Psi, \mathcal{N}\Psi}  \leq \braket{\Phi,\mathcal{N} \Phi }+\frac{C N}{\braket{\widetilde{\Phi},H_N \widetilde{\Phi}}-E_N}\! \left(\frac{\sqrt{N}}{M}+\frac{N}{M^2}+\delta \right).
    \end{align}
\end{lem}
\begin{proof}
    In the following let $f,g:\mathbb{R}\rightarrow [0,1]$ be smooth functions satisfying $f^2+g^2=1$ and $f(x)=1$ for $x\leq \frac{1}{2}$, as well as $f(x)=0$ for $x\geq 1$, and let us define $ m:= \|f\! \left(\frac{\mathcal{N}}{M}\right)\Psi\|^2$ and
    \begin{align*}
      \Phi: & =\frac{1}{\sqrt{m}} f\! \left(\frac{\mathcal{N}}{M}\right)\Psi, \\  
        \widetilde \Phi: & =\frac{1}{\sqrt{1-m}} g\! \left(\frac{\mathcal{N}}{M}\right)\Psi.
    \end{align*}
    Note that $\|\Phi\|=\|\widetilde \Phi\|=1$, $0\leq m \leq 1$ and clearly we have $\mathds{1}(\mathcal{N}\leq M)\Phi=\Phi$ and $\mathds{1}\! \left(\mathcal{N}>\frac{M}{2}\right)\widetilde{\Phi}=\widetilde{\Phi}$. Making use of the algebraic identity
    \begin{align*}
        H_N=f\! \left(\frac{\mathcal{N}}{M}\right) H_N f\! \left(\frac{\mathcal{N}}{M}\right)+g\! \left(\frac{\mathcal{N}}{M}\right) H_N g\! \left(\frac{\mathcal{N}}{M}\right)+\mathcal{E}
    \end{align*}
    with the residual term 
    \begin{align*}
        \mathcal{E}:=\frac{1}{2}\left[f\! \left(\frac{\mathcal{N}}{M}\right),\left[H_N, f\! \left(\frac{\mathcal{N}}{M}\right)\right]\right]+\frac{1}{2}\left[g\! \left(\frac{\mathcal{N}}{M}\right),\left[H_N, g\! \left(\frac{\mathcal{N}}{M}\right)\right]\right],
    \end{align*}
    we obtain
    \begin{align*}
       m \braket{\Phi,H_N \Phi}+(1-m)\braket{\widetilde \Phi,H_N \widetilde \Phi}= E_N+\delta-\braket{\Psi,\mathcal{E}\Psi}.
    \end{align*}
    In order to estimate $\braket{\Psi, \mathcal{E}\Psi}$, let $\pi^0$ denote the projection onto the constant function in $L^2(\Lambda)$ and $\pi^1:=1-\pi^0$. Then we can rewrite $\mathcal{E}$ as
    \begin{align*}
        \mathcal{E}=\frac{1}{4M^2}\sum_{I,J\in \{0,1\}^3}\sum_{ijk,\ell m n}(\pi^{I_1}\pi^{I_2}\pi^{I_3} V_N \pi^{J_1}\pi^{J_2}\pi^{J_3})_{ijk, \ell m n }a_k^\dagger a_j^\dagger a_i^\dagger X_{I,J} a_\ell a_m a_n,
    \end{align*}
    with 
    \begin{align*}
        X_{I,J}:=M^2 \left[f\! \left(\frac{\mathcal{N}+\, \#_I}{M}\right)-f\! \left(\frac{\mathcal{N}+\, \#_J}{M}\right)\right]^2+M^2 \left[g\! \left(\frac{\mathcal{N}+\, \#_I}{M}\right)-g\! \left(\frac{\mathcal{N}+\, \#_J}{M}\right)\right]^2
    \end{align*}
    and $\#_I$ counting the number of indices in $I$ that are equal to $1$. Using $0\leq X_{I,J}\leq X$, where 
    \begin{align*}
        X:=\left(\|\nabla f\|^2+\|\nabla g\|^2\right)\mathds{1}(\mathcal{N}\leq M),
    \end{align*}
 we obtain by the Cauchy-Schwarz inequality
    \begin{align}
    \label{Eq:IMS_Error}
        \pm \mathcal{E}\lesssim \frac{1}{4M^2}\sum_{I\in \{0,1\}^3}\sum_{ijk,\ell m n}(\pi^{I_1}\pi^{I_2}\pi^{I_3} V_N \pi^{I_1}\pi^{I_2}\pi^{I_3})_{ijk, \ell m n }a_k^\dagger a_j^\dagger a_i^\dagger X a_\ell a_m a_n.
    \end{align}
    In the following we want to show that for any $I\in \{0,1\}^3$, the $\Psi$-expectation value of the corresponding term apperaring in the sum on the right side of Eq.~(\ref{Eq:IMS_Error}) is of the order $\sqrt{N}M+N$. For $I=(0,0,0)$ we have
    \begin{align*}
        (V_N)_{000,000}(a_0^\dagger)^3 X a_0^3 \lesssim N^{-2}\|X\|(a_0^\dagger)^3 a_0^3\leq \left(\|\nabla f\|^2+\|\nabla g\|^2\right)N .
    \end{align*}
    Similarly,
    \begin{align*}
        \sum_{k\neq 0}(V_N)_{001,001}(a_0^\dagger)^2 a_k^\dagger X a_k a_0^2\lesssim M\leq N
    \end{align*}
in the case $I=(0,0,1)$. Regarding the case $I=(0,1,1)$, let us first observe that we have the upper bound
\begin{align}
\label{Eq:V_2,2}
    \sum_{jk,mn\neq 0}(V_N)_{0jk,0mn}a_0^\dagger a_k^\dagger a_j^\dagger X a_m a_n a_0\leq C_N \sum_{k}|k|^2 a_k^\dagger\left(\sum_{j\neq 0}a_0^\dagger a_j^\dagger X a_j a_0\right)a_k
\end{align}
with the constant $C_N$ being defined as
\begin{align*}
    C_N:=\sup_{m, n \neq0}\left\{ \sum_{j,k\neq 0}\frac{\left|(V_N)_{0jk,0mn}\right|}{|k|^2}\right\}=\frac{1}{N^2}\left\{\sup_{p,q\neq 0}\sum_{t:p+t\neq 0}\frac{\big|V\! \big(N^{-\frac{1}{2}}t\big)\big|}{|p+t|^2}\right\}.
\end{align*}
Due to our regularity assumptions on $V$ we have $\big|V\! \big(N^{-\frac{1}{2}}t\big)\big|\lesssim \frac{1}{1+N^{-1}|t|^2}$ and therefore
\begin{align}
\nonumber
    C_N & \lesssim N^{-2}\sum_{t:p+t\neq 0}\frac{1}{|p+t|^2(1+N^{-1}|t|^2)}\lesssim N^{-2}\int_{\mathbb{R}^3}\frac{\mathrm{d}x}{|p+x|^2(1+N^{-1}|x|^2)}\\
    \label{Eq:Hardy-Littlewood}
    & \leq N^{-2}\int_{\mathbb{R}^3}\frac{\mathrm{d}x}{|x|^2(1+N^{-1}|x|^2)} = N^{-2}4\pi \int_0^\infty \frac{\mathrm{d}r}{1+N^{-1}r^2}=2\pi^2N^{-\frac{3}{2}},
\end{align}
where we have used the Hardy-Littlewood inequality in the first estimate of Eq.~(\ref{Eq:Hardy-Littlewood}). Since
\begin{align*}
    \sum_{j\neq 0}a_0^\dagger a_j^\dagger X a_j a_0\lesssim MN
\end{align*}
 we obtain by Eq.~(\ref{Eq:V_2,2}) 
\begin{align*}
    & \left\langle \Psi, \sum_{jk,mn\neq 0}(V_N)_{0jk,0mn}a_0^\dagger a_k^\dagger a_j^\dagger X a_m a_n a_0 \, \Psi\right\rangle\lesssim N^{-\frac{1}{2}}M\left\langle \Psi, \sum_{k}|k|^2 a_k^\dagger a_k \, \Psi\right\rangle\\
    & \ \ \ \ \ \ \ \leq N^{-\frac{1}{2}}M\left\langle \Psi, H_N \, \Psi\right\rangle\leq N^{-\frac{1}{2}}M\left(E_N+N\right)\lesssim N^{\frac{1}{2}}M,
\end{align*}
where we have used the assumption $\delta\leq N$ and the upper bound on $E_N$ derived in Theorem \ref{Th:First_Order_Upper_Bound}. The only distinguished case left is $I=(1,1,1)$. We start its analysis by defining
\begin{align*}
   \mathbb{V}_{\alpha,\beta}: =\frac{1}{4M^2}\underset{\# I=\alpha, \# J=\beta}{\sum_{I,J\in \{0,1\}^3:}}\sum_{ijk,\ell m n}(\pi^{I_1}\pi^{I_2}\pi^{I_3} V_N \pi^{J_1}\pi^{J_2}\pi^{J_3})_{ijk, \ell m n }a_k^\dagger a_j^\dagger a_i^\dagger X_{I,J} a_\ell a_m a_n,
\end{align*}
which allows us to estimate, using the Cauchy-Schwarz inequality,
\begin{align*}
    \mathbb{V}_{3,3}\leq H_N \! -  \! \left(\mathbb{V}_{2,3} \! + \! \mathbb{V}_{3,2} \! + \! \mathbb{V}_{1,3} \! + \! \mathbb{V}_{3,1} \! + \! \mathbb{V}_{0,3} \! + \! \mathbb{V}_{3,0}\right)\leq H_N \! + \! \frac{1}{2}\mathbb{V}_{3,3} \! + \! 6\left(\mathbb{V}_{0,0} \! + \! \mathbb{V}_{1,1} \! + \! \mathbb{V}_{2,2}\right).
\end{align*}
From the previous cases we know that 
\begin{align*}
 &  \braket{\Psi,(\mathbb{V}_{0,0} \! + \! \mathbb{V}_{1,1} \! + \! \mathbb{V}_{2,2})\Psi}\lesssim N^{\frac{1}{2}}M+N,\\
 &  \ \ \ \ \ \ \braket{\Psi,H_N \Psi}\lesssim N,
\end{align*}
 and therefore $\braket{\Psi,\mathbb{V}_{3,3}\Psi}\lesssim N^{\frac{1}{2}}M+N$. Summarizing what we have so far yields the inequality
\begin{align*}
    m\braket{\Phi,H_N \Phi}+(1-m)\braket{\widetilde \Phi,H_N \widetilde \Phi}\leq E_N+\delta+C \! \left(\frac{\sqrt{N}}{M}+\frac{N}{M^2} \right).
\end{align*}
Using $\braket{\widetilde \Phi,H_NProbability and Mathematical Physics \widetilde \Phi}\geq E_N$ and the simple observation that $m\geq 1-\frac{2\braket{\Psi,\mathcal{N}\Psi}}{M}$ immediately yields Eq.~(\ref{Eq:Energy_After_IMS}), and using $\braket{ \Phi,H_N\Phi}\geq E_N$ we obtain for a suitable constant $C>0$
\begin{align*}
    1-m\leq \frac{C }{\braket{\widetilde{\Phi},H_N \widetilde{\Phi}}-E_N}\! \left(\frac{\sqrt{N}}{M}+\frac{N}{M^2}+\delta \right).
\end{align*}
In order to derive Eq.~(\ref{Eq:Energy_After_IMS_II}), we note that $\mathcal{N}=f\! \left(\frac{\mathcal{N}}{M}\right)\mathcal{N} f\! \left(\frac{\mathcal{N}}{M}\right)+g\! \left(\frac{\mathcal{N}}{M}\right) \mathcal{N} g\! \left(\frac{\mathcal{N}}{M}\right)$ and therefore
\begin{align*}
    \braket{\Psi,\mathcal{E}\Psi}=m\braket{\Phi,\mathcal{N}\Phi}+(1-m)\braket{\widetilde \Phi,\mathcal{N}\widetilde \Phi}\leq \braket{\Phi,\mathcal{N}\Phi}+(1-m)N.
\end{align*}
\end{proof}

\begin{center}
\textsc{Acknowledgments}
\end{center}
Funding from the ERC Advanced Grant ERC-AdG CLaQS, grant agreement n. 834782, is gratefully acknowledged. Furthermore, we would like to thank Marco Caporaletti and Benjamin Schlein for insightful discussions.

\end{document}